\documentclass[11pt,oneside]{article}
\usepackage{graphicx}
\usepackage{amssymb}
\usepackage{amsthm}
\usepackage{amsmath}
\usepackage{tikz}
\usetikzlibrary{shapes}
\usetikzlibrary{patterns}
\usetikzlibrary{arrows,positioning,decorations.pathmorphing,trees}
\usetikzlibrary{arrows.meta}
\usepackage{cancel}
\usepackage{tikz-cd}
\usepackage[ruled, lined, linesnumbered, commentsnumbered, longend]{algorithm2e}
\makeatletter
\tikzset{
  column sep/.code=\def\pgfmatrixcolumnsep{\pgf@matrix@xscale*(#1)},
  row sep/.code   =\def\pgfmatrixrowsep{\pgf@matrix@yscale*(#1)},
  matrix xscale/.code=%
    \pgfmathsetmacro\pgf@matrix@xscale{\pgf@matrix@xscale*(#1)},
  matrix yscale/.code=%
    \pgfmathsetmacro\pgf@matrix@yscale{\pgf@matrix@yscale*(#1)},
  matrix scale/.style={/tikz/matrix xscale={#1},/tikz/matrix yscale={#1}}}
\def\pgf@matrix@xscale{1}
\def\pgf@matrix@yscale{1}
\makeatother
\usepackage{stackengine}
\usepackage{pgfplots}
\usepackage{float}
\usepackage{braket}
\usepackage{authblk}
\usepackage{fullpage}
\usepackage[numbers]{natbib}
\usepackage{hyperref}
\hypersetup{
    colorlinks=true,
    linkcolor=purple,
    citecolor=purple,
    pdftitle="On Uniqueness of BCCE in Single-Item Auctions"
    }

\newtheorem{theorem}{Theorem}[section]

\newtheorem{corollary}[theorem]{Corollary}

\newtheorem{lemma}[theorem]{Lemma}

\newtheorem{proposition}[theorem]{Proposition}
\newtheorem{definition}[theorem]{Definition}

\newtheorem{assumption}{Assumption}

\newtheorem{example}[theorem]{Example}
\newtheorem*{remark*}{Remark}
\newtheorem*{observation*}{Observation}
\newtheorem*{proposition*}{Proposition}
\newtheorem*{corollary*}{Corollary}
\newtheorem*{theorem*}{Theorem}
\newtheorem*{definition*}{Definition}
\newtheorem*{axiom*}{Axiom}
\newtheorem*{claim*}{Claim}
\newtheorem*{lemma*}{Lemma}
\newtheorem{sup-assumption}{Support Assumption}
\newtheorem*{inftheorem*}{Theorem (informal)}
\newtheorem*{expresult*}{Numerical Result}
\newtheorem*{notation*}{Notation}

\title{On the Uniqueness of Bayesian Coarse Correlated Equilibria in Standard First-Price and All-Pay Auctions}
\author[1]{Mete \c{S}eref Ahunbay}
\author[1]{Martin Bichler}
\affil[1]{Technische Universit\"{a}t München, School of Computation, Information and Technology, Department of Computer Science. Boltzmannstraße 3, 85748, Garching bei München, Germany.}
\date{November 17\textsuperscript{th}, 2024\footnote{A first version of this paper appeared on ArXiV on January 2\textsuperscript{nd}, 2024, and a second version on June 20th. The current version incorporates reviewer feedback following acceptance to SODA'25.}}

\begin{document}
\maketitle

\allowdisplaybreaks

\begin{abstract}
  We study the Bayesian coarse correlated equilibrium (BCCE) of continuous and discretised first-price and all-pay auctions under the standard symmetric independent private-values model. Our study is motivated by the question of how the canonical Bayes-Nash equilibrium (BNE) of the auction relates to the outcomes learned by buyers utilising no-regret algorithms. Numerical experiments show that in two buyer first-price auctions the Wasserstein-$2$ distance of buyers' marginal bid distributions decline as $O(1/n)$ in the discretisation size in instances where the prior distribution is concave, whereas all-pay auctions exhibit similar behaviour without prior dependence. To explain this convergence to a near-equilibrium, we study uniqueness of the BCCE of the continuous auction. Our uniqueness results translate to provable convergence of deterministic self-play to a near equilibrium outcome in these auctions. In the all-pay auction, we show that independent of the prior distribution there is a unique BCCE with symmetric, differentiable, and increasing bidding strategies, which is equivalent to the unique strict BNE. In the first-price auction, we need stronger conditions. Either the prior is strictly concave or the learning algorithm has to be restricted to strictly increasing strategies. Without such strong assumptions, no-regret algorithms can end up in low-price pooling strategies. This is important because it proves that in repeated first-price auctions such as in display ad actions, algorithmic collusion cannot be ruled out without further assumptions even if all bidders rely on no-regret algorithms. 
\end{abstract}

\section{Introduction}

Consider the auction of a single item to $N$ buyers, whose valuations are drawn i.i.d. from a distribution $F$ with support on an interval $[0,\max V]$, with the first-price or all-pay auction rule. It is known that the unique equilibrium of these auctions is symmetric, and in pure strategies \cite{CH18}. In this paper, we seek to answer the question:
\begin{center}
    \emph{``Under which conditions is the unique equilibrium of this auction learnable?''}
\end{center}
By ``learnable'', we mean in the context of each buyer using simple no-external regret learning algorithms, with feedback based on their \emph{ex-interim}\footnote{That is, before considering deviations to a bid, they should know their own valuation for the item.} utility, to decide on their bidding strategies. Of course, for the purposes of practical implementations of these auctions, the set of possible valuations and bids are both discrete \& finite sets. Thus, we seek to classify when these no-regret algorithms necessarily result in an outcome \emph{close} to the equilibrium, and bound its distance to the equilibrium with respect to how \emph{fine} the discretisation is. 

The question is one of fundamental importance at the intersection of computer science and economics. Auctions provide the game theoretic framework to model competitive behaviour of agents interacting in a market, and traditional game theoretic analysis often asserts that the outcome should be a (Bayes)-Nash equilibrium (BNE). However, this provides an incomplete perspective of the problem at hand. 

First, this assertion omits the question on \emph{which} equilibrium is reached. In general, the BNE of an auction need not be unique. In this case, we are faced with an \emph{equilibrium selection problem} \cite{HarsanyiSelten1988}; how do we know which of the equilibria will be the result of play? Still, we avoid this issue in the continuous single-item auction settings we consider, as in a variety of single-parameter auctions with symmetric priors, the unique BNE ends up being the canonical symmetric equilibrium \cite{CH18}. However, discretised versions of a continuous auction with symmetric priors need not have a unique equilibrium \cite{RGC20}. How the equilibria of these discretised auctions compare with the equilibrium of the continuous auction, and in general what the expected outcome of such auctions should be is thus an open problem.

Second, real-world markets are often off equilibrium, which raises the questions of \emph{if \& how} an equilibrium is reached. Here, learning theory provides potential answers. The idea is that we may assume agents implement some learning algorithm, and then attempt to obtain a qualitative description or theoretical guarantees on their behaviour. A classical result is that the class of no-regret learning algorithms converges to the so-called \emph{[coarse] correlated equilibrium} ([C]CE) of a game \cite{fudenberg1998theory}. This is a generalisation of the idea of a Nash equilibrium for a game of complete information and can be solved for in time polynomial in the size of the game. 

However, the set of (C)CE of a game is in general larger than the set of its Nash equilibria and learning algorithms need not converge to a Nash equilibrium. This is true even in the fundamental setting of a complete-information first-price auction, whose CCE may significantly differ from its Nash equilibria \cite{FLN16}; and in these auctions, mean-based learning algorithms need not converge to equilibrium \cite{kolumbus2022auctions,deng2022nash}. In the incomplete information setting, Feng et al.~\cite{FGLMS21} show that in an incomplete information first-price auction with symmetric uniform priors, discretised as $\{1/n,...,1\}$, if all buyers implement mean-based learning algorithms \emph{with a lengthy pretraining phase}, then buyers concentrate their bids with high probability $\pm 1/n$ their equilibrium bids. The analyses of \cite{kolumbus2022auctions,deng2022nash,FGLMS21} all depend on dividing the learning process into several epochs, paired with probabilistic analysis of the evolution of bidding strategies. However, for the complete information first-price auction the convergence result is conditional on the number of buyers with highest valuation \cite{deng2022nash}, whereas in the incomplete information setting the probability lower bound depends on the length of the pretraining round \cite{FGLMS21}.

In other words, present day positive convergence to equilibrium results in literature lack generality. This stands in contradiction with the line of empirical work which has observed that learning algorithms converge to equilibrium in a wide range of complete and incomplete information auctions \cite{soda2023,bichler2023convergence}. This leads us to conjecture that there exists a class of learning algorithms that do converge to equilibrium, for which the theoretical framework to prove so is at present unknown. Known sufficient conditions for learnability are the existence of a potential function or a variationally stable Nash equilibrium \cite{MZ19}, which do not necessarily hold in our auction settings \cite{bichler2023convergence}. Thus, the analysis of learning behaviour in even the most fundamental of auction settings necessitates advances in our theoretical approach. If buyers do not necessarily learn to bid only within a small interval of their equilibrium bids, then we require another \emph{``notion of nearness''}. And if a large class of learning algorithms do converge to equilibrium in some sense, then we require the theoretical framework in which they can be proven in full generality.

This motivates our study of learnable outcomes of single-item first-price and all-pay auctions in the standard independent private values setting with symmetric priors. Both positive and negative results have far-reaching implications. 

On the one hand, settings in which convergence is guaranteed imply the existence of properties of the underlying game which ensure that its equilibrium is learnable. If these conditions hold, this would allow for the computation of (approximate) equilibrium outcomes with convergence guarantees in settings for which we don't have analytical solutions so far \cite{fibich2011numerical}. Importantly, we would have guarantees that at least in repeated auctions a broad class of no-external regret algorithms converges to the Bayes-Nash equilibrium. 

On the other hand, if we cannot expect learning agents to converge even in a static repeated auction setting, this raises fundamental questions about the welfare properties of markets such as display ad auctions, which rely on automated and learning agents. Negative convergence to equilibrium proofs for learning algorithms are often accompanied by adverse behaviour from the players; in settings where all agents implement no-regret learning algorithms, agents' history of play can exhibit cyclic and/or chaotic behaviour \cite{MPP18,CP19}. In the context of auctions, this raises the possibility of algorithmic collusion, price cycles, and chaotic price dynamics. This is especially important in the context of automated display ad and sponsored search auctions, where buyers' bids are often assumed to be consistent with automated no-regret learning algorithms \cite{nekipelov2015econometrics}. 

\subsection{Contributions}

Our main contribution is thus a technical framework, through which settings where no-regret learning algorithms converge to near-equilibrium can be identified experimentally, and this convergence can be proven theoretically. We begin our analysis by formulating a linear program (LP) (\ref{opt:generic-primal}) which can be used to bound the distance of a BCCE of a discretised auction from the unique Nash equilibrium of its continuous counterpart, and in fact for games with a unique pure strategy Nash equilibrium (NE) in general. Moreover, if the unique NE is known in closed form, the LP can be evaluated numerically. The primal variables are probabilities buyers bid following a vector of bidding strategies in the BCCE, and the equilibrium constraints are the CCE constraints for the agent-normal form of the game. 

Diverging from previous work which has bound revenue, welfare or buyers' surplus \cite{FLN16,bergemann2017first,pavlov2023correlated}, we first argue intuitively that an appropriate objective of the LP is the Wasserstein-$2$ distance (\ref{def:Wasserstein-objective-general}) of buyers' marginal bid distributions to the equilibrium. This is motivated by the requirement that our distance measure be robust to choice of discretisations; in particular, for settings where two bids $b,b'$ have small difference and the LP has feasible solutions that can assign high probability to either of them. But bids of small difference have about the same distance to the equilibrium bid, so do not change the objective so much. 

For each discretised auction, the value of the LP (\ref{opt:generic-primal}) thus provides the worst-case Wasserstein-$2$ distance of buyers' marginal bid distributions to the Bayes-Nash equilibrium. Thus, if the value of the LP is decreasing in the discretisation size, i.e. the cardinality of the valuation and bidding spaces, then we are able to conclude that average play converges to a near-equilibrium distribution for no-regret learners. Importantly, the finer the discretisation, the closer the outcome to the equilibrium. Else, there exists some no-regret sequence for which average play is bounded away from the equilibrium, and no-regret is too weak a property to ensure convergence to equilibrium. 

This motivates numerical evaluation of the LP for discretisations of continuous auctions. However, for an $N$ player auction with valuation and bidding spaces of size $n$, this LP has $n^{nN}$ primal variables and $1+N n^2$ constraints, which is clearly prohibitive for any reasonable discretisation size. To get around this limitation, we provide a relaxation (\ref{opt:Wasserstein}) which has $O(n^{2N})$ variables and constraints each for a fixed number of buyers $N$, which is polynomial in the discretisation size, and tractable for two buyer auctions with discretisation size $n \leq 21$.

The relaxed problem provides an overestimate of the worst-case distance bounds; vanishing Wasserstein-$2$ distances still imply learnability, else there might exist no-regret algorithms which fail to converge to a near-equilibrium outcome. Our numerical results in Section \ref{sec:numerical-results} suggests that convergence in all-pay auctions can be explained by the no-regret property, while for the first-price auction this is not necessarily the case. In particular, under suitable assumptions on the prior distribution, buyers independently using no-external regret learning algorithms would result in an outcome close to equilibrium; and these distributional assumptions are mild for the all-pay auction.

\begin{expresult*}
For the BCCE of all-pay auctions, the Wasserstein-$2$ distance decreases like $O(1/n)$ in the discretisation size. However, for the first-price auction, decreasing Wasserstein-$2$ distance bounds are conditional only for prior distributions with a concave cumulative distribution function. 
\end{expresult*}

Our theoretical goal is then to develop the framework within which this behaviour can be proven. Dual solutions to each distance bounding primal LP  (\ref{opt:Wasserstein}) for each discretised auction, of known value and feasibility, would prove these vanishing distance bounds. However, optimal dual solutions turn out to be sensitive to the choice of discretisation, and do not appear to admit any discernable closed form solution in the discretisation size $n$ (c.f. Figure \ref{fig:dual}). Our next, technical contribution is then to formalise an \emph{appeal to continuity}; the continuous auction has a unique BNE, which raises the possibility that its unique BCCE is its unique equilibrium. At a high level, we seek to lift the problem to a continuum of valuations and bids; a BCCE for a continuous auction is a correlated distribution over buyers' bidding strategies; i.e. a \emph{measure} over ordered tuples of functions from the set of valuations to the set of bids. And with arguments analogous to those which would prove uniqueness of a BCCE of a game of incomplete information, we may provide a proof of uniqueness of BCCE for the continuous auction via dual solutions to the continuous analogue of the distance bounding LP. We would then interpret a BCCE of the discretised auction as an approximate BCCE for the continuous auction; as a consequence, weak duality allows us to recover distance bounds.

Dealing with such an infinite LP, there several issues whose handling requires care. The concerns here are mostly analytic in nature, related to measurability of functions within the primal and dual LPs, which are needed to ensure that an asserted dual solution is feasible, and weak duality holds between the pair of infinite primal-dual LPs. For feasibility of the dual solution, the primal objective must be upper bounded by the integral of the dual constraints, multiplied by a non-negative function. Our solution is then to use variational calculus to deduce an appropriate primal objective. We observe that feasibility of the dual solution changes under different support\footnote{The set of bidding profiles assigned positive measure.} assumptions for the BCCE, and we iteratively relax the assumptions we work with.

Our first result then turns out to be a uniqueness result for the BCCEs of first-price and all-pay auctions, which at a high level shows that for symmetric BCCE, uniqueness can only be violated if buyers' bidding strategies fail to strictly increase. 

\begin{inftheorem*}[Theorems \ref{thm:thm-fp-uniq} \& \ref{thm:thm-ap-uniq}]
    In the symmetric independent private values (IPV) model we consider, when players' strategies are symmetric, normalised, differentiable, and strictly increasing in their values, the unique BCCE assigns probability $1$ to the canonical equilibrium bidding strategies for both the first-price and all-pay auction.
\end{inftheorem*}

Our variational proof shows that the change in an appropriate weighted sum of deviation profits is precisely the change in a \emph{modified} Wasserstein-$2$ distance. The modification depends on the number of buyers in the auction; and with two buyers, the distance measure is precisely the Wasserstein-$2$ distance extended linearly beyond the maximum equilibrium bid. Thus, our proof of uniqueness \emph{constructs} the Wasserstein-$2$ distance we used in our numerical work from first principles.

Strictly increasing bidding strategies is a standard assumption in economics literature. However, we cannot in general guarantee that buyers do not employ weakly increasing strategies throughout a learning process. Moreover, our proof, i.e. the associated dual solution, fails to demonstrate uniqueness for the first-price auction when the support of the BCCE may include bidding strategies which are weakly increasing. The reason for this is due to buyers potentially increasing their payoffs by pooling their bids; even though such an outcome might not be an equilibrium strategy, it can cause our specific dual solution to become infeasible. We get around this issue by finding an alternate dual solution; this solution reduces the question of uniqueness with support in weakly increasing strategies to uniqueness with support in strictly increasing strategies (Lemma \ref{lem:v-solution}). Second, we still have to show feasibility when the support is restricted to strictly increasing strategies. It is here that we recover the prior dependence that we observe in our numerical experiments.

\begin{inftheorem*}[Theorem \ref{thm:thm-weak-uniq}]
    In the symmetric IPV model we consider, when players' strategies are symmetric, normalised, differentiable, and weakly increasing in their values, the unique BCCE assigns probability $1$ to the canonical equilibrium bidding strategies for the all-pay auction when the inverse prior c.d.f. is strictly increasing and differentiable. This also holds for the first-price auction when the inverse prior c.d.f. is strictly convex (i.e. the prior c.d.f. is strictly concave).
\end{inftheorem*}

We again construct the associated distance measure here, but under some regularity assumptions (c.f. Proposition \ref{prop:wasserstein-prior}) this distance measure to equilibrium is bounded below by a modified Wasserstein-$2$ distance. For first-price auction, unfortunately, the concavity assumption on the prior turns out to be necessary. In particular, ranges of valuations over which buyers' prior is convex raises the possibility that when they pool their bids, their expected payoff is higher than their payoff if they received the item with probability $1$ for a price equal to their equilibrium bid. When the prior distribution is strictly convex, this holds for all buyers when they pool their bids at $0$ for any valuation they may hold. As a consequence of our characterisation of uniqueness for discretised auctions (Propositions \ref{prop:dual-form}), no primal-dual proof of uniqueness of the form we desire can exist when buyers may pool their bids. And rightfully so; we show via examples in Section \ref{sec:examples} that non-concavity of the prior distribution can result in symmetric BCCE of the continuous auction where buyers with high valuations collude by correlating and pooling their bids.

Finally, we demonstrate how to obtain Wasserstein-$2$ distance bounds for the BCCE of discretised auctions via our result of uniqueness of the continuous BCCE. Our idea is that the dual solution for the uniqueness result under the weak monotonicity assumption is still feasible, and that the BCCE of the discretised auction induces an approximate BCCE of the continuous auction. This holds by representing the discretised bidding strategies as a piecewise-constant bidding strategy for the continuous auction. Then, in our \emph{transport lemma} (Lemma \ref{lem:transport}), we show that considering weakly increasing strategies is without loss of generality. The approximation result then follows from weak duality, and the approximation factor then admits a bound which is decreasing in the fineness of the discretisation.

\begin{inftheorem*}[Theorem \ref{thm:main-disc-result}]
    In the discretised symmetric IPV model, whenever 
    \begin{enumerate}
        \item for an all-pay auction, if the prior distribution satisfies some mild assumptions, e.g. positive finite density on $(0,\max V]$ and a bound on its divergence at $0$,
        \item for a first-price auction, if the prior distribution is also strongly concave,
    \end{enumerate}
    then as the discretisation size increases, the modified Wasserstein-$2$ distance of buyers' marginal bid distributions to the canonical equilibrium in a BCCE supported in symmetric strategies tends to $0$. 
\end{inftheorem*}

We remark that our result holds for \underbar{all discretisations} with respect to their fineness, and not only $\{0,1/n,...,1\}$. An immediate consequence is that, if a buyer employs a no-regret algorithm and receives feedback with deterministic self-play (c.f. Algorithm \ref{alg:self-play}), they will learn to play an approximate equilibrium for the all-pay auction with any reasonable prior, e.g.

\begin{corollary*}{(Informal, of Proposition \ref{prop:bounds})}
    Consider a continuous all-pay auction where the buyers' prior distributions are drawn i.i.d. from a power law prior distribution $F(v) = v^\alpha$ for $\alpha > 0$ on $[0,1]$, and its discretisation where the valuation and bid sets are $1/n$-fine. Let a buyer $i$ decide on their bidding strategies in the discretised auction by using the Hedge algorithm to draw pure bidding strategies, updating the weights on the assumption that all other buyers would choose the same bidding strategy. Then the time-average of buyer $i$'s bidding strategies converges to a distribution at most $\tilde{O}(1/n^{\min\{ \alpha, 1\}})$ away from the equilibrium bidding strategies with respect to a modified Wasserstein-$2$ distance. 
\end{corollary*}

The first-price auction does not allow for such a strong result; we require either algorithms that restrict to \textit{strictly} increasing strategies\footnote{We are not aware of any such algorithm.} or a strictly concave prior distribution to get uniqueness of the BCCE and the associated approximation guarantees. On the bright side, our numerical experiments suggest that strict concavity can be weakened to weak concavity. However, a weakly concave prior is too strong an assumption for most realistic environments; it is violated by Gaussian distributions. 

And besides, for weakly increasing strategies with priors exhibiting convexities, we show that there are off-equilibrium pooling strategies at low prices to which no-regret algorithms might converge to in worst-case. This has important implications because first-price auctions are widely used for display-ad auctions. While the very algorithms being used are not known, the use of some no-regret algorithms is a reasonable assumption \citep{nekipelov2015econometrics}. However, we show that when we can \emph{only} assume no-regret learning, but not make additional assumptions on advertisers' bidding agents, we cannot rule out algorithmic collusion. For specific algorithms that allow for additional assumptions (such as restricting search to strictly increasing strategies) the phenomenon might vanish. However, the use of no-regret algorithms only does not rule out algorithmic collusion in first-price auctions, a phenomenon that was reported in experimental work \cite{banchio2023artificial}. This is a subtle, but important difference to all-pay auctions, where our experimental work suggests near unconditional converge to the BNE for no-regret algorithms, and deterministic self-play provably does so. 

\subsection{Related Work}

Auction theory studies the allocation and prices in markets with self-interested participants in equilibrium, via the methods of game theory \cite{Krishna2010}. Auctions are modelled as games of (in)complete information, with continuous type- and action-spaces. Traditionally, the solution concept of an auction game has then been a \emph{(Bayes)-Nash equilibrium} (BNE); a set of bidding strategies for each agent such that no agent may unilaterally deviate and increase their payoff. Vickrey's initial work \cite{vickrey1961counterspeculation}, which arguably was the start of auction theory, studied the symmetric first-price auction setting we study in this paper. Since then, the equilibria of single-item first-price and all-pay auctions have been the topic of literature too extensive to cite compherensively, iteratively relaxing the assumptions (e.g. affiliation of values \cite{milgrom1982atheory}, risk neutrality \cite{maskin1984optimal}) and studying the uniqueness of equilibria \cite{maskin2003uniqueness,CH18}.

It is in general computationally hard to compute an exact or approximate equilibrium in normal-form games. \citet{CD06} and \citet{DGP09} show that finding a Nash equilibrium in a two-player bimatrix game is $PPAD$-complete. The results extend stronger hardness results for the problem of optimising over the set of Nash equilibria \cite{CS08}. Auction games in general are no exception; finding a BNE in simultaneous Bayesian single-item second-price auctions is $PP$-hard when buyers' valuations are submodular \cite{CP14}. A recent line of literature also studies the case of Bayesian single-item first-price auctions we consider, though here negative hardness results for equilibrium computation require additional degrees of freedom in the rules of the auction or informational assumptions. For continuous value distributions and a discretised bidding space, \citet{chen2023complexity} show that with carefully designed tie-breaking rules, finding an $\epsilon$-BNE is $PPAD$-complete for small enough $\epsilon$; the same hardness of approximation is shown by \citet{filosratsikas2021fixp} for subjective priors and a uniform tie-breaking rule, along with $FIXP$-hardness of computing an exact equilibrium. The latter hardness result was recently extended to discretised valuation and bidding spaces, though in the IPV setting we consider with symmetric priors approximate equilibrium computation admits a PTAS \cite{filos2024computation}. Finally, \citet{wang2020bayesian} give an algorithm to compute an exact Nash equilibrium for discretised valuation spaces and continuous bidding spaces, though here the tie-breaking rule favours buyers with higher valuation. 

We remark that the question of computational complexity is distinct from our question on whether no-regret learners converge to bidding strategies nearby an equilibrium; wherein we inspect the convergence guarantees of specific algorithms which are already known to efficiently compute an approximate no-regret outcome.  But for general games, it is well-known that learning algorithms do not converge. The negative results are in fact even stronger; learning dynamics can cycle \cite{PP19,MPP18} and in fact exhibit chaotic behaviour \cite{CP19}. Multiplicative weights update can even approximate arbitrary dynamical systems \cite{AFP21}, and its continuous counterpart (replicator dynamics) is in fact Turing complete \cite{andrade2023turing}. There are even games where every learning dynamic has an initialisation which does not converge to an equilibrium \cite{MPPS22}. 

This doesn't mean that learning algorithms do not converge in specific game types, and a recent line of literature shows that a variety of no-regret or deep learning algorithms converge to an equilibrium in a \emph{broader} class of auction games than the one we consider in this paper, including single-item auctions with asymmetric priors and risk aversion, combinatorial auction settings, and double auctions \cite{BBLS17,BFHKS21,soda2023}. However, a theoretical justification for why such convergence is observed has proven elusive. \citet{wang2023noregret} point out that existing literature has almost exclusively focused on settings in the presence of a potential function (e.g. \cite{blum2010routing,kleinberg2009multiplicative,heliou2017learning,MZ19}) or monotonicity of the utility gradients (e.g. \cite{MZ19,tatarenko2019learning,golowich2020tight,cai2022finite}) for the convergence of \emph{specific} learning algorithms which have the no-regret property. \citet{wang2023noregret} also define a notion of strong monotonicity to obtain a sufficient condition for general no-regret algorithms to converge. However, the auction settings we consider in this paper are neither known to be monotone nor to admit a potential function. In fact, the first-price auction is not monotone in the alternate learning setting where at each time period, buyers make small modifications to the their bidding strategies \cite{bichler2023convergence}. 

Positive convergence results in literature do exist for complete-information models of auctions, with respect to \emph{mean-based} learning algorithms, a strengthening of the no-regret property proposed by \citet{braverman2018selling}. \citet{kolumbus2022auctions} show that \emph{if} the history of play for agents employing mean-based learning algorithms converges to some distribution, then that distribution is a \emph{co-undominated} CCE, and such CCE are necessarily close to equilibrium. Further work by \citet{deng2022nash} show that the convergence of mean-based learning algorithms to the Nash equilibrium then depends on the number of buyers who have highest valuation for the item. Closer to our setting and our results, for incomplete information auctions, is the work of \citet{FGLMS21}, which shows that buyers using mean-based learning algorithms in a discretised two buyer first-price auction will concentrate their bids about the equilibrium of the continuous auction. Two major caveats apply; the result holds only for valuation and bidding sets $\{1/n,2/n,...,1\}$, and is contingent on a lengthy pretraining period where each buyer bids uniformly at random.

Results are more limited within the framework of no-regret learning in auctions, equivalently, the (C)CE of our auctions. For the complete information case, \citet{dutting2014mechanism} prove the uniqueness of CE in a variety of auction settings, which includes single-item first-price auctions. \citet{FLN16} show that any CE of a first-price auction is a mixture of its Nash equilibria, hence efficient; but there are CCE with suboptimal welfare and revenue. \citet{einy2022strong} identify uniqueness of correlated equilibrium as the \emph{strong robustness} of a Nash equilibrium to incomplete information; this implies the uniqueness of CE of two player all-pay auctions.

In the incomplete information setting, we remark that any discussion of \emph{Bayesian} (coarse) correlated equilibria of auctions necessitates what such a notion should be, as various proposals exist in literature; including communication equilibrium \cite{myerson1982optimal,myerson1986multistage}, Bayes correlated equilibrium \cite{bergemann2016bce} (also studied in \cite{makris2023information,zhang2022polynomial}), and the multiple proposals of Forges \cite{forges1993five,forges2023correlated}. The notion we adopt in this paper is the coarse variant of a correlated equilibrium for the \emph{agent-normal form} of the game, previously considered by \citet{HST15} following one of the five notions of a BCE discussed by \citet{forges1993five}. We adopt this notion, as no-regret learners with ex-interim utility feedback converge to such an equilibrium. In particular, in the other notions of equilibrium, the role of the mediator is central; however, in our setting, there is no such mediator and any correlation of buyers' bids arises from the dynamics of the learning process itself. We refer the interested reader to \citet{Fujii2023} for a discussion of learning algorithms which compute a B(C)CE in a variety of its aforementioned definitions.

With respect to our auction settings, there have been few studies of the B(C)CE of auctions, with respect to any of its definitions. 
\citet{lopomo2011lp} undertake an LP based approach similar to ours in spirit, bounding the maximum buyer surplus in a variant of communication equilibrium where transfers between buyers are allowed. 
\citet{HST15} extend the smoothness framework of \citet{roughgarden2015intrinsic} to the (agent-normal form) BCCE of general games, from which welfare bounds for incomplete information auctions follow immediately. \citet{bergemann2017first} study the Bayes correlated equilibria of the first-price auction under common priors and provide best- and worst-case bounds for revenue and bidder surplus for any such information structure. The result implies that Bayes correlated equilibria of first-price auctions is in general not unique; a similar conclusion can be drawn for the communication equilibria of all-pay auctions via the results of \citet{pavlov2023correlated}. While the approach of Bergemann et al. \cite{bergemann2017first} is in fact similar to our infinite LP framework, we remark that their proof methods do not extend to the notion of BCCE we consider in this paper. This is because the sets of Bayes CE and BCCE do not in general contain each other; agent-normal form Bayesian \emph{coarse} correlated equilibria have weaker equilibrium constraints, while Bayes CE and communication equilibria can allow for more general type-action correlations. We mention how precisely this is the case in Section \ref{sec:prelim-games-equilibria}.

\subsection{Organisation of the Paper}

We begin in Section \ref{sec:prelims} by introducing the game theoretic definitions, especially notions of (Bayesian / coarse) correlated equilibria, and the discretised or continuous symmetric single-item auction models we use throughout the paper. We also discuss the connections of our equilibrium notions to no-regret learning. We proceed in Section \ref{sec:computation} with our initial study of discretised auction settings. In Section \ref{sec:distance-bounding-lps} we introduce our LP based framework for certifying when all Bayesian (coarse) correlated equilibria [B(C)CE] of a discretised auction are \emph{``near''} its unique pure strategy Nash equilibrium, or that of its continuous counterpart. The resulting LP can be exponential in the size of the discretised valuation and bid sets; to enable numerical experiments, in Section \ref{sec:relaxations} we consider a relaxation of the distance bounding primal LP which is amenable to numerical experiments. The results are presented in Section \ref{sec:numerical-results}, which show the prior dependence of convergence for first-price auctions, and apparent unconditional convergence for all-pay auctions.

Our numerical results raise the question of whether we may certify the value of these distance bounding LP via a family of solutions for its dual, which the focus of the rest of the paper. First in Section \ref{sec:structure} we momentarily conclude our discussion of discretised auctions, by investigating the form of dual solutions which certify structural properties of bidding strategies, such as monotonicity or uniqueness. These dual solutions require uniqueness of the primal solution, which need not necessarily hold for discretised auctions. This motivates our study of BCCE of \emph{continuous} first-price and all-pay auctions in Section \ref{sec:theory-cont}; to improve readability, we defer the analogous proofs for the all-pay auction to the appendix. In Section \ref{sec:formalism} we introduce our infinite LP based framework for certifying the uniqueness of BCCE in continuous auctions; measure theoretic considerations necessitate care in the specification in the set of its \emph{supporting strategies}, which we iteratively relax. Section \ref{sec:uniqueness-strict} then contains our first uniqueness result; for the first-price and all-pay auctions with symmetric priors, the canonical equilibrium \underbar{only} BCCE with support symmetric, strictly increasing, differentiable, normalised, and bounded bidding strategies. In Section \ref{sec:weak-uniqueness} we weaken our strict monotonicity assumption to weak monotonicity, but at the cost of the dependence on concave priors for first-price auctions we had observed in our numerical results. Section \ref{sec:examples}, we show that this assumption is necessary, by presenting examples of BCCE with support in non-equilibrium strategies for first-price auctions with non-concave priors.

Finally, in Section \ref{sec:cont-to-disc}, we show how the results on uniqueness of BCCE of a continuous auction can be translated to \emph{near}-uniqueness of the BCCE of its discretisations; as long as the discretised BCCE has support in symmetric strategies. In Section \ref{sec:disc-to-approx-cont-BCCE} we discuss how BCCE of a discretised auction induces an approximate BCCE of its continuous counterpart, and modify our infinite LP framework to account for the discretised strategies. Section \ref{sec:bounds-theory} then contains the machinery, which allows us to translate the proof of uniqueness for the weakly monotone setting to proofs of bounds for the distance between the BCCE of the discretised auction and the canonical equilibrium.  

\section{Preliminaries}\label{sec:prelims}

In this section, we provide an overview of the setting of our paper. In Section \ref{sec:prelim-games-equilibria}, we recall the notion of a normal form game of (in)complete information, the definitions of (Bayes-)Nash equilibria, and for finite normal form games of complete information, the notions of (coarse) correlated equilibria. In Section \ref{sec:learning}, we present the learning model we consider, and the notion of a \emph{Bayesian} (coarse) correlated equilibrium appropriate for its convergence analysis. Finally, we discuss the continuous and discrete models of single-item first-price and all-pay auctions we analyse. All throughout, we restrict attention to the \emph{independent private values} (IPV) model; a player's utility is determined by the vector of actions of all players, plus \emph{only} their own type / valuation for the item. 

\subsection{Games and Notions of Equilibria}\label{sec:prelim-games-equilibria}

\begin{notation*}
  In what follows, $\mathbb{N}$ is the set of counting numbers\footnote{i.e. we assume $0 \notin \mathbb{N}$.}, $\mathbb{R}$ is the set of real numbers, and $\mathbb{R}_+$ is the set of non-negative reals. For any pair of sets $S, T$, we denote by $\Delta(S)$ the set of probability distributions over $S$, and by $S^T$ the set of functions $T \rightarrow S$. Whenever $S$ is finite, we identify each $s \in \Delta(S)$ as an element of $\mathbb{R}_+^{|S|}$ whose entries sum up to $1$. Moreover, for any finite sequence $(s_\ell \in \mathbb{R}^{S_\ell})_{\ell \in \Lambda}$ of vectors indexed by $\Lambda$, we denote by $\otimes_{\ell \in \Lambda} s_\ell$ their outer product. Finally, for $s, s' \in S$, the indicator function $\mathbb{I}[s = s']$ evaluates to $1$ if $s = s'$ and $0$ otherwise. Due to concerns of length and readability, we will sometimes write the indicator function via the Kronecker delta, $\mathbb{I}[s = s'] = \delta_{s s'}$.   
\end{notation*}

\begin{definition}
  A \textbf{normal-form game} $\Gamma$ is specified as a tuple $\Gamma = \left(N,\Theta,F^N,A,u\right)$, where $N \in \mathbb{N}$ is the number of players, and we denote by $[N] = \{n \in \mathbb{N} | n \leq N\}$ the set of players. The set of vectors of types of players is $\Theta = \times_{i \in [N]} \Theta_i$, and $F^N \in \Delta(\Theta)$ is the prior distribution over players' types. Meanwhile, the set of vectors of actions of players is $A = \times_{i \in [N]} A_i$, and we call a pair $(a,\theta) \in A \times \Theta$ an outcome of the game. Players' utilities $u = \times_{i \in [N]} u_i$ are then functions $u_i : A \times \Theta_i \rightarrow \mathbb{R}$, mapping an outcome of the game to a payoff attained by a player.
  If $\Theta,A$ are both finite sets, then $\Gamma$ is called a \textbf{finite game}. If $|\Theta| = 1$ ($> 1$) then we say that $\Gamma$ is a game of \textbf{(in)complete information}. If $\Gamma$ is a game of complete information then we shall write $\Gamma = (N,A,u)$ to simplify notation.
\end{definition}

Games model settings in which agents interact strategically, seeking to maximise their payoff from an outcome. Thus analysis of games are meaningful with a concept of \emph{equilibrium}, a \emph{stable} outcome, specified via strategies $s_i : \Theta_i \rightarrow \Delta(A_i)$ for each player $i \in [N]$ such that no player $i$ of type $\theta_i$ has any incentive to choose a probability distribution other than $s_i(\theta_i)$, given the information available to them. A common benchmark is a so-called (Bayes)-Nash equilibrium.

\newcommand{\BNE}{\textnormal{BNE}}
\newcommand{\CE}{\textnormal{CE}}
\newcommand{\CCE}{\textnormal{CCE}}

\begin{definition}
    A \textbf{Bayes-Nash equilibrium} (BNE) of a game $\Gamma$ is a tuple of players' strategies $s$ such that for any player $i$, any type $\theta_i \in \Theta_i$ and any action $a'_i \in A_i$ of player $i$,
    \begin{equation}\mathbb{E}_{\theta_{-i} \sim F_{-i}, a \sim s}[u_i(a|\theta_i)] \geq \mathbb{E}_{\theta_{-i} \sim F_{-i}, a_{-i} \sim s_{-i}}[u_i(a'_i,a_{-i}|\theta_i)],\end{equation}
    where, following standard notation, $\theta_{-i} = (\theta_{j})_{j \in I \setminus \{i\}}, a_{-i} = (a_{j})_{j \in I \setminus \{i\}}$ and $s_{-i} = (s_{j})_{j \in I \setminus \{i\}}$ denote the vector of the respective term for every player except $i$, and $F_{-i}$ is the joint probability distribution of other players' types. For a game $\Gamma$, we denote by $\BNE(\Gamma)$ the set of its BNE. 
\end{definition}

The relaxation of the assumption of independent play leads to the (often strictly greater) sets of \emph{correlated} and \emph{coarse correlated} equilibria. For finite games of complete information, there is consensus in the literature on their definition, following the work of \citet{Aumann74}:
\begin{definition}\label{def:CE}
    A \textbf{correlated equilibrium} of a game $\Gamma$ of complete information is a distribution $\sigma \in \Delta(A)$ such that for any player $i$, and any pair of actions $a_i, a'_i$ of player $i$, 
    \begin{equation}\label{def:CE-CI}
        \mathbb{E}_{a_{-i} \sim \sigma_{-i}|_{a_i}} \left[ u_i(a) \right] \geq \mathbb{E}_{a_{-i} \sim \sigma_{-i}|_{a_i}} \left[ u_i(a_i', a_{-i}) \right].
    \end{equation}
    Here, $\sigma_{-i}|_{a_i}$ is the joint distribution of actions of players other than $i$, drawn from $\sigma$ conditional on player $i$'s action drawn to be $a_i$. In turn, a \textbf{coarse correlated equilibrium} is a distribution $\sigma \in \Delta(A)$ such that for any player $i$ and any action $a_i$ of player $i$, 
    \begin{equation}\label{def:CCE-CI}
        \mathbb{E}_{a \sim \sigma} \left[ u_i(a) \right] \geq \mathbb{E}_{a \sim \sigma} \left[ u_i(a_i', a_{-i}) \right].
    \end{equation}
    Here, a player can instead unilaterally commit to a fixed action $a_i'$, but may not exchange only the occurances of some action $a_i$ with $a'_i$ as the deviations considered in the case of correlated equilibria. For a game $\Gamma$, we denote by $\CE(\Gamma)$ $[\CCE(\Gamma)]$ the set its [C]CE. 
\end{definition}
As the inequalities (\ref{def:CCE-CI}) are a conical combination of the inequalities (\ref{def:CE-CI}), we have $\CCE(\Gamma) \supseteq \CE(\Gamma)$ for any game $\Gamma$. The inclusion $\CE(\Gamma) \supseteq \BNE(\Gamma)$ also holds, if we interpret $\BNE(\Gamma)$ as a subset of $\Delta(A)$ via the injection mapping a mixed strategy profile $s$ to the corresponding distribution induced by $\otimes_{i \in [N]} s_i$ over $A$; i.e. if $s \in \BNE(\Gamma)$ and  for any outcome $a \in A$ we let $\sigma(a) = \prod_{i \in N} s_i(a_i)$, then $\sigma \in \CE(\Gamma)$.

\subsection{Learning Model \& Bayesian (Coarse) Correlated Equilibria}\label{sec:learning}

One application of correlated equilibria is in modelling settings where agents are recommended actions via a mediator; in this case, at equilibrium they do not have any incentive to not follow their recommended actions. While for finite normal form games of complete information, both the definitions of CE and CCE are agreed upon, what their definition \emph{should be} in normal-form games of incomplete information has been the topic of extensive discussion, with different approaches discussed in \citep{aumann1987correlated,myerson1986multistage,forges1993five,CP14,HST15,bergemann2016bce,forges2023correlated}.

The notion of Bayesian (C)CE we will use in this paper is motivated by \emph{no-regret}\footnote{Also referred to as, perhaps more accurately, \emph{``vanishing regret''}.} learning algorithms~\citep{cesa-bianchiPredictionLearningGames2006} in the Bayesian setting. Specifically, we will consider distributions of play which arise when players receive feedback on their ex-interim payoffs, and choose their strategies independently of each other at each time period. The latter assumption will imply that any correlation between \emph{type realisations} and players' strategies would arise not due to the presence of a mediator, but due to players learning to do so organically through the \emph{mechanics} of the learning process, e.g. due to cyclic learning behaviour. However, our assumptions on the feedback rule out the possibility of such type-strategy correlations, implying convergence to a (C)CE for the \emph{agent-normal form} of the game, previously considered in \citep{forges1993five,HST15}.

\begin{definition}\label{def:learning-discrete}
  A learning process is simply given by its sequences in time $t \in \mathbb{N}$ for its \textbf{history of play}, i.e. strategies $s_i^t : \Theta_i \rightarrow \Delta(A_i)$ for each player $i$, and the \textbf{ex-interim utility feedback} to each player $i$, $u_i^t : A_i \times \Theta_i \rightarrow \mathbb{R}$. Such a learning process is said to satisfy \textbf{no internal} (or \textbf{swap}) \textbf{regret} if for any player $i$, and any type $\theta_i$ and any pair of actions $a_i, a'_i$ of player $i$,
  $$ \lim_{t \rightarrow \infty} \frac{1}{t} \sum_{\tau=1}^t s_i^t(\theta_i,a_i) \cdot (u_i^t(a_i|\theta_i) - u_i^t(a'_i|\theta_i)) \geq 0.$$
  In turn, a learning process is said to satisfy \textbf{no external regret} if for any player $i$, and any type $\theta_i$ and any action $a'_i$ of player $i$,
  $$ \lim_{t \rightarrow \infty} \frac{1}{t} \sum_{\tau=1}^t \sum_{a_i \in A_i} s_i^t(\theta_i,a_i) \cdot (u_i^t(a_i|\theta_i) - u_i^t(a'_i|\theta_i)) \geq 0.$$
  We then refer to $\sigma^T = \frac{1}{T} \sum_{\tau = 1}^{T} \otimes_{i \in [N]} s_i^t$, as the \textbf{aggregate history of play} at time $T$.
\end{definition}

\begin{example}\label{ex:full-feedback}
  In the \emph{full-feedback setting}, players observe their exact ex-interim utility for each action, 
  $u_i^t(a_i|\theta_i) = \sum_{\theta_{-i} \in \Theta_{-i},a_{-i} \in A_{-i}} F_{-i}(\theta_{-i}) \cdot \left(\prod_{j \in [N] \setminus \{i\}} s_j(\theta_j,a_j) \right) \cdot u_i(a|\theta_i)$. One example is the dual averaging based algorithms considered in \citep{soda2023} for learning bidding strategies in auctions, which satisfies no external regret \cite{Xiao10}.
\end{example}

\begin{example}\label{ex:self-play}
  With \emph{deterministic self-play}, we consider the problem of a single player learning to play in an $N$-player symmetric game. The player chooses a pure strategy at each time period, and receives feedback based on their ex-interim utility assuming the $N-1$ other players play the same strategies against them. The determinism here refers to the strategies chosen by the player being pure, and not to how they are chosen. Namely, we draw $s_i^t: \Theta_i \times A_i \rightarrow \{0,1\}$ for some fixed player $i$ in an appropriate manner that ensures the no-regret property, and further assume that $s^t_j = s_i^t$ for every player $j$. The utility feedback is then set as in the full-feedback setting.
\end{example}

\begin{example}\label{ex:population}
  In the \emph{population interpretation} of the Bayesian game, \citet{HST15} consider deterministic strategies $s_i^t: \Theta_i \times A_i \rightarrow \{0,1\}$, and types $\theta^t$ at each time $t$ drawn independently from the distribution $F^N$. The utility feedback to players is then fixed, $u^t_i(a_i|\theta_i) = \mathbb{I}[\theta_i = \theta_i^t] \cdot u_i(s(\theta)|\theta_i)$.
\end{example}

For a game of complete information, a folklore result is that for no internal (external) regret algorithms, any convergent subsequence of the aggregate history of play does so to a (C)CE of the game. The key observation then is that, for the learning processes considered in Examples \ref{ex:full-feedback}, \ref{ex:self-play}, and \ref{ex:population}, the aggregate history of play converges to that of a (C)CE of the agent-normal form of the game; for the former two, the result is immediate from evaluating the limit, while the latter was proven by \citet{HST15} (Lemma 10). 

\begin{definition}
  For a game $\Gamma = \left(N,\Theta,F,A,u\right)$ of incomplete information, the \textbf{agent-normal form} of $\Gamma$ is the game $\Gamma' = \left(J,A',u'\right)$ of complete information, with a set of players $J = \sqcup_{i \in [N]} \Theta_i$, the \emph{disjoint union} of agents' type sets. Each player $(i,\theta_i)$ has an action set $A_i$, and the set of vectors of actions of this game are given $A' = \times_{(i,\theta_i) \in J} A_i$. An outcome of this game is then an element $a'$ of $A'$; by the isomorphism $\times_{(i,\theta_i) \in J} A_i \cong \times_{i \in [N]} A_i^{\Theta_i}$ an outcome of $\Gamma'$ is a full specification of players' pure strategies in the incomplete information game $\Gamma$. Finally, players utilities $u_{(i,\theta_i)} : A' \rightarrow \mathbb{R}$ are defined via considering their expected utilities in $\Gamma$, $ u_{(i,\theta_i)}(a') = \mathbb{E}_{\theta_{-i} \sim F_{-i|\theta_i}}[ u_i(a(\theta)|\theta)]$.
\end{definition} 

Intuitively, in the agent-normal form of the game, we encode as a distinct player each type realisation of each player in the original game. Each player/type pair has, as their set of actions, the action set of the original player. An outcome of the agent-normal form game is thus an action chosen for each player of each type, which is equivalent to choosing a strategy for each player in the original game. Utilities are computed by expectation, where given a player/type pair, other players' types are drawn from the (in our setting, common and independent) prior distribution. 

\begin{definition}\label{def:BCE}
  A \textbf{Bayesian correlated equilibrium} of a game $\Gamma$ of incomplete information is a distribution $\sigma \in \Delta(\times_{i \in [N]} A_i^{\Theta_i})$ over the set of \emph{strategy profiles}, such that for any player $i$, any type $\theta_i$ of player $i$, and any pair of actions $a_i, a'_i$ of player $i$,
  \begin{equation}\label{def:CE-II}
      \mathbb{E}_{s \sim \sigma|_{s_i(\theta_i) = a_i}, \theta_{-i} \sim F_{-i}}[u_i(a_i, s_{-i}(\theta_{-i})|\theta_i)] \geq \mathbb{E}_{s \sim \sigma|_{s_i(\theta_i) = a_i}, \theta_{-i} \sim F_{-i}}[u_i(a'_i, s_{-i}(\theta_{-i})|\theta_i)].
  \end{equation}
  Here, $\sigma|_{s_i(\theta_i) = a_i}$ is the joint distribution of strategies of players other than $i$, conditional on player $i$ of type $\theta_i$ has their action drawn to be $a_i$. In turn, a \textbf{Bayesian coarse correlated equilibrium} is a distribution $\sigma \in \Delta(\times_{i \in [N]} A_i^{\Theta_i})$ such that for any player $i$, any type $\theta_i$ of player $i$, and any action $a'_i$ of player $i$,
  \begin{equation}\label{def:CCE-II}
      \mathbb{E}_{s \sim \sigma, \theta_{-i} \sim F_{-i}}[u_i(s(\theta)|\theta_i)] \geq \mathbb{E}_{s \sim \sigma, \theta_{-i} \sim F_{-i}}[u_i(a'_i, s_{-i}(\theta_{-i})|\theta_i)].
  \end{equation}
  Here, a player can unilaterally commit to a fixed action $a'_i$ \emph{after learning their type}, but as in the CCE of games of complete information, may not exchange only the occurrences of some $(\theta_i,a_i)$ with $(\theta_i,a'_i)$. The set of B[C]CE of $\Gamma$ is then again denoted $\CE(\Gamma)$ $[\CCE(\Gamma)]$.
\end{definition}

The key difference between the notion of Bayesian (C)CE we consider and the notions of communication equilibria or Bayes correlated equilibria discussed in \citep{forges1993five,bergemann2016bce,forges2023correlated} is the assumption that $\sigma \in \Delta(\times_{i\in N} A_i^{\Theta_i})$. Indeed, these alternative notions of Bayesian CE only assume that $\sigma \in \Delta(A)^\Theta$, which in general is a strict superset of $\Delta(\times_{i\in N} A_i^{\Theta_i})$. This is referred to as \emph{strategy representability} in \citep{Fujii2023}, which we refer the interested reader to for an extended discussion. 

As our eventual primal-dual formalism will result in guarantees for no-regret algorithms coupled with deterministic self-play, we find it illustrative to conclude this section with an explicit algorithm which uses such feedback paired with the Hedge algorithm \cite{arora2012multiplicative}. Denoting $\bar{U} = \max_{\theta_1, a} | u_1(a|\theta_1) |$, the algorithm is given:

\begin{algorithm}[H]
\caption{Deterministic self-play with Hedge in a symmetric game}\label{alg:self-play}
\SetKwComment{Comment}{/* }{ */}
\SetKwInput{Input}{Input}
\SetKwInput{Output}{Output}
\Input{Symmetric game $(N, \Theta, F, A, u)$, time length $T = 2^\tau$ for integer $\tau$}
\Output{$\sigma_T$, an $O(\sqrt{1/2^\tau})$ approximate BCCE}
$t \gets 1$\;
$w^1(\theta_1,a_1) \gets 0 \ \forall \ \theta_1 \in \Theta_1, a_1 \in A_1$\;
\While{$t \leq T$}{
    Draw strategies $s_1^t : \Theta_1 \times A_1 \rightarrow \{ 0,1 \}$, $s_1^t(\theta_1) = a_1$ w.p. $\sim e^{w^t(\theta_1,a_1)}$\;
    Fix utility feedback, $u_1^t(a_1|\theta_1) = \mathbb{E}_{\theta_{-1} \sim F_{-1}}[u_1(a_1,s_{-i}^t(\theta_{-1})|\theta_1)]$.\;
    Update weights, $w^{t+1}(\theta_1,a_1) = w^t(\theta_1,a_1) + u_1^t(a_1|\theta_1) \cdot \sqrt{\frac{ \ln{|A_i|}}{T \bar{U}^2}}$\;
    $t \gets t+1$ \;
} 
$\sigma^T(\theta,a) \gets \frac{1}{T} \sum_{1 \leq t \leq T} \prod_{j \in [N]} s_1^T(\theta_j,a_j)$. 
\end{algorithm}

Specifically, a $T = 2^\tau$ round implementation of the algorithm then ensures that the algorithm has regret at most $2\bar{U} \sqrt{ \ln |A_1| / T}$. The standard trick to ensure vanishing regret is to then bootstrap Algorithm \ref{alg:self-play}; we run Algorithm \ref{alg:self-play} for $T = 2^\tau$ for every $0 \leq \tau < \tau'$ and consider the overall history of play. Whereas the regret guarantee of Hedge is in expectation in general \cite{arora2012multiplicative}, in the setting of a normal form game the regret of all players and types vanishes almost surely as $T \rightarrow \infty$ for the resulting empirical distribution of play \cite{greenwald2003general}. As a consequence, any convergent subsequence of the history of play converges to a BCCE, supported only in symmetric pure strategies.

\subsection{Continuous Single-Item Auctions and their Discretisations}\label{sec:auction-models}

In this section, we introduce the continuous and discrete single-item auction models we use in this paper, and discuss known results on the (C)CE of complete and incomplete auction games. We focus on the standard symmetric setting where buyers are risk-neutral and have independent private valuations drawn i.i.d. from some distribution $F$ and quasi-linear utilities.

\newcommand{\valset}{\mathcal{V}}
\newcommand{\bidset}{\mathcal{B}}
\newcommand{\profset}{\Lambda}

\begin{definition}
    The data of a \textbf{single-item auction} is specified by a tuple $A = \left( N, \valset, F^N, \bidset, u \right)$, where $N$ is the number of buyers. The set of vectors of valuations buyers may have for the item for sale form the set of vectors of types, and equals $\valset = \times_{i \in [N]} V$ for some $V \subseteq \mathbb{R}$. For some distribution $F \in \Delta(V)$, the prior distribution over buyers' valuations is the product distribution $F^N \in \Delta(\valset)$. The set of allowed bids is $B \subseteq \mathbb{R}$, and the set of vectors of actions of players are vectors of their bids $\bidset$. Moreover, we denote the set of pure bidding profiles, $\profset \equiv \times_{i \in [N]} B^V$. The allocation rule $x : \bidset \rightarrow \Delta([N])$ and the payment rule $p : \bidset \rightarrow \mathbb{R}^N$ specify respectively the probability a buyer wins an item and their expected payment as a function of buyers' bids. Then, given valuations $v \in \valset$ and bids $b \in \bidset$, the utility of buyer $i$ is given
    $ u_i(b|v_i) = x_i(b) \cdot v_i - p_i(b).$ The auction is said to be \textbf{continuous} if $V \subseteq B$ are closed intervals in $\mathbb{R}_+$ containing $0$, and $F$ has a probability distribution function on $V$ with full support. If instead $|V|, |B| < \infty$, the auction is called \textbf{discrete}.
\end{definition}

Thus all we need to determine for our auctions are the allocation and payment rules. In this paper, we focus our attention on single-item first-price and all-pay auctions. The continuous versions of the auctions we consider are extensively studied in the literature (cf. for example, \cite{Krishna2010}), so our exposition of their allocation and payment rules along with their canonical equilibria will be brief.

These auctions award the item to the highest bidder; moreover, with continuous valuation and bid spaces, the tie-breaking rule does not matter \emph{for the canonical symmetric equilibrium}, as they occur with probability zero. However, this ceases to be the case when anti-symmetric or mixed-equilibria need to be ruled out, and we remark that both \citet{CH18} and \citet{RGC20} point to the importance of the choice of tie-breaking rule. In the continuous case, \citet{CH18} show that allocation probability strictly decreasing upon a tie (\emph{``win-vs-tie-strict''}) is a sufficient condition to rule out asymmetric mixed Bayes-Nash equilibria in the continuous auctions, when types are drawn from a distribution with continuous, bounded support. In turn, \citet{RGC20} illustrate that in a discretised auction setting, the choice of a tie-breaking rule has an effect on the number of equilibria and their functional form. 

In our setting, we shall consider a uniform tie-breaking rule; the allocation rule will be defined 
$$x_i(b) = \begin{cases}
  \frac{1}{| \arg\max_{j \in [N]} b_j |} & i \in \arg\max_{[N] \in I} b_j, \\
  0 & \textnormal{else.} 
\end{cases}$$
We remark that we have a simpler reason to restrict our attention so. For the first-price auction, awarding an item to no player upon ties forms a symmetric equilibrium where every buyer bids $\max_{b \in B} b$, while awarding an item to all players upon a tie forms a Bayes-Nash equilibrium (and hence, a CCE) where every buyer bids $0$. Our linear programming based approach will then be unable to rule out such equilibria without imposing additional linear constraints. We find such an approach unnatural.

The auctions we consider then differ in the payment rule. In the first-price auction, the winning buyer makes a payment equal to their own bid. Thus the expected payment of a buyer, given bids $b$, is $p_i(b) = x_i(b) \cdot b_i$. The canonical symmetric equilibrium of the first-price auction is for a buyer with valuation $v_i$ to bid the expectation of the highest valuation $v'$ amongst other buyers, conditional on $v' \leq v_i$,
$$ \beta^{FP}(v_i) = \mathbb{E}_{v_{-i} \sim F^{(N-1)}}[\max_{j \neq i} v_j | \ \forall \ j \neq i, v_j \leq v_i].$$

In turn, in an all-pay auction, every bidder makes a payment equal to their own bid, no matter whether they are allocated the item or not. Thus in this case, we have $p_i(b) = b_i$. The canonical equilibrium of the all-pay auction scales the first-price auction bids with the cumulative distribution function of the first-order statistic of other buyers' valuations, or 
$$ \beta^{AP}(v_i) = F(v_i)^{n-1} \cdot \mathbb{E}_{v_{-i} \sim F^{(N-1)}}[\max_{j \neq i} v_j | \ \forall \ j \neq i, v_j \leq v_i].$$

One of our goals will be to show that the B(C)CE of discretised first-price and all-pay auctions are \emph{``close''} to the canonical equilibria of their continuous counterparts. This will necessitate a formal notion of parametrised discretisations of a continuous auction model.

\begin{definition}\label{def:discretisations}
  For a continuous single-item auction $A$ and a decreasing function $\delta : \mathbb{N} \rightarrow \mathbb{R}_+$ such that $\lim_{n\rightarrow \infty} \delta(n) = 0$, a $\delta$-\textbf{fine family of discretisations parametrised by} $n \in \mathbb{N}$ is a sequence $(V_n,B_n,F_n)_{n \in \mathbb{N}}$ of valuation spaces $V_n \subseteq V$, bidding spaces $B_n \subseteq B$ and prior distributions $F_n$ over $V_n$ which determine a sequence of auction games $A_n$ such that $F_n \rightarrow F$ in probability, and 
$$
    \sup_{v \in V} \inf_{v' \in V_n} |v-v'|, \sup_{b \in B} \inf_{b' \in B_n} |b-b'| \leq \delta(n).$$
\end{definition}

To conclude this section, we remark that the set of CCE of single-item auctions need not equal the set of its Nash equilibria, even in the setting with complete information. For instance, \citet{FLN16} show that the complete information first-price auction with two buyers, when each buyer has value $1$, there is a CCE\footnote{We remark that this is a CCE for an auction of complete information, with an infinite set of feasible bids. With an infinite action space, existence theorems for Nash or (coarse) correlated equilibria break down; but in their setting, existence is assured \emph{by construction}. \citet{FLN16} also display CCE with discretised bidding spaces which have ``similar form'' to the continuous CCE they construct, which already hints at connections between BCCE of discretised auctions and the notion of continuous BCCE we consider in Section \ref{sec:theory-cont}.} where each buyer submits identical bids, whose correlated distribution $S$ on $[0,1-1/e]$ is given by the c.d.f. $S(b) = e^{-1}/(1-b)$. On the other hand, the unique Nash equilibrium of the auction has both buyers bidding $1$ with a probability $1$. We will see, however, that in the case of discretised auctions with a suitably concave prior $F$, the BCCE tends to be close to the unique Bayes-Nash equilibrium of the auction. On the other hand, the negative result of \citet{FLN16} will provide a counterexample that shows that some regularity assumption on the priors is necessary.

\section{LP-Based Analysis of Bayesian CCE of Discretised Auctions}\label{sec:computation}

In this section, we introduce an extended form LP based formalism to bound the distance of BCCE from a pure strategy, such as the BNE of an auction. Asking for the notion of distance to be robust to small changes in bids leads us to consider the Wasserstein-$2$ distance as an objective. Numerically evaluating the extended form LP turns out to be a difficult problem due to the exponential growth of its size in the size of the set of valuations and bids. As a result, we derive an LP relaxation that is tractable for numerical analysis and reasonable levels of discretisation. 

Our numerical results show that the Wasserstein distance is indeed an appropriate way to model the problem, and it allows for an interesting observation: in first-price and all-pay auctions, the worst-case distance of the BCCEs to the BNE shrinks with increasing levels of discretisation. For the first-price auction, we identify a dependence on the concavity of the prior distribution, which is not required for all-pay auctions.

That even relaxed LP formulations can in principle certify convergence bounds motivates us to inspect the form of dual solutions when the only BCCE of the auction is its unique pure strategy Nash equilibrium. We thus conclude the section with an analysis of dual solutions which certify properties of equilibrium, such as uniqueness or monotonicity. This final analysis enables our analysis of the \emph{continuous} BCCE we consider in Section \ref{sec:theory-cont}.

\subsection{Distance Bounding Linear Programs}\label{sec:distance-bounding-lps}

We first discuss what it means for the set of Bayesian CCE to be \emph{``close''} to a(n implicitly) given reference bidding profile (for our purposes, the canonical equilibrium), and how such ``closeness'' can be certified for fixed discretisations. To accomplish this, we first formulate an LP which measures how different buyers' bidding strategies are compared to the pure strategy Nash equilibrium of the continuous auction. Considerations of well-posedness and robustness against small changes in the magnitude of bids then lead us to pinpoint the expected Wasserstein-$2$ distance between the BCCE and the equilibrium bidding strategies of the continuous auction as the appropriate objective. 

To wit, consider a continuous auction $A$ of format $\mathcal{A} \in \{FP,AP\}$\footnote{Corresponding respectively to first-price or all-pay payment rules.}, its unique symmetric pure strategy equilibrium bidding strategy $\beta^\mathcal{A}$, and a family of its discretisations $A_n$. We are interested in proving that \emph{all} BCCE of the auction are close to $\beta^\mathcal{A}$. To do so, we consider formulating an LP whose value is greater whenever buyers place bids bounded away from the equilibrium bidding strategies, and maximise its value to obtain worst-case bounds for the BCCE of the auction. Therefore, the LP we consider should be of the form 
\begin{align}
    \max_{\sigma \in \mathbb{R}^{\profset_n}_+} \sum_{\lambda \in \profset_n} c_{n\lambda} \cdot \sigma(\lambda) \textnormal{ s.t.} & \label{opt:generic-primal} \\
    \sum_{\lambda \in \profset_n, v_{-i} \in V_n^{N-1}} \sigma(\lambda) \cdot F_{n(-i)}(v_{-i}) \cdot  ( u_i(b'_i,\lambda_{-i}(v_{-i})|v_i) - u_i(\lambda(v)|v_i) ) & \leq 0 \ \forall \ i \in [N], v_i \in V_n, b'_i \in B_n \tag{$\epsilon_i(v_i,b'_i)$} \\
    \sum_{\lambda \in \profset_n} \sigma(\lambda) & = 1, \tag{$\gamma$}
\end{align}
where $c_{n\lambda} \in \mathbb{R}^{\profset_n}$, and $F_{n(-i)}$ is the discrete joint prior distribution for buyers other than $i$. 

Our theoretical analysis will require iteratively modifying the extended form primal LP (\ref{opt:generic-primal}) and its dual, as well as their \emph{continuous, infinite} variants. Thus we mark each constraint with its corresponding dual variable, and for each \emph{class} of constraint we choose to maintain the same notation. Thus here and in extended form LPs studied throughout, constraints $(\epsilon_i)$ correspond to the (Bayesian) CCE constraints (\ref{def:CCE-II}), while the constraint $(\gamma)$ imposes that $\sigma$ is a probability distribution over pure strategy bidding profiles $\lambda : \valset_n \rightarrow \bidset_n$. We also write $\lambda_{-i} \equiv (\lambda_j)_{j \in [N] \setminus \{i\}}$ for the pure strategy bidding profile of buyers other than $i$. In particular, our distance-to-equilibrium measure is defined by $c_{n\lambda}$. Three simple choices for $c_{n\lambda}$ are:
\begin{enumerate}
    \item We let $c_{n\lambda} = 0$ whenever $\lambda_i = \beta^\mathcal{A}|_{V_n}$ for every buyer $i$, and $c_{n\lambda} = 1$ otherwise. Here, $\beta^\mathcal{A}|_{V_n}$ is the restriction of the equilibrium bidding strategy of the continuous auction to the set of discretised valuations. Therefore,  the value of (\ref{opt:generic-primal}) is the probability some buyer plays off-equilibrium strategies with respect to the continuous auction.
    \item For $\epsilon_n > 0$, we let $c_{n\lambda} = 0$ if for every buyer $i$ and every valuation $v_i$, $\lambda_i(v_i) \in [\beta^\mathcal{A}(v_i)-\epsilon_n,\beta^\mathcal{A}(v_i)+\epsilon_n]$ and $c_{n\lambda} = 1$ otherwise. As a consequence, the value of (\ref{opt:generic-primal}) is the maximum probability some buyer $i$ of valuation $v_i$ places a bid at least $\epsilon_n$ away from their equilibrium bid in the continuous auction.
    \item For a set of positive weights $w : [N] \times V_n$ and a differentiable, strictly increasing function $\phi : \mathbb{R}_+ \rightarrow \mathbb{R}_+$ such that $\phi(0) = 0$, we let 
    $$ c_{n\lambda} = \sum_{i \in [N], v_i \in V_n} w(i,v_i) \cdot \phi(|\lambda_i(v_i) - \beta^\mathcal{A}(v_i)|).$$
    Then the value of (\ref{opt:generic-primal}) measures a weighted average distance of buyers' bids from their equilibrium bids in the continuous auction.
\end{enumerate}

The first option turns out to be useful for studying the form of dual solutions in auctions which have a unique BCCE which is also a pure strategy NE, which enables our analysis in Section \ref{sec:structure} and the following study of the BCCE of continuous auctions. However, for general discretised auctions $A_n$, it is not necessarily the case that $\beta^\mathcal{A}(V_n) \subseteq B_n$; so the objective cannot be defined for every discretised auction. Moreover, even in cases when the inclusion holds, the value of (\ref{opt:generic-primal}) may fail to decrease. This is because the set of bids $B_n \subseteq \mathbb{R}_+$ also inherits a notion of distance from the reals; given $b, b' \in B_n$, the difference $|b-b'|$ can itself be small, and so too the change in players' utilities. This raises the possibility that, as we increase the discretisation size, for a worst-case distance BCCE the probability on the equilibrium bidding strategies decreases because we get to place more probability on bids between $\beta^\mathcal{A} \pm \epsilon_n$ for some strictly positive and decreasing sequence $\epsilon_n > 0$. Our numerical results suggest that this is indeed the case, and our choice of objective thus must be able to account for convergence when buyers' bids concentrate \emph{about} the canonical equilibrium bidding strategies, even though they might not concentrate \emph{on} it. Intuitively, we ask our linear objective to satisfy a robustness property, wherein our distance measure changes in a small amount in response to a small change in buyers' bids.

The second option is the form of concentration bounds considered by \citet{FGLMS21} in their work on convergence of mean-based learning algorithms to equilibria in uniform prior first-price auctions. However, our numerical results in Section \ref{sec:numerical-results} suggest that for BCCE, when we fix $\epsilon_n$ as in \cite{FGLMS21}, the worst-case concentration in this sense can be \emph{degrading} as the discretisation becomes finer. We thus consider the third option. Our solution is to define our objective function using the \emph{sum of the expectation} over buyers' valuations of the Wasserstein-$2$ distances between bidding strategies $\lambda$ they utilise in a BCCE and the equilibrium bidding strategy $\beta^\mathcal{A}$. In other words, we consider setting
\begin{equation}\label{def:Wasserstein-objective-general}
    c_{n\lambda} = \frac{1}{N} \sum_{i \in [N]} W^2_2(\lambda,\beta^\mathcal{A}) \equiv \frac{1}{N} \sum_{i \in [N]} \sum_{v_i \in V_n} F_n(v_i) \cdot (\lambda_i(v_i) - \beta^\mathcal{A}(v_i))^2
\end{equation}
for a family of discretisations parametrised by $n$, and check whether the value of (\ref{opt:generic-primal}) goes to $0$ as $n \rightarrow \infty$. We remark that if this holds, then the set of BCCE of discretised auctions converges to the canonical equilibrium of the continuous auction; in the sense that for any $\epsilon > 0$, the probability placed by the BCCE of $A_n$ on bidding profiles $\lambda \in \Lambda_n$ such that $W^2_2(\lambda,\beta^\mathcal{A}|_{V_n}) > \epsilon$ converges to $0$.

\subsection{Tractable Relaxations for Numerical Analysis}\label{sec:relaxations}

We are interested in, amongst other things, numerically evaluating bounds on the value of (\ref{opt:generic-primal}). However, the given extended LP formulation of the BCCE problem is intractable for numerical evaluation given any discretisation of meaningful size. For each $A_n$, we have a finite game in agent-normal form, and the set of CCE of a finite complete-information game may be encoded via a number of constraints polynomial in the size of the normal-form game. However, for an auction game $A_n$ with two buyers whose type space has size $|V_n|$ and bidding space has size $|B_n|$, the agent-normal form of the auction game has $2 \cdot |V_n|$ agents and thus a set of outcomes of size $(|B_n|)^{2|V_n|}$; resulting in an exponential blowup in the number of variables of the primal LP. Of course, (\ref{opt:generic-primal}) has only $1 + N \cdot |V_n| \cdot |B_n|$ constraints, which raises the possibility that its dual LP could be solved given access to a polytime separation algorithm. We do not know whether one exists for the discretised auctions we consider, however, so we shall instead aim to reduce the number of primal variables. 

To wit, note that any distribution $\sigma \in \Delta(\Lambda_n)$ induces a distribution $\sigma^{[N]} \in \Delta(\bidset_n)^{\valset_n}$ via the mapping $\sigma \mapsto \sum_{\lambda \in \Lambda_n} \sigma(\lambda) \cdot \otimes_{i \in N} \lambda_i$, i.e. for each bidding function $\lambda_i : V_n \rightarrow B_n$, we write $\lambda_i(v_i,b_i) \equiv \mathbb{I}[\lambda_i(v_i) = b_i]$, and let
$$ \sigma^{[N]}(v,b) = \sum_{\lambda \in \Lambda_n} \sigma(\lambda) \cdot \prod_{i \in N} \lambda_i(v_i,b_i).$$
In particular, for a discretised auction game, such a probability tensor $\sigma^{[N]}$ is a convex combination of the outer product of buyers' bidding strategies, represented as $|V_n|\times|B_n|$ matrices with entries in $\{0,1\}$. We shall abuse notation and denote this convex set also as $\Delta(\times_{i \in [N]} B_{n}^{V_n})$. Formulating the problem over such $\sigma^{[N]}$, we then \emph{``only''} have $(|V_n||B_n|)^N$ variables instead; as a result, for a fixed number of buyers, the number of variables in the LP is polynomial in the size of the discretisation $|V_n|, |B_n|$. 

However, we remark that we have simply traded the complexity arising from the number of variables with the complexity of separating over the set of such probability tensors; which is, in general, NP-hard even when the number of buyers and actions both equal $2$. This is due to the following reduction to bilinear programming over the hypercube, which is known to be NP-hard \cite{alon2004approximating}: given any $D \in \mathbb{R}^{d \times d}$, the values of the problems
    $$
        \max_{x,y \in \{-1,1\}^d} x^T D y \textnormal{ and } \max_{\sigma \in \Delta(\times_{i \in [2]} [d]^2)} \sum_{v_1, v_2 \in [d], b_1, b_2 \in [2]} \sigma(v,b) D_{v_1 v_2} (-1)^{b_1 + b_2}
    $$
are equal. Moreover, solutions of the first problem are in bijection with the extremal solutions of the second problem. Indeed, such an extremal solution $\sigma$ is of the form $\lambda_1 \otimes \lambda_2$; we then let $\lambda_1(v_1, 1) = 1$ if and only if $x_{v_1} = -1$ and likewise for $\lambda_2(v_2,1)$ and $y_{v_2}$.

Consequently, as we do not expect a polynomial size LP for an NP-hard problem, we do not hope to provide a succint set of linear inequalities which describe $\Delta(\times_{i \in [N]} B_{n}^{V_{n}})$ for any discretisation $A_n$ of reasonable size. Instead, we shall enforce the \emph{``obvious''} valid inequalities, obtaining a polynomial-size relaxation. For any $\sigma^{[N]} \in \Delta(\times_{i \in [N]} B_n^{V_n})$, the entries of $\sigma^{[N]}$ are non-negative; moreover $\sigma^{[N]}(v,\cdot)$ is a probability distribution over $B_n^N$ for any vector of buyers' valuations $v \in \valset_n$. 

We remark also that the assumption of a symmetric prior distribution and a uniform tie-breaking rule together imply that, whenever the objective is chosen to be a symmetric function over the buyers' indices, then there always exists an optimal solution $\sigma^{[N]}$ of the LP also symmetric over the buyers' indices, obtained by averaging an optimal solution over all permutations of buyers\footnote{In particular, this notion of symmetry of outcomes is independent from the \emph{symmetric support assumption} considered in our theoretical analysis in subsequent Section \ref{sec:theory-cont} and onwards.}. Thus, despite the (constant factor) increase in the number of variables, for convenience in writing the constraints we will also have variables $\sigma^{[\ell]} : V_n^\ell \times B_n^\ell \rightarrow [0,1]$ for each $1 \leq \ell \leq N$ for the joint distribution of bids of $\ell$ buyers. We then enforce supersymmetry constraints (i.e., symmetry over all indices of a tensor) over the buyers' indices,
\begin{equation}\label{cons:sym}
    \forall \textnormal{ bijections } p : [N] \rightarrow [N], v \in \valset_n, b \in \bidset_n, \sigma^{[N]}(v,b) = \sigma^{[N]}(v_{p(N)},b_{p(N)}),
\end{equation}
where for any $i \in [N]$, $(v_{p(N)})_i = v_{p^{-1}(i)}$ and likewise for $b_{p(N)}$. Such supersymmetry constraints then allow us to not have to consider $\sigma^J$ for any $\emptyset \subsetneq J \subseteq [N]$ instead. In this case, we can write the constraint that $\sigma^{[N]}$ is a probability distribution as 
\begin{equation}
    \sum_{b_1 \in B_n} \sigma^{[1]}(v_1,b_1) = 1 \ \forall \ v_1 \in V_n, \label{cons:prob1}
\end{equation}
provided that $\sigma^{[N]}$ contracts to $\sigma^{[1]}$ in the appropriate manner. This is assured by enforcing the constraint that, for any $\ell > 1$, any valuation vector $v_{[\ell]} \in V_n^\ell$ and any bid vector $b_{[\ell-1]}$ of the first $\ell-1$ buyers,
\begin{equation}\label{cons:restrict1}
    \sum_{b_\ell \in B_n} \sigma^{[\ell]}\left(v_{[\ell]},b_{[\ell]}\right) - \sigma^{[\ell-1]}_n\left( v_{[\ell-1]},b_{[\ell-1]}\right) = 0. 
\end{equation}
This is because, for buyers' strategies $(\lambda_i : V_n \rightarrow B_n)_{i \in [N]}$,
\begin{equation}\label{eqn:derive-valid}
    \sum_{b_\ell \in B_n} \prod_{i \in [\ell]} \lambda_i(v_i,b_i) = \left( \prod_{i \in [\ell-1]} \lambda_i(v_i, b_i)\right) \cdot \left( \sum_{b_\ell \in B_n} \lambda_\ell(v_\ell,b_\ell) \right) = \prod_{i \in [\ell-1]} \lambda_i(v_i, b_i),
\end{equation}
independent of the valuation of buyer $\ell$. Intuitively, each $\sigma^{[\ell]}(v^{\ell},\cdot)$ corresponds to the joint probability distribution of the bids placed by the buyers in $[\ell]$, given their valuations are $v^{[\ell]}$. Meanwhile, a BCCE is obtained when, in the learning process, buyers implement their bidding strategies independently in each round. This implies that the distribution of bids for the buyers in $[\ell-1]$ should not depend on the valuation of buyer $\ell$. As a consequence, we may obtain the distribution of bids for the buyers in $[\ell-1]$ by simply taking  $\sigma^{[\ell]}(v^{\ell},\cdot)$ for any valuation $v_\ell$ of buyer $\ell$, and then marginalising the distribution by summing over the probabilities of buyer $\ell$’s bids. 

Finally, the BCCE constraints (\ref{def:CCE-II}) are then given, for each buyer $i$, valuation $v_i$, and bid $b'_i$,
\begin{equation}
    \sum_{v_{-i} \in V^{N-1}_n, b \in \bidset_n} \sigma^{[N]}(v,b) \cdot F_{n(-i)}(v_{-i}) \cdot \left( u_i\left(b'_i, b_{-i}|v_i \right) - u_i(b|v_i) \right) \leq 0. \label{cons:BCCE}
\end{equation}

Combining our insights on the preceding discussions on variables, constraints and the objective, the class of LPs we solve have the form
$$ \max_{\sigma \geq 0} \sum_{0 < \ell \leq N} \sum_{v^{[\ell]} \in V_n^{[\ell]}, b^{[\ell]} \in B_n^{[\ell]}} \sigma^{[\ell]}(v_{[\ell]},b_{[\ell]}) \cdot c^{[\ell]}(v_{[\ell]},b_{[\ell]}) \textnormal{ s.t. } (\ref{cons:sym}), (\ref{cons:prob1}), (\ref{cons:restrict1}), (\ref{cons:BCCE})$$
for a positive linear objective $c$. When we take $c$ to be the Wasserstein-$2$ distance (\ref{def:Wasserstein-objective-general}) to the equilibrium bidding strategy $\beta^\mathcal{A}$ in the objective of (\ref{opt:generic-primal}) and employ our change of variables, we find that the objective linear function contracts suitably along the equalities (\ref{cons:restrict1}), yielding the LP 
\begin{equation}\label{opt:Wasserstein}
    \max_{\sigma \geq 0} \sum_{v_1 \in V_n, b_1 \in B_n} \sigma^{[1]}(v_1,b_1) \cdot F_n(v_1) (b_1 - \beta^\mathcal{A}(v_1))^2  \textnormal{ s.t. } (\ref{cons:sym}), (\ref{cons:prob1}), (\ref{cons:restrict1}), (\ref{cons:BCCE}).
\end{equation}

We will also implement an objective testing \emph{concentration bounds}. Our aim to check whether a result similar to that of \citet{FGLMS21} holds for the BCCE of discretised auctions. In particular, \citet{FGLMS21} show that for the single-item first price auction with uniform priors, and its discretisations with $V_n = B_n = \{1/n, 2/n \ldots, 1\}$, when buyers use mean based learning algorithms with a randomised initial learning period to learn how to bid, each buyer $i$ of valuation $v_i$ learns to bid $B_n \cap [\beta^{FP}(v_i), \beta^{FP}(v_i)+1/n)$ with high probability; and we would like to confirm whether convergence in this sense can be verified by our relaxed LPs for reasonably sized discretisations. Towards this end, we consider the problem 
\begin{equation}\label{opt:concentration}
    \min_{\sigma \geq 0} \sum_{v_1 \in V_n, b_1 \in [\beta^\mathcal{A}(v_1)-1/n,\beta^\mathcal{A}(v_1)+1/n] \cap B_n} \sigma^{\{[1]\}}(v_1,b_1) \textnormal{ s.t. } (\ref{cons:sym}), (\ref{cons:prob1}), (\ref{cons:restrict1}), (\ref{cons:BCCE}).
\end{equation}

Finally, we remark that the discussion in this section extends immediately also for the BCE of discretised auctions. In particular, to numerically bound the concentration of the agent-normal form BCE on the equilibrium bidding strategies or to bound its Wasserstein-$2$ distance to the equilibrium, we may use exactly the same LPs (\ref{opt:concentration}) and (\ref{opt:Wasserstein}), by simply exchanging the family of constraints (\ref{cons:BCCE}) with the constraints
\begin{equation}
    \sum_{v_{-i} \in V_n^{N-1}, b_{-i} \in B_n^{N-1}} \sigma^{[N]}(v,b) \cdot F_{n(-i)}(v_{-i}) \cdot \left( u_i\left(b'_i, b_{-i}|v_i \right) - u_i(b|v_i) \right) \leq 0\label{cons:BCE}
\end{equation}
for each buyer $i$, valuation $v_i$, and pair of bids $b_i, b'_i$, which is equivalent to (\ref{def:CE-II}).

\subsection{Numerical Results}\label{sec:numerical-results}

In this section, we report on our numerical results on the B(C)CE of discretised first-price and all-pay auctions with two buyers. Our results in general reflect a dependence on the \emph{information structure} of the auction for first-price auctions. In particular, for discretised first-price auctions, whether their BCCE are necessarily close to its continuous equilibrium bidding strategies depends on the concavity of the prior distribution. Such a dependence on prior distributions is notably absent for discretised all-pay auctions, for which we observe unconditional convergence.

We will evaluate our measures of worst-case distance bounds for B(C)CE of these auctions for priors when buyers' valuations are drawn i.i.d. from (a discretisation of) a power law distribution on $[0,1]$, an exponential distribution on $[0,\infty)$, or a truncated Gaussian distribution on $[0,1]$.  We evaluate our results on \emph{naive} discretisations unless otherwise stated, where the valuation and bid sets consist of equally separated points within the interval $[0,1]$, 
$$V_n = B_n = \{0,1/n,\ldots,1-1/n,1\}.$$
Meanwhile, the discrete prior distribution $F_n$ over $V_n$ is fixed in proportion to the prior p.d.f. $f$ for the continuous auction. We elaborate on the exact form of the discretisations below.

\newcommand{\erf}{\textnormal{erf}}

\begin{enumerate}
  \item The power law distribution on $[0,1]$ with parameter $\alpha > 0$ has c.d.f. $F^{\alpha}(v) = v^\alpha$. For the naive discretisation of size $n$, for $\alpha \in (0,1)$, note that the p.d.f. $f(v) \rightarrow \infty$ as $v \rightarrow 0$. To rectify this,  we opt to fix the discretised probability function $F_n$ by setting
  $$ F_n(v) = F(v+1/n) - F(v) \ \forall \ v \in V_n.$$

  \item The exponential distribution on $[0,\infty)$ with parameter $\lambda$ follows the c.d.f. $F^\lambda(v) = 1 - e^{-\lambda v}$. For the naive discretisation of size $n$, we set 
  $V_n, B_n,$ as before, 
  cutting off the distribution on $[0,1]$. Furthermore, for each $v \in V_n$, we set $F_n(v) \sim \lambda e^{-\lambda v}$.

  \item The truncated Gaussian distribution on $[0,1]$ with mean $\mu$ and variance $\sigma$ follows the c.d.f. 
  $$F^{\mu,\sigma}(v) = \frac{\textnormal{erf}\left(\frac{v-\mu}{\sigma\sqrt{2}}\right)-\textnormal{erf}\left(\frac{-\mu}{\sigma\sqrt{2}}\right)}{\textnormal{erf}\left(\frac{1-\mu}{\sigma\sqrt{2}}\right)-\textnormal{erf}\left(\frac{-\mu}{\sigma\sqrt{2}}\right)}.$$ 
  For this reason, for the naive discretisation  we fix $F_n(v) \sim e^{-\frac{1}{2}\left(\frac{x-\mu}{\sigma}\right)^2}$ for each $v \in V_n$. 
\end{enumerate}

We run all our numerical experiments on a desktop computer with an Intel i5 12600KF CPU and 32GB RAM. We build our models in MATLAB and use MOSEK as our solver. We inspect the B(C)CE of power law distributions for parameters $\alpha \in \{.5,1,\ldots,3\}$, the exponential distribution for $\lambda \in \{1,2,3,4\}$, and the truncated normal distribution for $(\mu,\sigma) \in \{0,.25,.5,.75\} \times \{.2,.5,1\}$. For two buyers, we can test instances with $|V_n| \leq 21$ for BCCE and $|V_n| \leq 18$ for BCE.

Our general result is that, for first-price auctions (with concave priors) and for all-pay auctions, buyers' marginal probability distributions on their bids concentrates heavily about the canonical equilibrium bidding strategies; i.e. most of the weight of the distribution $\sigma^{[1]}(v,\cdot)$ is placed about the equilibrium bid $\beta^\mathcal{A}(v)$. Figure \ref{fig:choice_1} demonstrates this phenomenon. Our numerical results then show that Wasserstein distance bounds certify convergence in this sense, insofar that they are declining in the size of the discretisation whenever such concentration occurs for a prior distribution. This is despite the probability mass placed on $[\beta(v)-1/n,\beta(v)+1/n]$ by the buyers declining in the discretisation size.

\begin{figure}
  \begin{center}
      \includegraphics[width=0.49\textwidth]{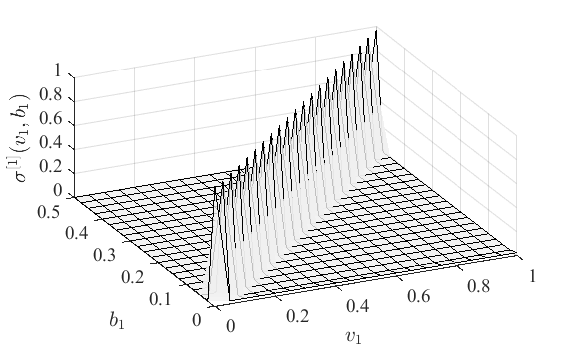}   
      \includegraphics[width=0.49\textwidth]{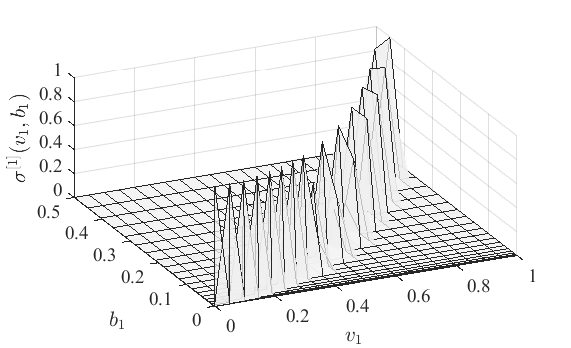}
  \end{center}
  \caption{A buyer's distribution of bids for Wasserstein distance maximising BCCE of first-price and all-pay auctions with two buyers, for $V_n = \{0,1/20,\ldots,1\}$ and $B_n = \beta^\mathcal{A}(V_n)$. For the first-price auction, the BCCE correlates on both buyers bidding the equilibrium strategy or both buyers bidding $0$. The form of the BCCE for the all-pay auction is less clear instead, and the maximal Wasserstein-$2$ distance BCCE does not appear to admit an obvious parametric form in $n$.}\label{fig:choice_1}
\end{figure}

\subsubsection{Convergence in first-price auctions}

For the BCCE of first-price auctions, we observe that whether the values of (\ref{opt:Wasserstein}) are declining in the discretisation size $n$ is connected to the concavity of the prior distribution; we illustrate this dependence in Figure \ref{fig:fp_2}. For naive discretisations with concave priors, the Wasserstein distance bounds are declining, with an observed rate of $O(1/n)$. However, these bounds in general do not exhibit a specific pattern in how exactly they decline as a function of the discretisation size; in particular, they need not be monotonically decreasing. One exception is when the prior distribution is uniform, in which case the concentration bounds decline over even $n$ or odd $n$, with lower Wasserstein distance bounds for even $n$. Meanwhile, for significantly non-concave priors like the power law distribution for $\alpha > 1$ or a truncated normal distribution on $[0,1]$ with positive mean and small variance, the values of (\ref{opt:Wasserstein}) remain bounded away from $0$ as the discretisation size grows large. We also note that BCE of first-price auctions still exhibit declining Wasserstein distance bounds without any condition on the priors. 

\begin{figure}
  \begin{center}
      \includegraphics[width=0.4\textwidth]{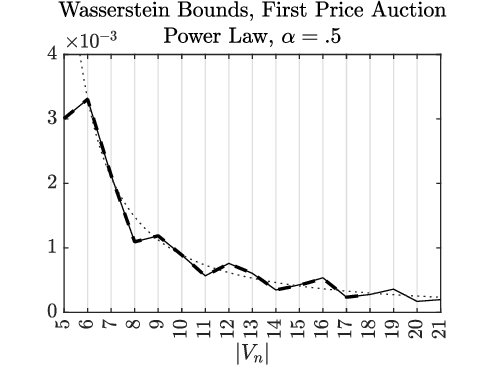}   
      \includegraphics[width=0.4\textwidth]{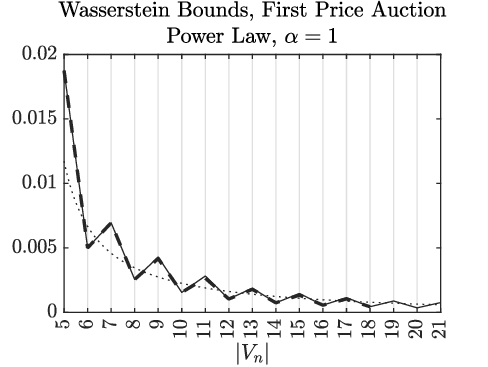}
      \includegraphics[width=0.4\textwidth]{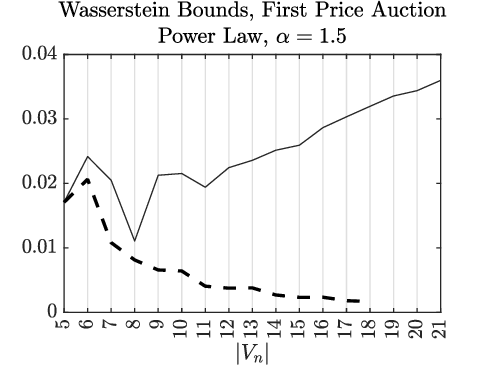}   
      \includegraphics[width=0.4\textwidth]{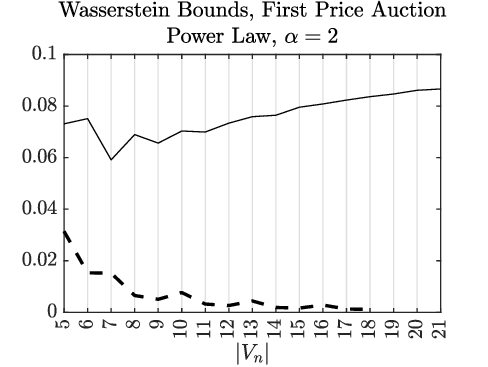}
      \includegraphics[width=0.4\textwidth]{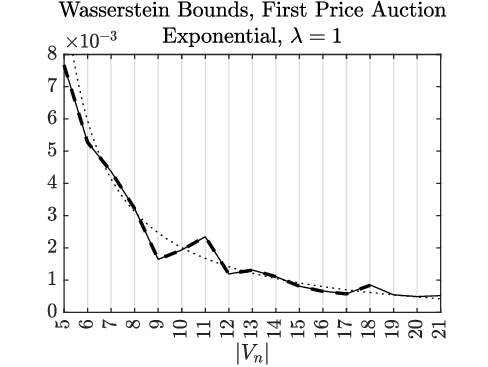}   
      \includegraphics[width=0.4\textwidth]{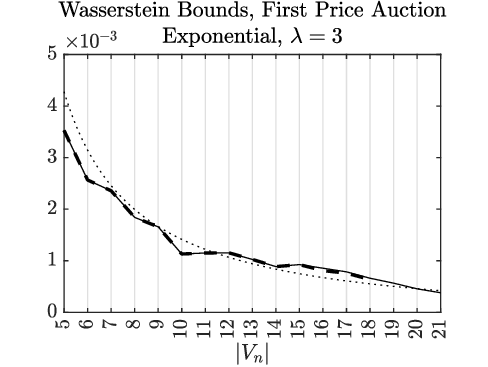}
      \includegraphics[width=0.4\textwidth]{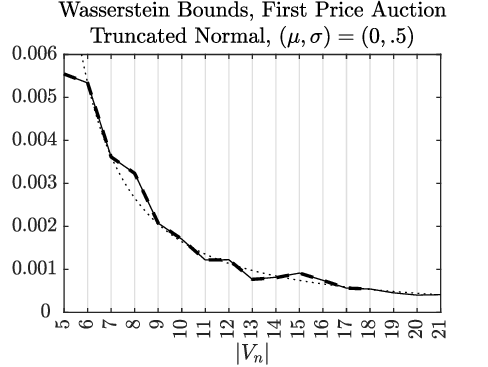}   
      \includegraphics[width=0.4\textwidth]{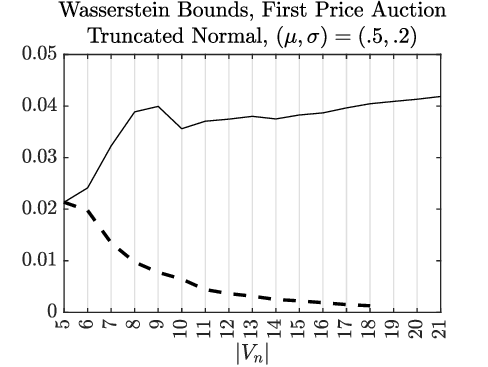}  
    \end{center}
  \caption{Values of (\ref{opt:Wasserstein}) for discretised first-price auctions with two buyers. In this and all further figures, (1) dashed (solid) lines correspond to bounds for B(C)CE, and (2) dotted lines correspond to a $O(n)$ rational function fit for the BCCE bounds with a degree $1$ numerator and degree $2$ denominator. \label{fig:fp_2}}
\end{figure}

The non-monotonicity of the Wasserstein bounds in the size of the discretisation belies the difficulty in obtaining an explicit parametric bound via the size of the discretisations we consider. However, by considering a different family of discretisations, we are able to identify an explicit family of primal optimal solutions. Specifically, for the discretisation $V_n = \{0,1/n,\ldots,1\}$ and $B_n = V_n / 2$ with a uniform prior, if $n\leq 22$ \emph{a} Wasserstein $2$-distance maximising BCCE with respect to the canonical equilibrium bids in the continuous auction is given 
$$\sigma^{[2]}_n(v_1,b_1,v_2,b_2) = \frac{2|V_n|}{2|V_n|+1}\mathbb{I}[(b_1,b_2) = (v_1,v_2)/2] + \frac{1}{2|V_n|+1}\mathbb{I}[(b_1,b_2) = (0,0)].$$
This is in fact precisely the BCCE depicted in Figure \ref{fig:choice_1}. For other distributions, we remain unable to identify a similar explicit family of Wasserstein-distance maximal BCCEs. We conclude that in general, the type of discretisation matters when going from a continuous to a discretised model. This motivates the analysis of the continuous case in the next section.

\subsubsection{Convergence in all-pay auctions}

For the Wasserstein bounds on the BCCE of discretised all-pay auctions, we no longer observe the dependence on the concavity of the prior distribution exhibited in the setting of the first-price auction for two buyers. Here, the Wasserstein distance bounds instead decline with $n$ unconditionally in all instances we consider. However, as in the case of first-price auctions, the decline of these bounds are non-monotone in general. Furthermore, while there is a difference between the Wasserstein distance bounds for BCE and BCCE of each auction instance, these differences turn out to be insignificant; this is illustrated in Figure \ref{fig:ap_2}.

\begin{figure}
  \begin{center}
      \includegraphics[width=0.4\textwidth]{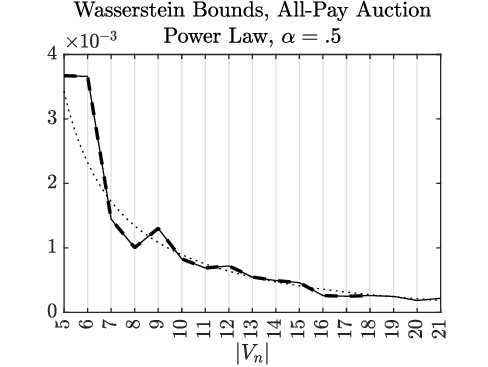}   
      \includegraphics[width=0.4\textwidth]{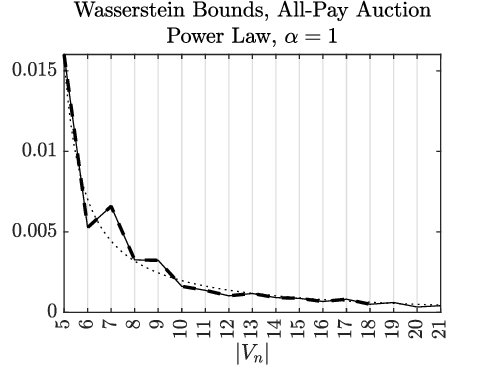}
      \includegraphics[width=0.4\textwidth]{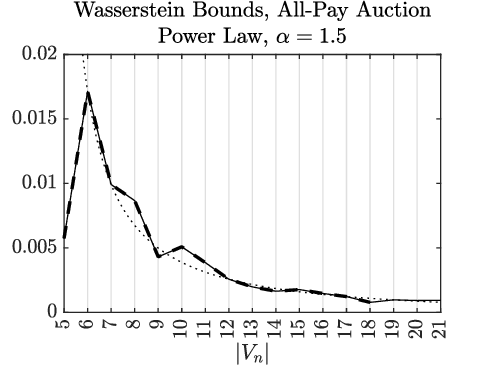}   
      \includegraphics[width=0.4\textwidth]{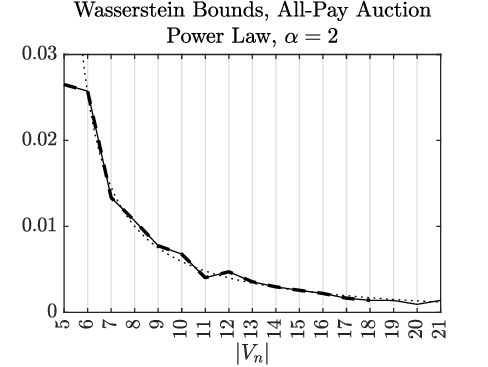}
      \includegraphics[width=0.4\textwidth]{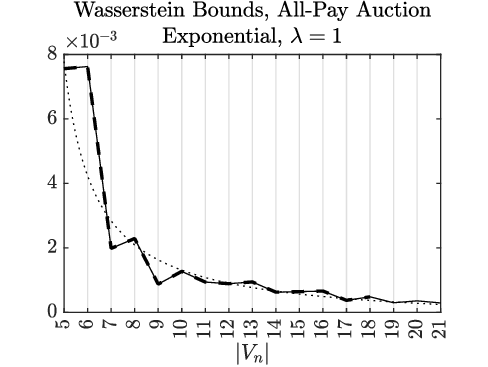}   
      \includegraphics[width=0.4\textwidth]{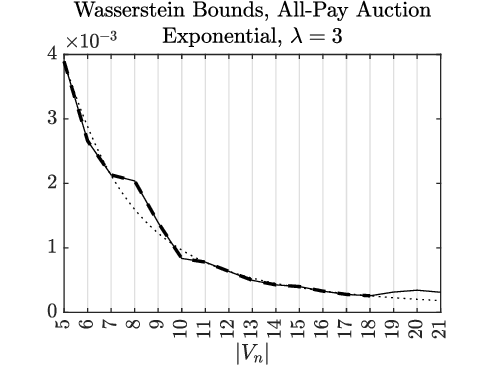}
      \includegraphics[width=0.4\textwidth]{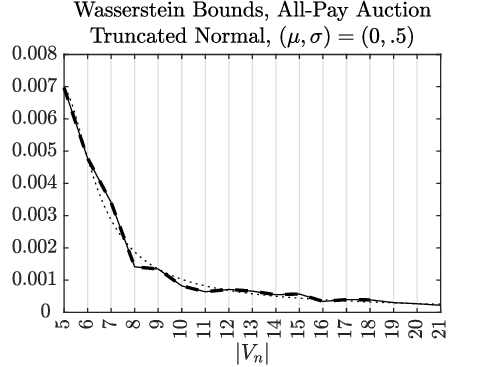}   
      \includegraphics[width=0.4\textwidth]{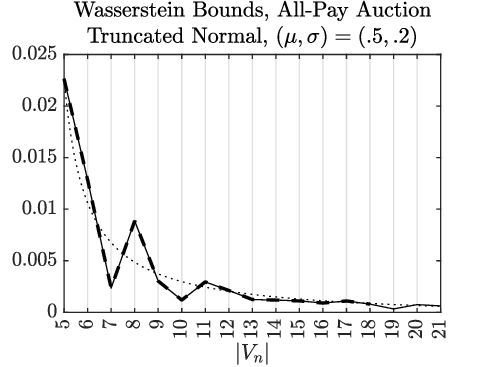}    
  \end{center}
  \caption{Values of (\ref{opt:Wasserstein}) for various discretised all-pay auctions with two buyers. \label{fig:ap_2}}
\end{figure}

\subsubsection{On concentration} Our final result is that concentration in the sense of \cite{FGLMS21} is not necessarily assured for the BCCE first-price and all-pay auctions, at least for the LP relaxations we consider. For the naive discretisations above, we consider the problem (\ref{opt:concentration}); we remark that this is a relaxation of the concentration bounds of \cite{FGLMS21}, who showed that with a pre-training round, mean-based learners concentrate their bids on $[\beta(v),\beta(v)+1/n)$ for first-price auctions with uniform priors. We demonstrate in Figure \ref{fig:concentration} that, while buyers may bid about the canonical BNE with high probability, there are discretisations for which the relaxed B(C)CE of the auction includes bids not in $[\beta(v)-1/n,\beta(v)+1/n]$. As a consequence, we infer that one or more of the following holds; (1) the results of \cite{FGLMS21} are truly dependent on the pre-training period, (2) that our LP based methodology is too weak to certify convergence for general mean-based learning algorithms, or (3) convergence in this sense is exhibited at discretisation sizes much higher than ones we can solve for numerically. In Figure \ref{fig:concentration} we also see that concentration about the BNE does not exhibit any clear pattern for the discretisation sizes we consider; althought there is some general trend of declining concentration bounds for odd $|V_n|$.

\begin{figure}
  \begin{center}
      \includegraphics[width=0.45\textwidth]{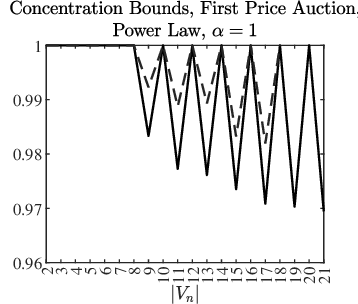}   
      \includegraphics[width=0.45\textwidth]{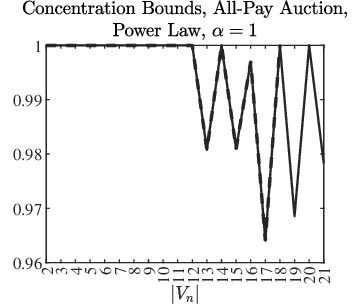}
  \end{center}
  \caption{Values of (\ref{opt:concentration}) for the uniform discrete distribution on $V_n = \{0,1/n,\ldots,1.\}$ and $B_n = V_n$, for first-price and all-pay auctions.\label{fig:concentration}}
\end{figure}

\subsection{Structural Properties of Bayesian CCE in Discretised Auctions}\label{sec:structure}

Our numerical results in Section \ref{sec:numerical-results} suggest that, for a continuous first-price auction with concave priors or for a continuous all-pay auction, the BCCE of a $O(1/n)$-fine family $(V_n,B_n,F_n)$ of its discretisations has $O(1/n)$ Wasserstein-$2$ distance to its canonical equilibrium. Furthermore, in principle, these vanishing bounds admit a \emph{dual proof certificate}, i.e. a family of solutions to the dual of (\ref{opt:generic-primal}) or (\ref{opt:Wasserstein}) for each discretisation $A_n$ with value $O(1/n)$, which by the weak duality of linear programming proves that the BCCE of $A_n$ converges in probability to the canonical equilibrium as $n \rightarrow \infty$.

However, the dual problem of (\ref{opt:Wasserstein}) does not appear to admit easy to characterise families of such dual solutions. By further leveraging the symmetries of the problem, for an auction with two buyers, the dual LP of (\ref{opt:Wasserstein}) may be written 
\begin{align}
  \min_{\epsilon \geq 0, \mu, \rho} \sum_{v_1 \in V_n} \mu(v_1) \textnormal{ s.t.} & \label{opt:dual-Wasserstein} \\
  \mu(v_1) - \sum_{v_2 \in V_n} \rho(v_1,b_1,v_2) & \geq F_n(v_1) \cdot (b_1 - \beta^\mathcal{A}(v_1))^2 \ \forall \ v_1 \in V_n, b_1 \in B_n \nonumber\\
  \rho(v_1,b_1,v_2) + \rho(v_2,b_2,v_1) & \ldots  \nonumber\\
  + \sum_{b' \in B_n} \epsilon(v_1,b') \cdot F_n(v_2) \cdot (u_1(b',b_2|v_1) - u_1(b^{[2]}|v_1)) & \ldots  \nonumber\\
   + \sum_{b' \in B_n} \epsilon(v_2,b') \cdot F_n(v_1) \cdot (u_2(b',b_1|v_2) - u_1(b^{[2]}|v_2)) & \geq 0 \ \forall \ v_1, v_2 \in V_n, b_1, b_2 \in B_n. \nonumber
\end{align}
Here, $\epsilon_i$ corresponds to the equilibrium constraints (\ref{cons:BCCE}), while $\mu$ and $\rho$ are dual variables associated respectively with the probability constraint (\ref{cons:prob1}) and (\ref{cons:restrict1}). 

However, we are unable to derive explicit solutions to (\ref{opt:dual-Wasserstein}) for a family of discretisations. The issue is that the solutions to (\ref{opt:dual-Wasserstein}) are in general poorly behaved. Figure \ref{fig:dual} demonstrates this phenomenon in the context of a first-price auction with uniform priors on $V_n = \{0,1/n,\ldots,1\}$ and $B_n = V_n$. We observe that $\epsilon(v_1,b_1) > 0$ when $b_1 \simeq \beta(v_1) = v_1/2$, however, \emph{how} $\epsilon$ disperses about the equilibrium bids appears to follow no easily discernible pattern. Any attempts to characterise the form of $\mu$ or $\rho$ turn out to be equally unfruitful.

\begin{figure}
  \begin{center}
      \includegraphics[width=0.45\textwidth]{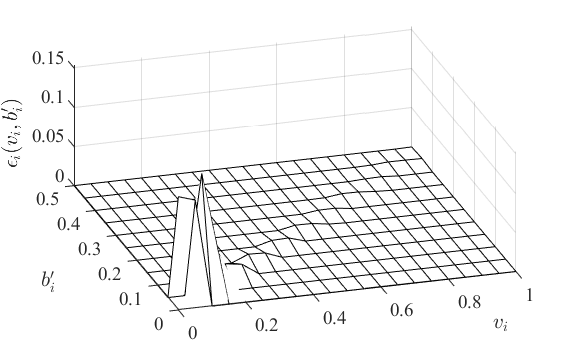}   
      \includegraphics[width=0.45\textwidth]{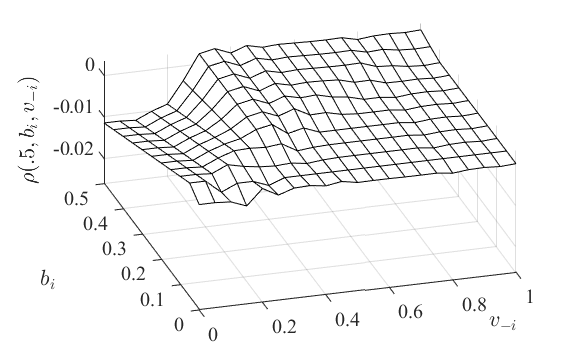}
  \end{center}
  \caption{Values of $\epsilon_i(v_i,b'_i)$ and $\rho(.5,b_i,v_{-i})$ for the first-price auction with uniform discrete prior distribution on $V_n = B_n = \{0,1/20,\ldots,1.\}$.\label{fig:dual}}
\end{figure}

This motivates us to adopt an alternative approach. In this section, we study properties of the dual LP of (\ref{opt:generic-primal}), its solutions, and its corresponding binding primal variables. This leads us to a characterisation of when a discretised auction has a unique BCCE in pure strategies in terms of the form of the dual variables $\epsilon$ for an optimal solution. Then in Section \ref{sec:theory-cont} we define a notion of BCCE for the continuous auction, which intuitively corresponds to the $n \rightarrow \infty$ limit of (\ref{opt:generic-primal}). For this infinite LP, our results in this section allow us to specify the forms of appropriate dual solutions which yield proofs of all our main theoretical results.

To wit, consider the dual of the extended form LP (\ref{opt:generic-primal}), which may be written
\begin{align}
  \min_{\epsilon \geq 0, \gamma} \gamma \textnormal{ s.t.} & \label{opt:generic-dual} \\
  \gamma + \sum_{i \in [N], v \in \valset_n, b'_i \in B_n} \epsilon_i(v_i,b'_i) \cdot F_{n(-i)}(v_{-i}) (u_i(b'_i,\lambda_{-i}(v_{-i})|v_i) - u_i(\lambda(v)|v_i)) & \geq c_{n\lambda} \ \forall \ \lambda \in \Lambda_n. \tag{$\sigma(\lambda)$} \label{cons:eq}
\end{align}

Our first result aims to simply the dual problem (\ref{opt:generic-dual}) when the auction has a unique BCCE, which places probability $1$ on a pure strategy Nash equilibrium $\lambda^*$. The following proposition accomplishes that, by showing that we may assume that $\epsilon_i(v_i,b'_i)$ is non-zero only when $b'_i$ is an equilibrium bid. The result in fact applies to uniqueness of the BCCE in pure strategies for general normal-form games, though we express it in the language of our discretised auction settings.

\begin{proposition}\label{prop:dual-form}
  Suppose that $c_{n\lambda} = 1-\delta_{\lambda \lambda^*}$. If the optimal value of (\ref{opt:generic-dual}) equals $0$, then for every dual optimal solution $\epsilon$, every buyer $i$ and every type $v_i$, $\epsilon_i(v_i,b'_i) > 0 \Leftrightarrow b'_i = \lambda^*_i(v_i)$.
\end{proposition}

\begin{proof}
  First suppose that we have a dual optimal solution $\epsilon$ such that for some buyer $i$ and a type $v_i$, $\epsilon_i(v_i,b'_i) = 0$ for every bid $b'_i$. Then by complementary slackness, dropping the equilibrium constraints for buyer $i$ of type $v_i$ does not change the value of the primal problem (\ref{opt:generic-primal}). Then consider solution $\sigma$ to the relaxed problem, obtained by committing $(i,v_i)$ to bidding $\lambda_i(v_i) \neq \lambda^*_i(v_i)$\footnote{Right now we implicitly assume that $|B| > 1$.} in the agent-normal form of the auction, and then finding a Nash equilibrium in the resulting finite game. By construction, $\lambda \neq \lambda^*$, hence the primal LP with the $\epsilon_i(v_i,\cdot)$ constraints omitted has value $> 0$; contradiction.  
  
  So it remains to show that for any $b'_\ell \neq \lambda_\ell(v_\ell), \epsilon_i(v_\ell,b'_\ell) = 0$. Again suppose for a contradiction that for some buyer $\ell$, type $v_\ell$ and bid $b'_\ell \neq \lambda^*_\ell(v_\ell)$, $\epsilon_\ell(v_\ell,b'_\ell) > 0$. Consider the mixed bidding strategies 
  $$ s_i^*(v_i,b_i) = \frac{\epsilon_i(v_i,b_i)}{\sum_{b'_i \in B} \epsilon_i(v_i,b'_i)} \ \forall \  i\in [N], v_i \in V_n, b_i \in B_n. $$
  We then consider the value of the sum, 
  \begin{align}\label{eqn:uniqueness}
    &\sum_{\lambda \in \profset_n} \left( \prod_{(j,v'_j) \in \sqcup_{j \in [N]} V_n} s^*_j(v'_j,\lambda_j(v'_j)) \right) \ldots \\ & \cdot \sum_{i \in [N], v \in \valset_n, b'_i \in B_n} \epsilon_i(v_i,b'_i) F_{n(-i)}(v_{-i})  \left( u_i(b'_i, \lambda_{-i}(v_{-i}) | v_i) - u_i(\lambda(v)|v_i) \right).\nonumber
  \end{align}
  Note that each term $\left( \prod_{(j,v'_j) \in \sqcup_{j \in [N]} V_n} s^*_j(v'_j,\lambda_j(v'_j)) \right)$ is the probability the bidding profile $\lambda$ is realised, when each buyer $j$ \emph{independently} implements randomised bidding strategies $s^*_j \in \Delta(B_n)^{V_n}$. By dual feasibility, i.e. by considering (\ref{eqn:uniqueness}) as a weighted sum over constraints (\ref{cons:eq}), (\ref{eqn:uniqueness}) is $> 0$. This is since by assumption, for some buyer $\ell$ and type $\ell$, $s^*_\ell(v_\ell, \lambda^*_\ell(v_\ell)) < 1$. On the other hand, 
  \begin{align*}
      & \sum_{\lambda \in \profset_n, v_{-i} \in V_{n}^{N-1}} \left( \prod_{(j,v'_j) \in \sqcup_{j \in [N]} V_n} s^*_j(v'_j,\lambda_j(v'_j)) \right) \cdot F_{n(-i)}(v_{-i}) \cdot \left( u_i(b'_i, \lambda_{-i}(v_{-i}) | v_i) - u_i(\lambda(v)|v_i) \right) \\
      = & \sum_{b \in \bidset_n, v_{-i} \in V_{n}^{N-1}} \left( \prod_{j \in [N]} s^*_j(v_j,b_j) \right) \cdot F_{n(-i)}(v_{-i}) \cdot \left( u_i(b'_i, b_{-i} | v_i) - u_i(b|v_i) \right) \\
      = & \sum_{b_i \in B_n} \left( \delta_{b_i b'_i} - s^*_i(v_i,b_i) \right) \cdot \sum_{b_{-i} \in B_{n(-i)}, v_{-i} \in V_{n(-i)}} \left( \prod_{j \neq i} s^*_j(v_j,b_j) \right) \cdot F_{n(-i)}(v_{-i}) \cdot u_i(b|v_i).
  \end{align*}
  Therefore, by rearranging the sum, (\ref{eqn:uniqueness}) is equal to 
  $$ \sum_{i \in [N], v_i \in V_n, b'_i \in B_n} \epsilon_i(v_i,b'_i) \cdot \sum_{b_i \in B_n} \left( \delta_{b_i b'_i} - s^*_i(v_i,b_i) \right) \cdot \sum_{b_{-i} \in B_{n}^{N-1}, v_{-i} \in V_{n}^{N-1}} \left( \prod_{j \neq i} s^*_j(v_j,b_j) \right) \cdot f(v) \cdot u_i(b|v_i).$$ 
  However, by construction of $s^*$, for any $i \in [N]$, $v_i \in V_n$, and $b'_i \in B_n$,
  \begin{align*}
    \sum_{b'_i \in B_n} \epsilon_i(v_i,b'_i) \cdot  (\delta_{b_i b'_i} - s^*_i(v_i,b_i)) & = \epsilon_i(v_i,b_i) - \sum_{b'_i \in B_n} \epsilon_i(v_i,b'_i) \cdot \frac{\epsilon_i(v_i,b_i)}{\sum_{b''_i \in B_n} \epsilon_i(v_i,b''_i)} \\
      & = \epsilon_i(v_i,b_i) - \epsilon_i(v_i,b_i) = 0.
  \end{align*}
  Thus we have $0 > 0$, contradiction.
\end{proof}

Proposition \ref{prop:dual-form} allows us to limit the number of equilibrium constraints we need to consider in settings where we suspect the BCCE is unique in pure strategies. In particular, by complementary slackness, the value of (\ref{opt:generic-primal}) does not change when we disregard equilibrium constraints $\epsilon_i(v_i,b'_i)$ for $b'_i \neq \lambda^*_{i}(v_i)$. We may thus only impose constraints $\epsilon_i(v_i,\lambda^*_{i}(v_i))$ in the primal problem, and adjust the dual LP (\ref{opt:generic-dual}) appropriately.

We also remark on a connection of the result of Proposition \ref{prop:dual-form} to learning algorithms based on gradient descent or dual averaging methods. Towards this end, for a game of complete information $\Gamma$, call a pure strategy (Bayes-)Nash equilibrium $a^* \in A$ \textbf{strict} if for any player $i \in N$ and any alternate action $a'_i \in A_i$, $u_i(a'_i, a_{-i}) < u_i(a)$. In \cite{MZ19} it is shown that for so-called dual averaging based learning algorithms a strict Nash equilibrium is \emph{locally stable}. Meanwhile, \cite{Xiao10} shows that dual averaging based learning algorithms converge to a CCE for concave games. Utilities in our distributional setting are multilinear for the purposes of these learning algorithms, hence concave in the argument of each player, and we infer that strictness of a BNE is a necessary condition for the uniqueness of BCCE as a pure strategy BNE.  This conclusion, of course, follows readily from Proposition \ref{prop:dual-form}.

\begin{corollary}
  Suppose that $c_{n\lambda} = 1-\delta_{\lambda \lambda^*}$. If the optimal value of (\ref{opt:generic-primal}) equals zero, then $\lambda^*$ is a strict Nash equilibrium.
\end{corollary}

\begin{proof}
    Take any optimal dual solution $\epsilon$. By Proposition \ref{prop:dual-form}, $\epsilon_i(v_i,b'_i) = 0$ for any $b'_i \neq \lambda^*_i(v_i)$. However, if there exists $b'_i$ such that $\mathbb{E}_{v_{-i} \sim F_{n(-i)}} [ u_i(b'_i, \lambda^*_{-i}(v_i)|v_i)  - u_i(\lambda^*(v) | v_i) ] = 0$, then for large enough $M$ and sufficiently small $\delta > 0$, 
    $$ \epsilon'_j(v_j,b_j) = \begin{cases}
        M \cdot \epsilon_j(v_j,b_j) & (j,v_j,b_j) \neq (i,v_i,b'_i) \\
        M \cdot \epsilon_i(v_i,b'_i) + \delta & \textnormal{otherwise}
    \end{cases} $$
    is dual feasible with objective value $0$, hence dual optimal. Contradiction to Proposition \ref{prop:dual-form}.
\end{proof}

Next, we shall aim to further simplify the dual LP (\ref{opt:generic-dual}). In particular, our goal now is to reduce the number of \emph{dual constraints} we need to consider. Each constraint (\ref{cons:eq}) of (\ref{opt:generic-dual}) corresponds, of course, to a vector of pure bidding strategies $\lambda$; and eliminating its corresponding constraint is equivalent to setting $\sigma(\lambda) = 0$ in the primal LP (\ref{opt:generic-primal}). 

\begin{definition}
  For a discretised auction $A_n$, suppose that $\sigma : \Lambda_n \rightarrow \mathbb{R}_+$. Then the \textbf{support} of $\sigma$ is the set of bidding strategies $S(\sigma) = \{ \lambda \in \Lambda_n \ | \ \sigma(\lambda) > 0\}$.
\end{definition}

From here onwards, we are not only looking for ways to restrict our attention to (approximate) BCCE supported in a restricted set of bidding strategies satisfying certain intuitive properties, but also to demonstrate the explicit form of dual solutions that \emph{prove} so. The latter necessitates us to make some assumptions on a given discretised auction $A_n$.

\begin{assumption}[Existence of zero type and bid]\label{asmptn:zero}
  There exists a zero type and a zero bid, moreover, if any buyer $i$ bids zero then their expected payment is zero. Formally, 
  $0 \in V_n \cap B_n$ and $\forall \ i\in [N], b_{-i} \in B_n^{N-1}, p_i(0,b_{-i}) = 0$.
\end{assumption} 

\begin{assumption}[Non-triviality of allocation \& payments]\label{asmptn:non-trivial}
  If buyer $i$ has a maximal positive bid, then they will win an item with positive probability and pay an amount strictly higher than zero,
  $ i \in \arg \max_{j \in [N]} b_j \Rightarrow x_i(b), p_i(b) > 0.$
\end{assumption}

\begin{assumption}[Monotonicity]\label{asmptn:monotone}
  The allocation probability and the expected payment of a buyer is weakly increasing in their bid. Furthermore, the expected payment is strictly increasing whenever a buyer's bid is maximal,
      $$\forall \ i \in [N], b_{-i} \in B_n^{N-1} \textnormal{ and } b,b' \in B_n, [b < b' \Rightarrow x_i(b,b_{-i}) \leq x_i(b',b_{-i}), p_i(b,b_{-i}) \leq p_i(b',b_{-i}) ], $$
      $$\forall \ i \in [N], b \in B_n^N, b' \in B_n, [ i \in \arg \max_{j \in N} b_j \textnormal{ and } b' > b \Rightarrow p_i(b,b_{-i}) < p_i(b',b_{-i}) ].$$
\end{assumption}

Assumption \ref{asmptn:zero} is intuitive, while Assumptions \ref{asmptn:non-trivial} and \ref{asmptn:monotone}, apply for a wider class of auctions in single-parameter settings (in the sense of \cite{myerson1984}); e.g. sponsored search auctions with generalised first-price payment rules. However, we remark that pairs $(x,p)$ which define DSIC mechanisms violate Assumptions \ref{asmptn:non-trivial} and \ref{asmptn:monotone}.

First, as a warm up, we show that without loss of generality we may restrict our attention to bidding strategies normalised such that each buyer $i$ of valuation $0$ bids $0$. As a consequence, in (\ref{opt:generic-primal}) we may without loss of generality restrict attention BCCE supported in \emph{normalised} bidding strategies, in which buyers of valuation $0$ bid $0$. 

\begin{proposition}\label{prop:zero-bid}
  Suppose that $0 \in V_n \cap B_n$, and $\lambda' \in \Lambda_n$ such that for some buyer $i$, $\lambda'_i(0) > 0$. Then for any primal solution $\sigma$ of (\ref{opt:generic-primal}), $\sigma(\lambda') = 0$. In other words, for any bidding strategy $\lambda$ in the support of a BCCE $\sigma$, $\lambda_i(0) = 0$. 
\end{proposition}

\begin{proof}
  Let $c_{n\lambda} = \delta_{\lambda \lambda'}$, and suppose without loss of generality that $i \in \arg\max_{j \in [N]} \lambda'_j(0)$. We proceed by constructing a solution to (\ref{opt:generic-dual}) with objective value $0$. Since the primal LP (\ref{opt:generic-primal}) is feasible with finite value, so too is its dual. Then note that, for any bidding profile $\lambda$ by Assumption \ref{asmptn:zero},
  \begin{align*}
  &\sum_{v_{-i} \in V_{-i}} F_{n(-i)}(v_{-i}) \cdot \left( 0 \cdot [x_i(0, \lambda_{-i}(v_{-i})) - x_i(\lambda(0,v_{-i}))] - p_i(0,\lambda_{-i}(v_{-i})) + p_i(\lambda(0,v_{-i})) \right) \\
  = & \sum_{v_{-i} \in V_{-i}} F_{n(-i)}(v_{-i}) \cdot p_i(\lambda(0,v_{-i})) \geq 0.
  \end{align*}
  Moreover, by Assumption \ref{asmptn:non-trivial} the inequality is strict for $\lambda'$. since buyer $i$ of type $0$ has a maximal bid whenever all other buyers are also of type $0$, which happens with positive probability. Then by setting $\gamma = 0$, $\epsilon_i(0,0)$ to be sufficiently large, and all other $\epsilon_j(\theta_j,b'_j) = 0$, we obtain a feasible dual solution to (\ref{opt:generic-dual}) with objective value $0$.
\end{proof}

Now, suppose that the BCCE of the discretised auction $A_n$ is uniquely given via a symmetric pure strategy Nash equilibrium, $\lambda^*$. Then by Proposition \ref{prop:dual-form} and by the symmetry of $\delta_{\lambda \lambda^*}$ in buyers' indices, we may assume that $\epsilon_i(v_i,b'_i) = \delta_{\lambda^*_i(v_i) b'_i}\epsilon(\theta_i)$ for some function $\epsilon : V_n \rightarrow \mathbb{R}_+$, while Proposition \ref{prop:zero-bid} tells us that $\epsilon(0)$ can be set to some sufficiently large value. However, we do not yet have an inkling of what functional form the dual solution should have. 

Our final result in this section addresses this, placing bounds on the rate of decrease of $\epsilon$ necessary to certify that the BCCE is supported in \emph{weakly increasing} strategies. We will prove the result for uniform probability distributions over $V_n$, as it will be conducive to our analysis in the following sections. For brevity in notation, given bidding strategies $\lambda$ we introduce a shorthand for \emph{monotonised} bidding strategies obtained from $\lambda$.

\newcommand{\mon}{\textnormal{mon} \ }

\begin{notation*}
  Given bidding strategies $\lambda$, $\mon\lambda$ is the weakly increasing bidding strategies obtained via permutations $p_i : V_n \rightarrow V_n$ such that $\mon\lambda_i(v_i) = \lambda_i(p_i(v_i))$ for any $i \in [N]$ and $v_i \in V_n$.
\end{notation*}

\begin{proposition}\label{prop:inc-bids}
  Let $A_n$ be a discretised auction such that $F_n(v_i) = 1/|V_n|$ for each $v_i \in V_n$. Consider a pure strategy profile $\lambda^*$ and the restricted dual LP for some fixed $\epsilon : V_n \rightarrow \mathbb{R}_+$, 
  \begin{align}
    \min_{\gamma} \gamma \textnormal{ s.t.} & \label{opt:restricted-dual} \\
    \gamma + \sum_{i \in [N], v \in \valset_n} \epsilon(v_i) \cdot \frac{1}{|V_n|^{N-1}} ( u_i(\lambda^*_i(v_i),\lambda_{-i}(v_{-i})|v_i) - u_i(\lambda(v)|v_i)) & \geq c_{n\lambda} \ \forall \ \lambda \in \Lambda_n. \tag{$\sigma(\lambda)$} \label{cons:eq-restrict}
  \end{align}
  Suppose that $c_{n(\mon\lambda)} = c_{n\lambda}$ for each bidding strategies $\lambda$ and for $v_i > 0$, $\epsilon(v_i)$ is strictly decreasing, and $\epsilon(v_i)v_i$ is non-increasing. Then the inequalities (\ref{cons:eq-restrict}) are redundant for each $\lambda$ such that $\lambda \neq \mon\lambda$. Equivalently (via complementary slackness), for the relaxed primal LP 
  \begin{align}
    \max_{\sigma \geq 0} \sum_{\lambda \in \profset_n} \left(c_{n\lambda} + \sum_{i \in [N], v \in \valset_n} \epsilon(v_i) \cdot \frac{1}{|V_n|^{N-1}} (u_i(\lambda(v)|v_i) - u_i(\lambda^*_i(v_i),\lambda_{-i}(v_{-i})|v_i))\right) \sigma(\lambda) & \textnormal{ s.t.} \label{opt:relaxed-primal} \\
    \sum_{\lambda \in \profset_n} \sigma(\lambda) & = 1 \tag{$\gamma$}
  \end{align}
  an optimal solution $\sigma$ has support in weakly increasing bidding strategies.
\end{proposition}

\begin{proof}
  The feasible region of (\ref{opt:relaxed-primal}) is the probability simplex $\Delta(\Lambda_n)$, thus any basic solution is equal to some pure bidding profile $\lambda$. Thus it is sufficient to show that for any $\lambda \in \Lambda_n$, if $\lambda \neq \mon\lambda$ then $\mon \lambda$ has strictly higher objective value. Consider any bidding strategies $\lambda \neq \mon\lambda$. Then there exists some buyer $i$ and valuations $v_i^+, v_i^- \in V_n$ such that $v_i^+ > v_i^-$ but $\lambda_i(v_i^-) > \lambda_i(v_i^+)$. Consider the bidding strategies $\lambda'$ obtained by swapping buyer $i$'s bids for types $v^\pm_i$, i.e.
  \begin{equation*}
    \lambda'_j(v_j) = \begin{cases}
      \lambda_j(v_j) & j \neq i \textnormal{ or } v_j \notin \{v^+_i, v^-_i\}, \\
      \lambda_i(v_i^-) & j = i, v_j = v^+_i, \\
      \lambda_i(v_i^+) & j = i, v_j = v^-_i.
    \end{cases} 
  \end{equation*}
  We will show that 
  \begin{equation}\label{eq:objective-disc} \tag{$\ast$} \sum_{i \in [N], v \in \valset_n} \epsilon(v_i) \cdot \frac{1}{|V_n|^{N-1}} (u_i(\lambda(v)|v_i) - u_i(\lambda^*_i(v_i),\lambda_{-i}(v_{-i})|v_i)) \equiv \sum_{i \in [N], v_i \in V_n} \mu_i(v_i|\lambda)\end{equation}
  strictly increases when we change $\lambda \rightarrow \lambda'$. As any permutation is obtained via compositions of pairwise swaps, we conclude that the (\ref{eq:objective-disc}) is strictly higher for $\mon\lambda$.

  To wit, by our assumption of a uniform prior distribution over $V_n$, note that for any $j \neq i$, or if $j = i$ then for any $v_i \neq v_i^\pm$, $\mu_j(v_j|\lambda) = \mu_j(v_j|\lambda')$; buyer $i$ makes the same expected bid distribution against any buyer $j \neq i$, so for such buyers $j$ their utilities for following $\sigma$ or deviating to bidding $\lambda^*_j$ do not change. Therefore, the change in (\ref{eq:objective-disc}) when we let $\lambda \rightarrow \lambda'$ equals, up to a multiplicative factor $\frac{1}{|V_n|^{N-1}}$, 
  \begin{align}
    & \quad\quad \mu_i(v^+_i|\lambda') + \mu_i(v^-_i|\lambda') - \mu_i(v^+_i|\lambda) - \mu_i(v^-_i|\lambda) \label{eq:diff} \tag{$\ast\ast$} \\
    & = \sum_{v_{-i} \in V_n^{N-1}} \Bigg[ \epsilon(v^+_i) \cdot \left( x_i(\lambda(v^-_i,v_{-i}))v^+_i - p_i(\lambda(v^-_i,v_{-i})) - x_i(\lambda^*_i(v^+_i),\lambda_{-i}(v_{-i}))v^+_i + p_i(\lambda^*_i(v^+_i),\lambda_{-i}(v_{-i})) \right)  \nonumber \\ 
    & + \epsilon_i(v^-_i) \cdot \left( x_i(\lambda(v^+_i,v_{-i}))v^-_i - p_i(\lambda(v^+_i,v_{-i})) - x_i(\lambda^*_i(v^-_i),\lambda(v_{-i}))v^-_i + p_i(\lambda^*_i(v^-_i),\lambda(v_{-i})) \right) \nonumber \\
    & - \epsilon_i(v^-_i) \cdot \left( x_i(\lambda(v^-_i,v_{-i}))v^-_i - p_i(\lambda(v^-_i,v_{-i})) - x_i(\lambda^*_i(v^-_i),\lambda_{-i}(v_{-i}))v^-_i + p_i(\lambda^*_i(v^-_i),\lambda_{-i}(v_{-i})) \right) \nonumber \\ 
    & - \epsilon_i(v^+_i) \cdot \left( x_i(\lambda(v^+_i,v_{-i}))v^+_i - p_i(\lambda(v^+_i,v_{-i})) - x_i(\lambda^*_i(v^+_i),\lambda_{-i}(v_{-i}))v^+_i + p_i(\lambda^*_i(v^+_i),\lambda_{-i}(v_{-i})) \right) \Bigg]. \nonumber
  \end{align}
  The $x_i(\lambda^*_i(v^\pm_i),\lambda_{-i}(v_{-i}))$ and $p_i(\lambda^*_i(v^\pm_i),\lambda_{-i}(v_{-i}))$ terms, associated with the expected utility for bidding according to $\lambda^*_i$, of course cancel -- as our modification $\lambda \rightarrow \lambda'$ does not change the bids of buyers other than $i$. To simplify this expression, we contract conditional expected allocation probabilities and expected payments into shorthand, $X_i(b^\pm) = \sum_{v_{-i} \in V_{-i}} x_i(\lambda(v^\mp,v_{-i}))$ and $
      P_i(b^\pm) = \sum_{v_{-i} \in V_{-i}} p_i(\lambda^*(v^\mp,v_{-i}))$. 
  We emphasise that here, $b^+$ corresponds to the higher bid $\lambda_i(v^-)$ and vice versa. Then (\ref{eq:diff}) can be further simplified,
  $$ = (\epsilon(v^+_i)v^+_i - \epsilon(v^-_i)v^-_i) \cdot (X_i(b^+) - X_i(b^-)) + (\epsilon(v^-_i) - \epsilon(v^+_i)) \cdot (P_i(b^+)-P_i(b^-)).$$
  By Proposition \ref{prop:zero-bid}, the higher bid $b^+$ is winning whenever all buyers have valuation $0$, and as a consequence by Assumption \ref{asmptn:monotone}, $P_i(b^+) - P_i(b^-) > 0$ and $X_i(b^+) - X_i(b^-) \geq 0$. Therefore, if $\epsilon(v_i)v_i$ is weakly increasing and $\epsilon(v_i)$ is strictly decreasing, then (\ref{eq:diff}) is strictly positive.
\end{proof}

\section{Bayesian CCE of Continuous Auctions}\label{sec:theory-cont}

Proposition \ref{prop:inc-bids} hints that, in our search for suitable dual solutions to (\ref{opt:generic-dual}), we might consider fixing $\epsilon(v_i) = K/v_i$ for some constant $K > 0$ when attempting to prove uniqueness of BCCE with support in weakly increasing bidding strategies. However, from our numerical results in Section \ref{sec:numerical-results}, we know that most discretised auctions do not admit a unique BCCE in pure strategies. Therefore, the results of Section \ref{sec:structure}, except for Proposition \ref{prop:zero-bid}, cannot be directly applied to obtain dual proofs of vanishing Wasserstein-$2$ distance bounds for families of discretisations of some continuous auction $A$.

To rectify this, we shall analyse the equivalent of the primal LP (\ref{opt:generic-primal}) for the \emph{continuous auction itself}. This results in an \emph{infinite LP}; intuitively, we obtain it by \emph{``simply''} exchanging sums in (\ref{opt:generic-primal}), (\ref{opt:generic-dual}) with integrals. However, working with infinite dimensions poses its own challenges in the specification of a well-posed pair of primal-dual problems. This is because not even weak duality is guaranteed to hold in general for an infinite dimensional pair of primal-dual LPs \cite{grinold1973duality}. Then to certify uniqueness, we need to not only ensure that the dual problem has a solution of value $0$, but also that such a solution does indeed certify a bound on the value of the primal LP.

We thus identify assumptions on the \emph{support of the BCCE} which address these issues; this corresponds to restricting (\ref{opt:generic-primal}) with inequalities $\sigma(\lambda) = 0$ for a set of disallowed bidding strategies. In particular, we first restrict attention to the case when the support consists of symmetric, differentiable and strictly increasing bidding strategies $\lambda$, such that $\lambda(0) = 0$ and the allowed bidding strategies are uniformly bounded with $\max_v \lambda(v) < M$. Then our primal-dual formulation provides proof that for both first-price auctions and all-pay auctions, the only such BCCE places probability $1$ on the canonical equilibrium bidding strategies. 

We then relax the assumption of strictly increasing bids, by leveraging the result of Proposition \ref{prop:inc-bids}. In this case, feasibility of our optimal dual solution will require strictly concave priors for the first-price auction, while uniqueness will keep holding for the BCCE of all-pay auctions under much weaker assumptions. Finally, we conclude this section by showing that the assumption of strictly increasing strategies is indeed necessary for uniqueness of the BCCE of continuous first-price auctions without concavity assumptions on the prior distribution, via examples of first-price auctions with convex priors where there exist BCCE which place probability $< 1$ on the canonical equilibrium.

\subsection{Continuous Bayesian CCE as Feasible Solutions to an Infinite LP}\label{sec:formalism}

\newcommand{\bilin}[2]{\langle #1, #2 \rangle}

An infinite linear program is an optimisation problem with a linear objective, and potentially both infinitely many\footnote{And perhaps uncountably so.} variables and constraints. It is often studied as a problem of optimisation over vector spaces over $\mathbb{R}$ \cite{anderson1987linear,shapiro2001duality}, which is the formalism we adapt and adopt in this paper; we refer the interested reader to \cite{kellerer1988measuretheoretic,clark2003infinite,martin2016slater} for alternative approaches. Whereas most of the literature is concerned with establishing conditions under which strong duality holds and/or approximation algorithms with guarantees exist, we work in a setting in which we first seek to ensure only well-posedness and weak duality\footnote{If the primal-dual framework is poorly specified, weak duality need not hold for infinite dimensional LPs \cite{grinold1973duality}.}. As the canonical equilibrium of the auction provides a feasible primal solution, strong duality then follows by construction of a matching dual solution.

We begin by a quick introduction; the canonical maximisation problem is given
\begin{align}
    \sup_{x \geq 0} \ \bilin{x}{c} & \textnormal{ subject to} \label{opt:inf-canon-max}\\
    Ax & \leq 0 \nonumber \\
    \bilin{x}{\Sigma} & = 1, \nonumber
\end{align}
where $x$ is an element of the real vector space $E$, and $c, \Sigma$ are elements of $E^*$, the dual vector space of $E$. We write $\bilin{\cdot}{\cdot}$ for the bilinear evaluation map between any dual pair of vector spaces. Furthermore, $A : E \rightarrow G^*$ is a linear function mapping elements of $E$ to the elements of $G^*$, the dual of some vector space $G$. Given this data, we write the dual minimisation problem,
\begin{align}
    \inf_{y \geq 0, \gamma} \ \gamma &\textnormal{ subject to} \label{opt:inf-canon-min}\\
    \gamma \Sigma + A^\dagger y  & \geq c. \nonumber
\end{align}
Here, $y$ is taken as an element of the real vector space $G$, $\gamma$ is some real number, and $A^\dagger : G \rightarrow E^*$ is the \emph{adjoint} of $A$, satisfying the equation
\begin{equation}\label{eq:adjoint-cond}
    \bilin{Ax}{y} = \bilin{x}{A^\dagger y}
\end{equation}
for any $x \in E$ and $y \in G$. In this setting, weak duality holds between the problems (\ref{opt:inf-canon-max}) and (\ref{opt:inf-canon-min}). That is to say, for any $x \in E$ feasible for (\ref{opt:inf-canon-max}) and $y \in G$ feasible for (\ref{opt:inf-canon-min}), 
$$
\bilin{x}{c} \leq \bilin{x}{\gamma \Sigma + A^\dagger y} = \gamma \bilin{x}{\Sigma} + \bilin{Ax}{y} \leq \gamma,$$
where the inequalities follow from primal/dual feasibility, and the equality is a consequence of bilinearity of $\bilin{\cdot}{\cdot}$ \& the adjoint condition (\ref{eq:adjoint-cond}). As a result, whenever both problems (\ref{opt:inf-canon-max}) and (\ref{opt:inf-canon-min}) are feasible, the value of (\ref{opt:inf-canon-max}) is at most the value of (\ref{opt:inf-canon-min}).

Recall that our goal in this paper has been to certify that the BCCE of discretised auctions are close to the canonical equilibrium, via using the definition of a BCCE for a finite normal-form game as a feasible solution to an LP and then appealing to the strong duality of linear programming. This motivates us to define a continuous BCCE \emph{with a given support} as the feasible solution to an infinite LP, by adapting the primal-dual pair of finite LPs (\ref{opt:generic-primal}) \& (\ref{opt:generic-dual}) to the infinite dimensional case. By the proof of Proposition \ref{prop:inc-bids} we will also want to work in quantile space or \emph{``inverse prior coordinates''} -- a buyer of \emph{type} $\theta \in [0,1]$ will have valuation $v(\theta) = F^{-1}(\theta)$, and equilibrium bid $\beta(\theta)$, where
\begin{equation}\label{def:eqil-bids}
  \beta(\theta) = \begin{cases}
    \frac{1}{\theta^{N-1}} \int_0^\theta (N-1) \theta'^{N-2} v(\theta') d\theta' & \textnormal{ for a first-price auction, and} \\
    \int_0^\theta (N-1) \theta'^{N-2} v(\theta') d\theta' & \textnormal{ for an all-pay auction.}
\end{cases}\end{equation}
We will assume that the prior distribution $F$ has a probability density function, and $v : [0,1] \rightarrow \mathbb{R}$ is continuous and strictly increasing on $[0,1]$, and infinitely smooth on $(0,1]$. Our definition of a continuous BCCE then:

\begin{definition}\label{def:cont-BCCE}
  Let $\Lambda \subseteq \times_{i \in [N]} B^{[0,1]}$, equipped with a topology under which 
  $$ \int_{[0,1]^{N-1}} d\theta_{-i} \cdot \left(u_i(b'_i,\lambda_{-i}(\theta_{-i})|v(\theta)) - u_i(\lambda(\theta)|v(\theta)) \right) $$
  is measurable over $(\lambda,\theta_i,b'_i) \in \Lambda \times [0,1] \times B$. Then for $\rho : [0,1] \rightarrow \mathbb{R}_+$, a \textbf{continuous $\rho$-Bayesian coarse correlated equilibrium} of a continuous auction $A = (N,\mathcal{V},F^N,\mathcal{B},u)$ \textbf{with support in} $\Lambda$ is a Borel probability measure $\sigma$ on $\Lambda$ such that for any player $i$, any type $\theta_i \in [0,1]$ and any bid $b'_i \in B$,
  $$\int_\Lambda d\sigma(\lambda) \cdot \left( U_i(b'_i,\lambda_{-i}|\theta_i) - U_i(\lambda|\theta_i) \right) \leq \rho(\theta_i),$$
  where we denote $U_i(\lambda|\theta_i) = \int_{[0,1]^{N-1}} d\theta_{-i} \cdot u_i(\lambda(\theta)|v(\theta_i))$ as shorthand.
\end{definition}

Note that we need to be careful about concerns of measurability in this setting. The requirement that $\sigma$ is Borel over $\Lambda$ addresses this partially, however, our main requirement is the well-posedness of the infinite LP 
\begin{align}
  \sup_{\sigma \geq 0} \int d\sigma(\lambda) \cdot c(\lambda) \textnormal{ subject to}& \label{opt:inf-LP-primal-gen}\\
  \int_\Lambda d\sigma(\lambda) & = 1 \tag{$\gamma$} \label{cons:inf-gamma} \\
  \int_\Lambda d\sigma(\lambda) \cdot \left( U_i(b'_i,\lambda_{-i}|\theta_i) - U_i(\lambda|\theta_i) \right) & \leq \rho(\theta_i) \tag{$\epsilon_i(\theta_i,b'_i)$} \ \forall \ i \in [N], \theta_i \in [0,1], b'_i \in B, \label{cons:inf-eq}
\end{align}
and whether we may bound its value via a solution to the infinite dual LP 
\begin{align}
  \inf_{\epsilon \geq 0, \gamma} \gamma + \sum_{i \in [N]} \int_{[0,1] \times B} d\epsilon_i(\theta_i,b'_i) \cdot \rho(\theta_i) \textnormal{ subject to} &  \label{opt:dual-inf-2}\\
  \gamma + \sum_{i \in [N]} \int_{[0,1] \times B} d\epsilon_i(\theta_i,b'_i) \cdot  \left( U_i(b'_i,\lambda_{-i}|\theta_i) - U_i(\lambda|\theta_i) \right) & \geq c(\lambda) \label{cons:inf-dual-gen} \tag{$\sigma(\lambda)$}
  \ \forall \ \lambda \in \Lambda. 
\end{align}
Weak duality in our setting then holds whenever the adjoint condition (\ref{eq:adjoint-cond}) does, which in our setting is equivalent to 
\begin{align}
    & \int_\Lambda d\sigma(\lambda) \cdot \sum_{i \in [N]} \int_{[0,1] \times B} d\epsilon_i(\theta_i,b'_i) \cdot  \left( U_i(b'_i,\lambda_{-i} | \theta_i) - U_i(\lambda|\theta_i) \right) \label{cond:fubini} \\
    = & \sum_{i \in [N]} \int_{[0,1] \times B} d\epsilon_i(\theta_i,b'_i) \cdot \int_\Lambda d\sigma(\lambda) \cdot \left( U_i(b'_i,\lambda_{-i} | \theta_i) - U_i(\lambda|\theta_i) \right). \nonumber
\end{align}
Ensuring (\ref{cond:fubini}) necessitates that each iterated integral is well-defined and that the Fubini-Tonelli theorem holds for the measure $\sigma$ on $\Lambda$ and the measure induced by the dual solution $\epsilon$ on $[0,1]$. This will require addressing on a case-by-case basis, conditional on the support $\Lambda$ of our continous BCCE.

The final question we address in this section is whether the extended infinite LP we study is an appropriate notion of BCCE for a continuous game of incomplete information, where the type space might be infinite. The literature on Bayesian (C)CE has up so far has overwhelmingly studied incomplete information games with \emph{finite} type space \cite{HST15,bergemann2016bce,Fujii2023,forges1993five,forges2023correlated}. However, we remark that consideration of infinite bidding and / or type spaces is not totally unprecedented in literature; \citet{FLN16} construct a CCE for the complete information first-price auction on an infinite bidding space, and \citet{bergemann2017first} consider the Bayes correlated equilibria (in the sense of \cite{bergemann2016bce}) for the first-price auction where both the valuation and the bidding spaces can be infinite. Moreover, even though equilibrium existence theorems may fail when the valuation and bidding spaces are infinite, we are guaranteed one continuous BCCE always exists: the probability distribution which assigns probability $1$ to the canonical symmetric equilibrium. Finally, we note that our notion of a continuous BCCE allows us to recover distance bounds for the BCCE of discretised auctions in Section \ref{sec:cont-to-disc}. Thus, our notion of continuous BCCE turns out to be the appropriate object to facilitate analysis of the BCCE of discretised auctions.

\subsection{Uniqueness of BCCE in Strictly Monotone Strategies}\label{sec:uniqueness-strict}

With our formalism defined Section \ref{sec:formalism}, we may proceed with our proofs of uniqueness, conditional on the support of a BCCE. We first consider restrictions on $\Lambda$, such that uniqueness of BCCE hold with minimal assumptions on the prior distribution.

\begin{assumption}\label{asmptn:sup}
  For the continuous auction $A$ and some $M > 0$, we restrict $\Lambda \subseteq \times_{i\in N} B^{[0,1]}$ to the set of bidding strategies $\lambda$ satisfying:
  \begin{enumerate}
    \item (Symmetry) For each buyer $i,j$, $\lambda_i = \lambda_j$.
    \item (Smoothness \& Strict Monotonicity) For each buyer $i$, $\lambda_i$ is strictly increasing and differentiable.
    \item (Normalisation) For each buyer $i$, $\lambda_i(0) = 0$.
    \item (Uniform Boundedness) For each buyer $i$, $\lambda_i < M$.  
  \end{enumerate}
  Given such a $\lambda$, by the symmetry assumption, we shall for notational convenience identify it with the bidding strategy of any single buyer. 
\end{assumption}

Assumption \ref{asmptn:sup}.(3) is justified by Proposition \ref{prop:zero-bid}, while Assumption \ref{asmptn:sup}.(4) holds whenever $B = [0,M]$ for some $M \in \mathbb{R}_+$, which we shall assume is the case. It is thus the symmetry, smoothness and strict monotonicity assumptions which are particularly restrictive. We remark that Algorithm \ref{alg:self-play} is a no-regret learning algorithm whose history of play provides a BCCE with symmetric support for the discretised auction; or namely, there exists learning algorithms whose output have symmetric support. It will turn out, however, that the assumption of strict monotonicity is crucial to ensuring uniqueness of BCCE without additional assumptions on the prior distribution.

It remains to precisely specify the ingredients of our primal-dual linear programs. First, keeping in mind the result of Proposition \ref{prop:dual-form}, we will relax the primal LP (\ref{opt:inf-LP-primal-gen}) by eliminating the constraints (\ref{cons:inf-eq}) whenever $b'_i \neq \beta(\theta_i)$. We then need to fix our primal objective $c(\lambda)$. For discretised auctions, we were able to use indicator functions to show uniqueness. However, in the continuous setting, we can always find differentiable and strictly increasing sequence $\lambda_n \rightarrow \beta$ such that the LHS of the constraints (\ref{cons:inf-dual-gen}) for $\lambda_n$ tend to $0$ as $n \rightarrow \infty$. The objective must then be carefully chosen as a positive measurable function $c$ on $\Lambda$ such that we can find a dual solution $\epsilon$, for which given any such sequence $\lambda_n$, 
$$c(\lambda_n) - \sum_{i \in [N]} \int_{0}^1 d\theta_i \cdot \epsilon_i(\theta_i) \cdot  \left( U_i(\beta,\lambda_{n(-i)} | \theta_i) - U_i(\lambda_n|\theta_i) \right) \rightarrow 0.$$
While any function $c : \lambda \mapsto c(\lambda)$ which is continuous with respect to the $\sup$-norm on $\profset$ will rectify this, for the moment we shall simply fix our objective 
$$ c(\lambda) = \sum_{i \in [N]} \int_0^1 d\theta_i \cdot \alpha(\theta_i|\lambda),$$
where
$$ \alpha(\theta_i|\lambda) = \theta_i^{N-2} \cdot \begin{cases}
    \frac{(\lambda_i(\theta_i) - \beta(\theta_i))^2}{2} & \lambda_i(\theta_i) \leq \beta(1) \\
    \frac{(\beta(1) - \beta(\theta_i))^2}{2} + (\beta(1)-\beta(\theta_i)) \cdot (\lambda_i(\theta_i) - \beta(1)) & \lambda_i(\theta_i) > \beta(1).
\end{cases}$$
Whenever the range of all $\lambda_i$ are contained within that of the canonical equilibrium bidding strategies $\beta$, the objective equals a weighted expectation of the Wasserstein-$2$ distance between $\lambda$ and $\beta$; and in the two buyer setting, it equals the usual Wasserstein-$2$ distance. The function $\alpha(\theta_i|\lambda)$ then needs to be extended linearly beyond $\beta(1)$ for technical reasons. Note that given any continuous $\lambda_i : [0,1] \rightarrow B$, $c(\lambda) = 0$ if and only if $\lambda_i = \beta$. Moreover, the objective $c(\lambda)$ and the constraints $\epsilon_i(\theta_i, \beta(\theta_i))$ are invariant under permulations $i \mapsto p(i)$. As a consequence, leveraging the symmetry in the resulting dual LP we shall be able to fix $\epsilon_i(\theta_i,\beta(\theta_i)) = \epsilon(\theta_i)$.

What remains is the matters of measurability to ensure that weak duality holds between the resulting primal and dual problems. We thus turn our attention to the identification of the relevant vector spaces for our problem. Fubini-Tonelli theorem necessitates that both $\sigma$ on $\Lambda$ and the measure on $[0,1]$ induced by $\epsilon$ are $\sigma$-finite. For the measure $\sigma$, we make the assumption that $\sigma$ is a Borel measure with respect to the sup-norm on $\Lambda$. Meanwhile, we will want to take $\epsilon$ to be the set of non-negative bounded continuous functions on $[0,1]$. But, while this works for the all-pay auction, for the first-price auction we will need a dual solution $\epsilon$ which diverges like $1/\theta_i$ as $\theta_i \downarrow 0$. Such a measure is not $\sigma$-finite; to address this, we shall absorb the factor $1/\theta_i$ into the constraints, and take $\epsilon$ to be a uniformly continuous function on $(0,1]$. 

In sum, for the first-price auction, our infinite primal LP is 
\begin{align}
  \sup_{\sigma \geq 0} \int d\sigma(\lambda) \cdot \sum_{i \in [N]} \int_0^1 d\theta_i \cdot \alpha(\theta_i|\lambda) \textnormal{ subject to}& \label{opt:inf-FP-strict}\\
  \int_\Lambda d\sigma(\lambda) & = 1 \tag{$\gamma$} \label{cons:inf-gamma-fps} \\
  \int_\Lambda d\sigma(\lambda) \cdot \frac{1}{\theta_i}\left( U_i(\beta,\lambda_{-i}|\theta_i) - U_i(\lambda|\theta_i) \right) & \leq 0 \tag{$\epsilon_i(\theta_i)$} \ \forall \ i \in [N], \theta_i \in (0,1], \label{cons:inf-eq-fps}
\end{align}
with a dual infinite LP 
\begin{align}
  \inf_{\epsilon \geq 0, \gamma} \gamma \textnormal{ subject to} &  \label{opt:dual-inf-FP-strict}\\
  \gamma + \sum_{i \in [N]} \int_0^1 d\theta_i \cdot \epsilon(\theta_i) \cdot \frac{1}{\theta_i} \left( U_i(\beta,\lambda_{-i}|\theta_i) - U_i(\lambda|\theta_i) \right) & \geq \sum_{i \in [N]} \alpha(\theta_i|\lambda) \label{cons:inf-dual-gen-fps} \tag{$\sigma(\lambda)$}
  \ \forall \ \lambda \in \Lambda. 
\end{align}
Meanwhile, for the all-pay auction, we consider the infinite primal LP 
\begin{align}
  \sup_{\sigma \geq 0} \int d\sigma(\lambda) \cdot \sum_{i \in [N]} \int_0^1 d\theta_i \cdot \alpha(\theta_i|\lambda) \textnormal{ subject to}& \label{opt:inf-AP-strict}\\
  \int_\Lambda d\sigma(\lambda) & = 1 \tag{$\gamma$} \label{cons:inf-gamma-aps} \\
  \int_\Lambda d\sigma(\lambda) \cdot \left( U_i(\beta,\lambda_{-i}|\theta_i) - U_i(\lambda|\theta_i) \right) & \leq 0 \tag{$\epsilon_i(\theta_i)$} \ \forall \ i \in [N], \theta_i \in [0,1], \label{cons:inf-eq-aps}
\end{align}
with a dual infinite LP
\begin{align}
  \inf_{\epsilon \geq 0, \gamma} \gamma \textnormal{ subject to} &  \label{opt:dual-inf-AP-strict}\\
  \gamma + \sum_{i \in [N]} \int_0^1 d\theta_i \cdot \epsilon(\theta_i) \left( U_i(\beta,\lambda_{-i}|\theta_i) - U_i(\lambda|\theta_i) \right) & \geq \sum_{i \in [N]} \alpha(\theta_i|\lambda) \label{cons:inf-dual-gen-aps} \tag{$\sigma(\lambda)$}
  \ \forall \ \lambda \in \Lambda. 
\end{align}

We would like to now show that weak duality holds for the primal-dual pairs (\ref{opt:inf-FP-strict})-(\ref{opt:dual-inf-FP-strict}) and (\ref{opt:inf-AP-strict})-(\ref{opt:dual-inf-AP-strict}). First, by Assumption \ref{asmptn:sup}, for any $\lambda \in \Lambda$, if $A$ is a first-price auction we have 
$$
        U_i(\beta, \lambda_{-i}|\theta_i) = P_\beta(\lambda,\theta_i)^{N-1}(v(\theta_i) - \beta(\theta_i) )  \textnormal{ and } U_i(\lambda|\theta_i) = \theta_i^{N-1} (v(\theta_i) - \lambda(\theta_i)),
$$
whereas if $A$ is an all-pay auction,
$$
        U_i(\beta, \lambda_{-i}|\theta_i) = P_\beta(\lambda,\theta_i)^{N-1}v(\theta_i) - \beta(\theta_i)  \textnormal{ and } U_i(\lambda|\theta_i) = \theta_i^{N-1} v(\theta_i) - \lambda(\theta_i);
$$
where we let
$$ P_\beta(\lambda,\theta) = \begin{cases}
  \lambda^{-1}\beta(\theta_i) & \beta(\theta_i) \leq \lambda(1), \\
  1 & \textnormal{otherwise}. \end{cases}$$
We first prove that the difference $U_i(\beta, \lambda_{-i}|\theta_i) - U_i(\lambda|\theta_i)$ is continuous with respect to the product topology on $\Lambda \times [0,1]$, where the topology on $\Lambda$ is defined via the $\sup$-norm while the topology of $[0,1]$ is the standard open ball topology induced by the usual Euclidean norm. As the sums and products of continuous functions are continuous, and $\theta_i, \lambda(\theta_i), v(\theta_i),$ and $\beta(\theta_i)$ are so by assumption, we only need to show the continuity of $P_\beta(\lambda,\theta)$. As a consequence, we infer that each of the iterated integrals in (\ref{cond:fubini}) makes sense.

\begin{lemma}\label{lem:measurability}
  Let $\lambda \in \Lambda$ and $\theta \in [0,1]$. Then for any $\varepsilon > 0$, there exists $\delta_\lambda, \delta_\theta > 0$ such that for any $\lambda' \in \Lambda$ and $\theta' \in [0,1]$, if $\| \lambda-\lambda'\|_\infty \equiv \sup_{\hat\theta \in [0,1]} |\lambda(\hat\theta) - \lambda'(\hat\theta)| < \delta_\lambda$ and $|\theta -\theta'| < \delta_\theta$, then $|P_\beta(\lambda,\theta) - P_\beta(\lambda',\theta')| < \varepsilon$.
\end{lemma}

\begin{proof}
  Since $\lambda$ is continuous and strictly increasing on $[0,1]$, $P_\beta(\lambda,\theta)$ is uniformly continuous for $\theta \in [0,1]$. Therefore, there exists $\delta_\theta > 0$ such that for any $\theta' \in [0,1]$, if $|\theta - \theta'| < \delta_\theta$ then $|P_\beta(\lambda,\theta) - P_\beta(\lambda,\theta')| < \varepsilon / 2$. Now, note that as $\lambda$ is a fixed continuous and strictly increasing function, the function $g(\theta) = \lambda(\theta + \varepsilon/2) - \lambda(\theta)$ defined for $\theta \in [0,1-\varepsilon/2]$ is continuous and strictly positive everywhere. As a result, it attains a minimum value $m > 0$ on the interval $[0,1-\varepsilon/2]$.

  Now, suppose that $\lambda^{-1}\beta(\theta) \in [\varepsilon/2,1-\varepsilon/2]$. In this case, if $\|\lambda' - \lambda\|_\infty < \delta_\lambda$, we have
$$
      \lambda'(\lambda^{-1}\beta(\theta) + \varepsilon/2) + \delta_\lambda \geq \lambda(\lambda^{-1}\beta(\theta) + \varepsilon/2) \geq \beta(\theta) + m.
$$
  Here, the first inequality follows immediately from the $\sup$-norm distance between $\lambda'$ and $\lambda$, while the second inequality is from the definition of $g$. A similar line of argument on $\lambda^{-1}\beta(\theta)$ provides 
$$
      \lambda'(\lambda^{-1}\beta(\theta) - \varepsilon/2) - \delta_\lambda \leq \lambda(\lambda^{-1}\beta(\theta) - \varepsilon/2) \leq \beta(\theta) - m.
$$    
  Then, picking $\delta_\lambda < m$, adding/subtracting as suitable $\delta_\lambda$ and applying $\lambda'^{-1}$ to the above inequalities, and using the fact that $\lambda'^{-1}$ is increasing, we see that 
  $$ \lambda^{-1}\beta(\theta) - \varepsilon/2 \leq \lambda'^{-1}(\beta(\theta) - (m-\delta_\lambda) ) \leq \lambda'^{-1}\beta(\theta) \leq \lambda'^{-1}(\beta(\theta) + (m-\delta_\lambda) ) \leq \lambda^{-1}\beta(\theta) + \varepsilon/2,$$
  i.e. $\lambda'^{-1}\beta(\theta) \in [\lambda'^{-1}\beta(\theta) - \varepsilon/2, \lambda'^{-1}\beta(\theta) + \varepsilon/2]$. The cases $\lambda^{-1}\beta(\theta) < \varepsilon / 2$ or $> 1- \varepsilon/2$ as well as the case $P_\beta(\lambda,\theta) = 1$ proceed via analogous arguments. To finalise the proof, note that the bound for $\delta_\lambda$ is independent of $\theta$. Then, for any such $\lambda'$ and $\theta'$, we have 
  $$|P_\beta(\lambda',\theta') - P_\beta(\lambda,\theta)| \leq |P_\beta(\lambda',\theta') - P_\beta(\lambda,\theta')| + |P_\beta(\lambda,\theta') - P_\beta(\lambda,\theta)| < \varepsilon/2 + \varepsilon/2 = \varepsilon.$$
\end{proof}

As a consequence, given $\epsilon$ and $\sigma$, $U_i(\beta, \lambda_{-i}|\theta_i) - U_i(\lambda|\theta_i)$ is measurable with respect to the product measure $d(\sigma \times \epsilon)$ on $\Lambda \times [0,1]$, and so too is $1/\theta_i \cdot [U_i(\beta, \lambda_{-i}|\theta_i) - U_i(\lambda|\theta_i)]$ by Lusin's theorem. With this, we are ready to show the sufficient condition for the Fubini-Tonelli theorem.

\begin{proposition}\label{prop:weak-duality-strict}
  Let $\sigma$ be a Borel measure on $\Lambda$, and $\epsilon : [0,1] \rightarrow \mathbb{R}_+$ a continuous function. Then for a first-price auction, 
  $$ \int_{\profset \times [0,1]} d(\sigma \times \theta_i) \cdot \epsilon(\theta_i) \cdot \frac{1}{\theta_i} \left| U_i(\beta,\lambda_{-i} | \theta_i) - U_i(\lambda|\theta_i) \right| < \infty,$$
  whereas for an all-pay auction, 
  $$ \int_{\profset \times [0,1]} d(\sigma \times \theta_i) \cdot \epsilon(\theta_i) \cdot \left| U_i(\beta,\lambda_{-i} | \theta_i) - U_i(\lambda|\theta_i) \right| < \infty.$$
  As a consequence, (\ref{cond:fubini}) holds, hence weak duality too for the primal-dual pairs (\ref{opt:inf-FP-strict})-(\ref{opt:dual-inf-FP-strict}) and (\ref{opt:inf-AP-strict})-(\ref{opt:dual-inf-AP-strict}).
\end{proposition}

\begin{proof}
  For any fixed $\lambda \in \Lambda$ and $\theta_i \in [0,1]$, we have 
  \begin{align*}
    \frac{1}{\theta_i} \cdot \left| U_i(\beta,\lambda_{-i} | \theta_i) - U_i(\lambda|\theta_i) \right| & \leq \frac{1}{\theta_i} \cdot \left| U_i(\beta,\lambda_{-i} | \theta_i) \right| + \frac{1}{\theta_i} \cdot \left| U_i(\lambda|\theta_i) \right| \\
    & = \frac{1}{\theta_i} \cdot \left| P_\beta(\lambda,\theta_i)^{N-1} \cdot (v(\theta_i) - \beta(\theta_i)) \right| + \frac{1}{\theta_i} \cdot \left| \theta_i^{N-1} \cdot (v(\theta_i)-\lambda(\theta_i)) \right| \\
    & = \frac{1}{\theta_i} \cdot \left| P_\beta(\lambda,\theta_i)^{N-1} \cdot \frac{\theta_i \beta'(\theta_i)}{N-1} \right| + \frac{1}{\theta_i} \cdot \left| \theta_i^{N-1} \cdot (v(\theta_i)-\lambda(\theta_i)) \right| \\
    & \leq \frac{\beta'(\theta_i)}{N-1} + \theta_i^{N-2} \max \{v(\theta_i), \lambda(\theta_i)\} \leq \frac{\beta'(\theta_i)}{N-1} + \theta_i^{N-2} \max \{v(\theta_i), M\}.
  \end{align*}
  The first inequality here is simply the triangle inequality, while the second equality holds by (\ref{def:eqil-bids}). The final inequality in turn follows from the uniform boundedness assumption on $\Lambda$. Now, $[0,1]$ with the measure induced by $\epsilon(\theta_i)$ and $\Lambda$ with the Borel probability measure $\sigma$ are both $\sigma$-finite measurable spaces, and the interated integral is bounded above as
  \begin{align*}
  \int_\profset d\sigma(\lambda) \int_0^1 d\theta_i \cdot \epsilon(\theta_i) \cdot \frac{\left| U_i(\beta,\lambda_{-i} | \theta_i) - U_i(\lambda|\theta_i) \right|}{\theta_i} &
  \leq \int_\profset d\sigma(\lambda) \int_0^1 d\theta_i \cdot \epsilon(\theta_i) \cdot \left( \frac{\beta'(\theta_i)}{N-1} + v(\theta_i) \right) \\
  & \leq  \int_\profset d\sigma(\lambda) \int_0^1 d\theta_i \cdot M' \cdot \left( \frac{\beta'(\theta_i)}{N-1} + v(\theta_i) \right) \\
  & =  M' \cdot \left(\beta(1) + \int_0^1 d\theta_i \cdot \max \{v(\theta_i), M\} \right).
  \end{align*}
  Here, the first inequality holds by the bound we derived, the second holds as $\epsilon(\theta_i)$ is a continuous positive function on the compact interval $[0,1]$ (hence bounded). Since by assumption $v(\theta_i)$ is also bounded, we conclude that the integral is finite. Then by the Fubini-Tonelli theorem, the order of integration may be exchanged, and the integral equals the value of the integral evaluated over the associated product measure.
\end{proof}

As a corollary, we see that feasible dual solutions of value $0$ for the dual problems (\ref{opt:dual-inf-FP-strict}) and (\ref{opt:dual-inf-AP-strict}) respectively provide proofs that any continuous BCCE of the first-price and all-pay auctions with support in strategies satisfying Assumption \ref{asmptn:sup} places probability $1$ on the canonical symmetric equilibrium bidding strategies of the auction.

\begin{corollary}\label{cor:weak-dual-consequence}
  Let $\gamma = 0$, $\epsilon : (0,1] \rightarrow \mathbb{R}_+$ be a solution of (\ref{opt:dual-inf-FP-strict}). Then for any continuous BCCE of the first-price auction with support in strategies satisfying Assumption \ref{asmptn:sup}, for any $\lambda \neq \beta$, there exists an open set $L \subseteq \Lambda$ containing $\lambda$ such that $\sigma(L) = 0$. An analogous statement also holds for the all-pay auction.
\end{corollary}

\begin{proof}
  Since $\lambda \neq \beta$, $\| \lambda - \beta\|_\infty = \delta$ for some $\delta > 0$. Let $L = \{\lambda' \in \lambda \ | \ \| \lambda - \lambda' \|_\infty < \delta/3 \}$. Now, $\lambda$ and $\beta$ are both continuous, hence the sup-norm bound is attained; i.e. there exists $\hat\theta \in [0,1]$ such that $|\lambda(\hat\theta) - \beta(\hat\theta)| = \delta$. By the continuity of $|\lambda(\cdot) - \beta(\cdot)|$, for some closed interval $I \ni \hat\theta$, $|\lambda(\theta) - \beta(\theta)| \geq 2\delta/3$ for any $\theta \in I$. As a consequence, for any $\lambda' \in L$ and any $\theta \in I$, $| \lambda'(\theta) - \beta(\theta)| > \delta/3$. This implies that for some $\Delta > 0$, $c(\lambda') > \Delta$ for any $\lambda' \in L$. In particular, $\int_\Lambda d\sigma(\lambda) \cdot c(\lambda) \geq \sigma(L) \cdot \Delta$. However, by weak duality, $\int_\Lambda d\sigma(\lambda) \cdot c(\lambda) \leq \gamma = 0$. Therefore, we conclude that $\sigma(L) = 0$.
\end{proof}

All that remains for our uniqueness results in this section is thus the construction of dual solutions for (\ref{opt:dual-inf-FP-strict}) and (\ref{opt:dual-inf-AP-strict}). First, we note that without loss of generality we may fix for the first-price auction 
\begin{equation}\label{eqn:cont-gamma}
\gamma = \sup_{\lambda \in \Lambda} \sum_{i \in N} \int_0^1 d\theta_i \cdot \left( \alpha(\theta_i|\lambda) - \epsilon(\theta_i) \cdot \frac{1}{\theta_i} \left( U_i(\beta,\lambda_{-i} | \theta_i) - U_i(\lambda|\theta_i) \right) \right).
\end{equation}
From here our strategy is to show -- via variational calculus -- that given the dual solution, for any bidding profile $\lambda \in \Lambda$, the right-hand side of (\ref{eqn:cont-gamma}) increases when we let $\lambda_i \leftarrow \alpha \lambda_i + (1-\alpha) \beta$ for any $\alpha > 0$.

\begin{theorem}\label{thm:thm-fp-uniq}
  Suppose for a first-price auction that $\Lambda$ is a uniformly bounded set of symmetric, strictly increasing, differentiable \& normalised bidding strategies containing the equilibrium bidding strategies $\beta$. Then $\epsilon(\theta_i) = \beta(1)-\beta(\theta_i)$ is a solution to (\ref{opt:dual-inf-FP-strict}) of value $0$. As a consequence, if a Borel measure $\sigma$ on $\Lambda$ is a continuous BCCE, it places probability $1$ on the equilibrium bidding strategies~$\beta$.
\end{theorem}

\begin{proof}
  The final part of the theorem statement follows from Corollary \ref{cor:weak-dual-consequence}, so we show the construction of the dual solution. For $0 \leq \underline{\theta} < \overline{\theta} \leq 1$, denote by $\gamma(\lambda|\underline{\theta},\overline{\theta})$,
  $$ \gamma(\lambda|\underline{\theta},\overline{\theta}) = \int_{\underline{\theta}}^{\overline{\theta}} d\theta_1 \cdot \left( \alpha(\theta_1|\lambda) - \epsilon(\theta_1) \cdot \frac{1}{\theta_1} \cdot  \left( U_1(\beta,\lambda_{-1} | \theta_1) - U_1(\lambda|\theta_1) \right) \right). $$
  We want to, for some dual solution $\epsilon$, show that $\gamma(\lambda|0,1)$ is maximised when $\lambda_i = \beta$ for any buyer $i$. Towards this end, it is sufficient to show that $\beta - \lambda$ is a weak ascent direction.\footnote{We actually show a stronger condition, that we may find a dual solution $\epsilon$ such that $\gamma(\lambda|0,1) = 0$ for any $\lambda \in \Lambda$. Thus, we actually simultaneously construct our primal objective here.}

  Divide the interval $[0,1]$ into subintervals, $[0,1] = \cup_{k} [\underline{\theta}_k,\overline{\theta}_k]$, such that one of the following hold for any $k$:
  \begin{enumerate}
    \item $\overline{\theta}_k < 1$, or
    \item $\overline{\theta}_k = 1$ and $\beta(\theta) \geq \lambda(\theta)$ for any $\theta \in [\underline{\theta},1]$, or
    \item $\overline{\theta}_k = 1$ and $\beta(\theta) \leq \lambda(\theta)$ for any $\theta \in [\underline{\theta},1]$,
  \end{enumerate}
  where for each subinterval $\beta(\underline{\theta}_k) = \lambda(\underline{\theta}_k)$, and also $\beta(\overline{\theta}_k) = \lambda(\overline{\theta}_k)$ so long as $\overline{\theta}_k \neq 1$. We remark that we can pick $k$ to be finite.

  To establish that $\beta - \lambda$ is a ascent direction, proceed by case analysis. For a fixed interval $k$, first suppose that $\overline{\theta}_k < 1$. Taking the first-order variation of $\gamma(\lambda|\underline{\theta}_k,\overline{\theta}_k)$, we get
  \begin{align*}
      & \delta\gamma(\lambda | \underline{\theta}_k , \overline{\theta}_k ) = \int_{\underline{\theta}_k}^{\overline{\theta}_k} d\theta \cdot \Bigg[\theta^{N-2}\cdot (\lambda(\theta)-\beta(\theta))\cdot  \delta\lambda(\theta) \ldots \\ &  + \epsilon(\theta) \cdot \frac{(N-1)\lambda^{-1}\beta(\theta)^{N-2}}{\lambda' \lambda^{-1} \beta(\theta) \cdot \theta}\cdot  (v(\theta) - \beta(\theta))\cdot  \delta\lambda(\lambda^{-1}\beta(\theta)) - \epsilon(\theta) \cdot \theta^{N-2} \cdot \delta\lambda(\theta) \Bigg].
  \end{align*}
  We remark that here, $\lambda'(\theta) = d\lambda(\theta)/d\theta$, and for notational simplicity we omit parentheses; e.g. $\lambda'\lambda^{-1}\beta(\theta) = \lambda'(\lambda^{-1}(\beta(\theta)))$ and $\lambda^{-1}\beta(\theta)^{N-2} = \lambda^{-1}(\beta(\theta))^{N-2}$. We want all term of the integrand to be multiplied by $\delta\lambda(\theta)$, the pointwise variation in $\lambda$. Therefore, for the term 
  $$ \int_{\underline{\theta}_k}^{\overline{\theta}_k} d\theta \cdot \epsilon(\theta) \cdot \frac{(N-1)\lambda^{-1}\beta(\theta)^{N-2}}{\lambda' \lambda^{-1} \beta(\theta) \cdot \theta} \cdot (v(\theta) - \beta(\theta)) \cdot \delta\lambda(\lambda^{-1}\beta(\theta)) $$ 
  we make a change of variables,
  \begin{align*}
      \theta' & = \lambda^{-1} \beta (\theta), \\
      d\theta' & = \frac{\beta'(\theta)}{\lambda'\lambda^{-1}\beta(\theta)} d\theta.
  \end{align*}
  By the assumption, we have $\underline{\theta}_k = \lambda^{-1}\beta(\underline{\theta}_k)$ and $\overline{\theta}_k = \lambda^{-1}\beta(\overline{\theta}_k)$. Moreover, for the first price auction, we have 
  $$ \theta \beta'(\theta) = (N-1) \cdot (v(\theta)-\beta(\theta)).$$
  We thus have, after making the change of variables and relabelling $\theta' \leftarrow \theta$,
  \begin{align*}
      \delta\gamma(\lambda | \underline{\theta}_k , \overline{\theta}_k ) = \int_{\underline{\theta}_k}^{\overline{\theta}_k} d\theta \cdot \delta\lambda(\theta) \cdot \left[\theta^{N-2}\cdot(\lambda(\theta)-\beta(\theta)) + \epsilon\beta^{-1}\lambda(\theta) \cdot \theta^{N-2} - \epsilon(\theta) \cdot \theta^{N-2} \right].
  \end{align*}
  We want to pick $\epsilon(\theta)$ to be strictly decreasing. This can be achieved by setting $\epsilon(\theta) = K - L \cdot \beta(\theta)$ for some constants $K,L$ to be determined later. Then note that our choice of $\epsilon(\theta)$ satisfies $\epsilon\beta^{-1}\lambda(\theta) = K-L\lambda(\theta)$.
  Hence, $\epsilon\beta^{-1}\lambda(\theta)- \epsilon(\theta) = L \cdot (\beta(\theta) - \lambda(\theta))$.
  As a consequence, the variation in direction $\beta - \lambda$ equals,
  \begin{align*}
     & \int_{\underline{\theta}_k}^{\overline{\theta}_k} d\theta \cdot (\beta(\theta)-\lambda(\theta)) \cdot \left[\theta^{N-2}\cdot(\lambda(\theta)-\beta(\theta)) + \epsilon\beta^{-1}\lambda(\theta) \cdot \theta^{N-2} - \epsilon(\theta) \cdot \theta^{N-2} \right] \\
      = \ & \int_{\underline{\theta}_k}^{\overline{\theta}_k} d\theta \cdot \theta^{N-2} \cdot \left[-(\lambda(\theta)-\beta(\theta))^2 + L \cdot (\lambda(\theta)-\beta(\theta))^2 \right] \\
      = \ & \int_{\underline{\theta}_k}^{\overline{\theta}_k} d\theta \cdot \theta^{N-2} \cdot (\lambda(\theta) - \beta(\theta))^2 \cdot \left[ -1 + L \right].
  \end{align*}
  Letting $L \geq 1$, we thus have the desired result. 

  Now suppose instead that $\overline{\theta}_k = 1$, and moreover, $\beta(1) \neq \lambda(1)$. We have two cases to consider; first suppose that $\lambda(1) > \beta(1)$ (and hence, by assumption, $\lambda(\theta) \geq \beta(\theta)$ for any $\theta > \underline{\theta}_k$). 
  In this case, we have $\lambda^{-1}\beta(1) < 1$, and 
  \begin{align*}
      \delta\gamma(\lambda | \underline{\theta}_k , 1 ) & = \int_{\underline{\theta}_k}^{1} d\theta \cdot \Bigg[\frac{\delta \alpha(\theta|\lambda)}{\delta\lambda} \cdot  \delta\lambda(\theta) \ldots \\ & + \epsilon(\theta) \cdot \frac{(N-1)\lambda^{-1}\beta(\theta)^{N-2}}{\lambda' \lambda^{-1} \beta(\theta) \cdot \theta} \cdot (v(\theta) - \beta(\theta)) \cdot \delta\lambda(\lambda^{-1}\beta(\theta)) - \epsilon(\theta) \cdot \theta^{N-2} \cdot \delta\lambda(\theta) \Bigg] \\
      & = \int_{\lambda^{-1}\beta(1)}^{1} d\theta \cdot \delta\lambda(\theta) \cdot \left[\theta^{N-2} \cdot  (\beta(1) - \beta(\theta))  - \epsilon(\theta) \cdot \theta^{N-2} \right] \ldots \\
      & + \int_{\underline{\theta}_k}^{\lambda^{-1}\beta(1)} d\theta \cdot \delta\lambda(\theta) \cdot \left[\theta^{N-2}\cdot (\lambda(\theta)-\beta(\theta)) + \epsilon\beta^{-1}\lambda(\theta) \cdot \theta^{N-2} - \epsilon(\theta) \cdot \theta^{N-2} \right].
  \end{align*}
  Note now that under our assumptions, $\beta(\theta) - \lambda(\theta) < 0$ for any $\theta \in (\underline{\theta}_k, \overline{\theta}_k)$. When we consider the variation in direction $\beta-\lambda$, i.e. when we let $\delta\lambda(\theta) = \beta(\theta) - \lambda(\theta)$, the second term is non-negative under our previous assumption on $L$; i.e. whenever $L \geq 1$. In turn, for the first term, we have
  \begin{align*}
      & \int_{\lambda^{-1}\beta(1)}^{1} d\theta \cdot (\beta(\theta) - \lambda(\theta)) \cdot \left[\theta^{N-2} \cdot (\beta(1) - \beta(\theta))  - \epsilon(\theta) \cdot \theta^{N-2} \right] \\  =\ &  \int_{\lambda^{-1}\beta(1)}^{1} d\theta \cdot \delta\lambda(\theta) \cdot \theta^{N-2} \cdot (\lambda(\theta) - \beta(\theta)) \cdot \left[ K - L \beta(\theta) - \beta(1) + \beta(\theta) \right] \\
       =\ &  \int_{\lambda^{-1}\beta(1)}^{1} d\theta \cdot \delta\lambda(\theta) \cdot \theta^{N-2} \cdot (\lambda(\theta) - \beta(\theta)) \cdot \left[ (K-\beta(1)) + (\beta(\theta)-L\beta(\theta)) \right].
  \end{align*}
  As a result, $K = L \cdot \beta(1)$ also ensures that the first term is non-negative.

  Finally suppose that $\lambda(1) < \beta(1)$. In this case, for $\theta > \beta^{-1}\lambda(1)$, a buyer wins with probability $1$ after deviating to $\beta(\theta)$, and we have
  \begin{align*}
      \gamma(\lambda|\underline{\theta}_k, 1)
      & = \int_{\underline{\theta}_k}^1 d\theta \cdot \left( \alpha(\theta|\lambda) + \epsilon(\theta) \cdot \theta^{N-2} \cdot (v(\theta) - \lambda(\theta)) \right) \\ & - \int_{\underline{\theta}_k}^{\beta^{-1}\lambda(1)} d\theta \cdot \frac{\epsilon(\theta)}{\theta} \cdot \lambda^{-1}\beta(\theta)^{N-1} \cdot (v(\theta) - \beta(\theta) ) \\ & - \int_{\beta^{-1}\lambda(1)}^1 d\theta \cdot \frac{\epsilon(\theta)}{\theta} \cdot (v(\theta) - \beta(\theta))
  \end{align*}
  Again, considering its first-order variation, we get
  \begin{align*}
      \delta\gamma(\lambda | \underline{\theta}_k , 1 ) & = \int_{\underline{\theta}_k}^1 d\theta \cdot \left(\theta^{N-2} \cdot (\lambda(\theta)-\beta(\theta)) \cdot \delta\lambda(\theta) - \epsilon(\theta) \cdot \theta^{N-2} \cdot  \delta\lambda(\theta) \right) \\
      & + \int_{\underline{\theta}_k}^{\beta^{-1}\lambda(1)} d\theta \cdot \epsilon(\theta) \cdot \frac{(N-1) \lambda^{-1}\beta(\theta)^{N-2}}{\lambda' \lambda^{-1}\beta(\theta) \cdot \theta}\cdot  (v(\theta) - \beta(\theta)) \cdot  \delta\lambda( \lambda^{-1} \beta(\theta) ) \\
      & - \frac{1}{\beta' \beta^{-1} \lambda(1) \cdot\theta} \cdot \frac{\epsilon\beta^{-1}\lambda(1)}{\beta^{-1}\lambda(1)} \cdot (v\beta^{-1}\lambda(1) - \lambda(1)) \cdot \delta\lambda(1) \\
      & + \frac{1}{\beta' \beta^{-1} \lambda(1) \cdot \theta} \cdot  \frac{\epsilon\beta^{-1}\lambda(1)}{\beta^{-1}\lambda(1)} \cdot (v\beta^{-1}\lambda(1) - \lambda(1)) \cdot \delta\lambda(1).
  \end{align*}
  The latter two terms of course cancel out, and after the change of variables as before, we get
  \begin{align*}
      \delta\gamma(\lambda | \underline{\theta}_k , 1 ) = \int_{\underline{\theta}_k}^1 d\theta \cdot \delta\lambda(\theta) \cdot \left(\theta^{N-2} \cdot (\lambda(\theta)-\beta(\theta))  - \epsilon(\theta) \cdot \theta^{N-2}  
       + \epsilon\beta^{-1}\lambda(\theta) \cdot \theta^{N-2} \right) .
  \end{align*}
  We thus have a reduction from this case to the first case we considered.
\end{proof}

For the all-pay auction, a tight dual solution for (\ref{opt:dual-inf-AP-strict}) can be constructed through identical arguments with the proof of Theorem \ref{thm:thm-fp-uniq}. As a consequence, the uniqueness result also holds for the all-pay auction; we defer its proof to the appendix.

\begin{theorem}\label{thm:thm-ap-uniq}
  Suppose for an all-pay auction that $\Lambda$ is a uniformly bounded set of symmetric, strictly increasing, differentiable \& normalised bidding strategies containing the equilibrium bidding strategies $\beta$. Then $\epsilon(\theta_i) = \theta^{N-2}\cdot (\beta(1)-\beta(\theta_i))$ is a solution to (\ref{opt:dual-inf-AP-strict}) of value $0$. As a consequence, if a Borel measure $\sigma$ on $\Lambda$ is a continuous BCCE, it places probability $1$ on the equilibrium bidding strategies~$\beta$.
\end{theorem}

\subsection{Weakening the Smoothness \& Monotonicity Assumption}\label{sec:weak-uniqueness}

Our proof uniqueness of continuous BCCE results for the first-price and all-pay auctions in Theorems \ref{thm:thm-fp-uniq} \& \ref{thm:thm-ap-uniq} are conditional on the rather restrictive set of assumptions outlined in Assumption \ref{asmptn:sup}. As discussed in the beginning of Section \ref{sec:uniqueness-strict}, even though there are learning algorithms which converge to a BCCE for the discretised auction whose support satisfies Assumptions \ref{asmptn:sup}.(1, 3, 4), the assumption of strictly increasing and differentiable bidding strategies is harder to justify; we are unaware of any learning algorithm which learns such a BCCE in the discretised setting. Therefore, we seek to relax Assumption \ref{asmptn:sup}.(2) to only requiring that buyers' bidding strategies are \emph{weakly} increasing, and differentiable except for finitely many jump discontinuities. This will allow us to convert uniqueness results for the BCCE of a continous auction into near-uniqueness results for the BCCE of its discretisations.

It is enlightening to consider why the dual solution in Theorem \ref{thm:thm-fp-uniq} does not apply when the bidding strategies may be only weakly increasing. Suppose that $\lambda$ is a vector of symmetric bidding strategies such that $\lambda_1$ is differentiable and weakly increasing, and consider a sequence $\lambda_n$ of bidding strategies satisfying Assumption \ref{asmptn:sup}; e.g. by fixing $\lambda_{ni}(\theta_i) = \lambda_i(\theta_i) + \theta_i/n$ for each player $i$. Then our dual solution $\epsilon$ would be guaranteed to be valid, if
\begin{equation}
  \lim_{n\rightarrow \infty} \int_0^1 d\theta_i \cdot \frac{\epsilon(\theta_i)}{\theta_i} \cdot \left( U_i(\beta,\lambda_{n(-i)}|\theta_i) - U_i(\lambda_n|\theta_{i}) \right) \leq \int_0^1 d\theta_i \cdot \frac{\epsilon(\theta_i)}{\theta_i} \cdot \left( U_i(\beta,\lambda_{-i}|\theta_i) - U_i(\lambda|\theta_i) \right) \label{con:good-conv}
\end{equation}
as $\lambda_n$ converges to $\lambda$ in the sup-norm, and with the dual solution $\epsilon$ we had, the left hand side is exactly equal to the objective which is continuous in $\lambda$. As a consequence, the dual constraint associated with $\lambda$ would only weakly bind. However, this cannot be assured for arbitrary distributions, as the following example illustrates.

\begin{example}
  Consider a first-price auction with two buyers, and an inverse prior function $v(\theta) = \sqrt{\theta}$. Then $\beta(\theta) = 2\sqrt{\theta}/3$. Now, consider a sequence of bidding strategies $\lambda_{ni}(\theta_i) = \theta_i/n$ for each buyer $i$. As $n \rightarrow \infty$, we have 
  \begin{equation*}
    \lim_{n\rightarrow \infty} \int_0^1 d\theta_i \cdot \frac{\epsilon(\theta_i)}{\theta_i} \cdot \left( U_i(\beta,\lambda_{n(-i)}|\theta_i) - U_i(\lambda_n|\theta_{i}) \right) = \int_0^1 d\theta_i \cdot \frac{2 (1-\sqrt(\theta_i))}{3\theta_i} \cdot \left( \frac{\sqrt{\theta_i}}{3} - \theta_i^{3/2} \right) = \frac{1}{9},
  \end{equation*}
  which as expected equals $\int_0^1 d\theta_i \cdot \beta(\theta_i)^2/2$. However, if each buyer commits to bidding $0$ (i.e. \emph{at} the limit $\lambda$), then they win the item at zero cost with probability $1/2$, which implies that $U_i(\lambda|\theta_i) = \sqrt{\theta_i}/2$. Therefore,
  \begin{equation*}
    \int_0^1 d\theta_i \cdot \frac{\epsilon(\theta_i)}{\theta_i} \cdot \left( U_i(\beta,\lambda_{-i}|\theta_i) - U_i(\lambda|\theta_i) \right) = \int_0^1 d\theta_i \cdot \frac{2 (1-\sqrt(\theta_i))}{3\theta_i} \cdot \left( \frac{\sqrt{\theta_i}}{3} - \frac{\sqrt{\theta_i}}{2} \right) = -\frac{1}{9}.
  \end{equation*}
  We conclude that our dual solution no longer has value $0$ for (\ref{opt:dual-inf-FP-strict}), when the support $\Lambda$ is expanded to include weakly increasing strategies. 
\end{example}

To circumvent this issue, we will seek an alternate dual solution. By Proposition \ref{prop:inc-bids}, we infer that $\epsilon(\theta_i) \propto \theta_i/v(\theta_i)$ may be a suitable candidate for the first-price auction, and $\epsilon(\theta_i) \propto 1/v(\theta_i)$ for the all-pay auction. We will want to show that such a dual solution guarantees (\ref{con:good-conv}) holds, though to ensure integribility of utilities given such $\epsilon$, we will strengthen our uniform boundedness assumption.

\begin{assumption}\label{asmptn:sup-weak}
  For the continuous auction $A$ and some $M \geq 1$, we restrict $\Lambda \subseteq \times_{i\in N} B^{[0,1]}$ to the set of bidding strategies $\lambda$ satisfying:
  \begin{enumerate}
    \item (Symmetry) For each buyer $i,j$, $\lambda_i = \lambda_j$.
    \item (Smoothness \& Weak Monotonicity) For each buyer $i$, $\lambda_i$ is weakly increasing and continuously differentiable on $(0,1]$.
    \item (No Gross Overbidding) For each buyer $i$, $\lambda_i \leq M v$.  
  \end{enumerate}
\end{assumption}

While Assumption \ref{asmptn:sup-weak}.(2) is weakened compared to Assumption \ref{asmptn:sup}.(2), the no overbidding assumption is stronger than the combination of normalisation and uniform boundedness assumptions. We remark, however, that the no overbidding assumption we make is weaker than the standard no overbidding assumption in literature, where $M$ would be set equal to $1$. Moreover, when we give our \emph{``translation''} results between the BCCE of discretised and continuous auctions, this assumption will be immediately satisfied in the support we consider. Our modified assumptions then allow us to show that (\ref{con:good-conv}) holds for the dual solutions we consider now; their feasibility for strictly increasing bidding strategies imply their feasibilty for weakly increasing ones.

\begin{lemma}\label{lem:v-solution}
  Suppose, that $\lambda$ satisfies Assumption \ref{asmptn:sup-weak}. Then there exists a sequence $\lambda_n$ of strictly increasing bidding strategies satisfying Assumption \ref{asmptn:sup-weak}\footnote{And hence Assumption \ref{asmptn:sup}.}, converging uniformly to $\lambda$ on $[0,1]$. Moreover, for both the first-price auction and the all-pay auction, 
  \begin{equation}
    \lim_{n\rightarrow \infty} \int_0^1 d\theta_i \cdot \frac{1}{v(\theta_i)} \cdot \left( U_i(\beta,\lambda_{n(-i)}|\theta_i) - U_i(\lambda_n|\theta_{i}) \right) \leq \int_0^1 d\theta_i \cdot \frac{1}{v(\theta_i)} \cdot \left( U_i(\beta,\lambda_{-i}|\theta_i) - U_i(\lambda|\theta_i) \right). \label{con:good-conv-ap}
  \end{equation}
\end{lemma}

\begin{proof}
  Fix such a $\lambda$. We want to construct a sequence of $\lambda_n$, satisfying Assumption \ref{asmptn:sup}, such that $\lambda_n$ converges uniformly to $\lambda$, and (\ref{con:good-conv-ap}) holds for the first-price and all-pay auctions. We will first focus on the construction of $\lambda_n$, which is the difficult part of the proof. Intuitively, the $U_i(\lambda_n|\theta_{i})$ term turns out to be non-problematic with our choice of $\epsilon$, while the term $U_i(\beta,\lambda_{n(-i)}|\theta_i)$ behaves differently only when the pre-image of $\beta(\theta)$ under $\lambda$ is not a singleton, i.e. when $\lambda^{-1}\beta(\theta)$ is not defined. However, as $\beta$ is strictly increasing, the set of such $\theta$ has measure zero, so the increasing flatness of $\lambda_n$ should not affect the integral. Our construction makes this argument precise.
  
  \newcommand{\I}{\mathcal{I}}

  For each $\underline{\theta} < \overline{\theta}$, call $(\underline{\theta},\overline{\theta})$ a flat interval if $\lambda(\underline{\theta}) = \lambda(\overline{\theta})$. Consider the union $I$ of all flat intervals; as a union of open intervals in $\mathbb{R}$, $I$ is an open set, and $I \subseteq [0,1]$. Together, this implies that $I$ is the union of countably many disjoint intervals, $I = \cup_{k \in \I} (\underline{\theta}_k,\overline{\theta}_k)$ for some countable index set $\I$. By flatness, we shall denote by $\ell_k$ the value of $\lambda$ over the interval $(\underline{\theta}_k,\overline{\theta}_k)$.

  For each $n \in \mathbb{N}$, we will fix a $\Delta_{kn}$ which describes how much we cut about the interval $(\underline{\theta}_k, \overline{\theta}_k)$. First, as a uniform bound, we will want $\Delta_{kn} < 2^{-k-1}/n$. Since $\lambda$ is continuous on the compact interval $[0,1]$ it is uniformly continuous, and there exists $\delta_{kn}$ such that for any $|\theta - \theta'| \leq \delta_{kn}$, $|\lambda(\theta) - \lambda(\theta')| \leq 2^{-k-2}/n$. We will also demand that $\Delta_{kn} < \delta_{kn}$.
  
  Now if $\ell_k$ is not in the image $\beta([0,1])$, then $\ell_k > \beta(1)$, and by continuity of $\lambda$ we will pick $\Delta_{kn}$ such that $\lambda(\underline{\theta}_k - \Delta_{kn}) \notin \beta([0,1])$. If instead $\ell_k = \beta(\gamma_k)$ for some $\gamma_k \in [0,1]$, note that $\beta$ maps open sets to open sets of $[0,\beta(1)]$\footnote{With respect to the open ball topology it inherits from $\mathbb{R}$.} as it is an invertible continuous function, and since $\beta$ is increasing $\beta((\gamma_k-2^{-k}/n,\gamma_k+2^{-k}/n))$ is an interval strictly containing $(\underline{\theta}_k,\overline{\theta}_k)$. We shall then require that $(\underline{\theta}_k-\Delta_{kn},\overline{\theta}_k + \Delta_{kn}) \cap [0,1] \subseteq \beta((\gamma_k-2^{-k}/n,\gamma_k+2^{-k}/n))$, and let 
  $$I_n = (0,1) \cap \left(\cup_{k \in \I} (\underline{\theta}_k-\Delta_{kn},\overline{\theta}_k + \Delta_{kn}) \right).$$ 
  $I_n$ is a union of open intervals, and thus $I_n = \cup_{k \in \I_n} (\underline{\varphi}_k, \overline{\varphi}_k)$ is also a countable union of disjoint intervals indexed by some set $\I_n$. For $\lambda_n$, we shall fix $\lambda_n(\theta) = \lambda(\theta)$ for each $\theta \in [0,1] \setminus I_n$. Otherwise, $\theta \in (\underline{\varphi}_k, \overline{\varphi}_k)$ for some $k \in I_n$, and we will construct $\lambda_n(\theta)$ by setting 
  $$\lambda_n'(\theta) = \frac{\lambda'(\underline{\varphi}_k)(\overline{\varphi}_k-\theta) + \lambda'(\overline{\varphi}_k)(\theta-\underline{\varphi}_k)}{\overline{\varphi}_k - \underline{\varphi}_k} + c_{kn}(\theta) \geq 0$$
  on $[\underline{\varphi}_k,\overline{\varphi}_k]$, where $c_{kn}(\theta)$ is a suitable $C^\infty$ function with compact support on $[\underline{\varphi}_k,\overline{\varphi}_k]$, such that 
  $$\int_{\underline{\varphi}_k}^{\overline{\varphi}_k} d\theta \cdot \lambda'_n(\theta) = \lambda(\overline{\varphi}_k) - \lambda(\underline{\varphi}_k) \textnormal{ and } \lambda_n(\theta) \leq Mv(\theta) \textnormal{ for any } \theta \in (\underline{\varphi}_k,\overline{\varphi}_k).$$
  The latter choice is possible by the differentiability and no gross overbidding assumptions on $\lambda$.

  It remains to show that $\lambda_n$ converges uniformly to $\lambda$. For any $\theta \notin I_n$, $\lambda_n(\theta) = \lambda(\theta)$. If instead $\theta \in I_n$, then for some $\kappa \in \I_n$, $\theta \in (\underline{\varphi}_\kappa,\overline{\varphi}_\kappa)$. Then, as 
  $$\lambda(\underline{\varphi}_{\kappa}) = \lambda_n(\underline{\varphi}_{\kappa}) < \lambda_n(\theta) < \lambda_n(\overline{\varphi}_{\kappa}) = \lambda(\overline{\varphi}_{\kappa}),$$
  and likewise for $\lambda(\theta)$, it is sufficient to obtain a bound for $\lambda(\underline{\varphi}_{\kappa}) - \lambda(\overline{\varphi}_{\kappa})$. Now, as $\theta \in (\underline{\varphi}_{\kappa},\overline{\varphi}_{\kappa})$, then $\theta \in (\underline{\theta}_{k^*},\overline{\theta}_{k^*})$. For two  initial segments $L,R$ of $\mathbb{N}$, we shall construct a mapping $\ell : L \rightarrow \I_n, r : R \rightarrow \I_n$, such that $\ell(1) = r(1) = {k^*}$, and 
  \begin{equation}\label{eq:flat-exp}\left(\cup_{m \in L} (\underline{\theta}_{\ell(m)},\overline{\theta}_{\ell(m)})\right) \cup \left(\cup_{m \in R} (\underline{\theta}_{r(m)},\overline{\theta}_{r(m)})\right)= (\underline{\varphi}_{\kappa},\overline{\varphi}_{\kappa}).\end{equation}

  The idea is to iteratively expand the interval $(\underline{\theta}_1,\overline{\theta}_1)$ left and right for each $\eta \in \mathbb{N}$, potentially terminating whenever lower and upper bounds are attained. We present the construction of $L$; the construction of $R$ is analogous. We fix $L^1 = \{1\}$ and $\ell^1: 1 \mapsto 1$. Now, given $L^\eta$, suppose that $\cup_{m \in L^\eta} (\underline{\theta}_{\ell(m)},\overline{\theta}_{\ell(m)})$ is some interval $(\underline{\psi}_\eta,\overline{\psi}_\eta)$; this is of course true for $\eta = 1$. If $\underline{\psi}_\eta = \underline{\varphi}_\kappa$, then we let $L = L^\eta$ and terminate. Otherwise, we define 
  $$I^\eta_L \equiv \left[ \left(\underline{\varphi}_\kappa + \underline{\psi}_\eta\right)/2, \underline{\psi}_\eta\right] \subsetneq (\underline{\varphi}_{\kappa},\overline{\varphi}_{\kappa}).$$
  
  This implies that $(\underline{\theta}_k, \overline{\theta}_k)_{k \in \I_n}$ is an open cover of $I^\eta_L$; hence it has a finite subcover $J^\eta$ minimal with respect to set inclusion. This union of this subcover is necessarily an interval $(\underline{\psi}'_\eta,\overline{\psi}'_\eta)$ containing $I^\eta_L$. Extend $L^\eta$ to $L^{\eta+1}$, by letting $\ell(m+1) \in J^\eta$ be such that $(\underline{\theta}_{\ell(m+1)},\overline{\theta}_{\ell(m+1)}) \ni \underline{\theta}_{\ell(m)}$ for each $m \geq |L^\eta|$ as long as such an $\ell(m+1)$ exists in $J^\eta$; since $J^\eta$ is finite and $(\underline{\psi}'_\eta,\overline{\psi}'_\eta)$ is open, this iteration can be done only finitely many times. In particular, once we terminate, $\underline{\psi}_{\eta+1} = \underline{\psi}'_\eta$. 

  If $L = L^\eta$ for some finite $\eta$, then $\underline{\psi}_\eta = \underline{\varphi}_\kappa$, and hence $\cup_{m \in L} (\underline{\theta}_{\ell(m)},\overline{\theta}_{\ell(m)}) \supseteq (\underline{\varphi}_\kappa,\overline{\theta}_{k^*})$. If $L$ is countably infinite, then note that $\underline{\psi}_{\eta+1} - \underline{\varphi}_\kappa < (\underline{\psi}_{\eta} - \underline{\varphi}_\kappa)/2$, and hence 
  $$ \cup_{m \in L^{\eta+1}} (\underline{\theta}_{\ell(m)},\overline{\theta}_{\ell(m)}) \supseteq \left[ \underline{\varphi}_\kappa \cdot \left( 1 - \frac{1}{2^\eta} \right) + \frac{\underline{\theta}_{k^*}}{2^\eta}, \overline{\theta}_{k^*} \right).$$
  We thus conclude that, with an analogous construction of $R,r$, (\ref{eq:flat-exp}) holds. Given this, 
  \begin{align*}
    0 \leq \lambda(\theta) - \lambda(\underline{\varphi}_\kappa) & = \lambda(\theta) - \lambda(\underline{\theta}_{k^*}) + \sum_{m \in L, m > 1} \lambda(\underline{\theta}_{\ell(m)}) - \lambda(\underline{\theta}_{\ell(m-1)}) \\
    & \leq \sum_{m \in L} \frac{2^{-\ell(m)-2}}{n} \leq \sum_{k \in \mathbb{N}} \frac{2^{-k-2}}{n} = \frac{1}{2n}.
  \end{align*}
  Here, the first inequality follows from the choice of $\Delta_{kn}$ Via an analogous argument, $\lambda(\overline{\varphi}_\kappa) - \lambda(\theta) \leq 1/2n$, which implies that $|\lambda(\theta) - \lambda_n(\theta)| \leq 1/n$. As a consequence, $\| \lambda - \lambda_n \|_\infty \leq 1/n$. 

  Given our sequence $\lambda_n$, we want to show that (\ref{con:good-conv-ap}) holds. We thus look at the difference 
  \begin{align*}
    & \int_0^1 d\theta_i \cdot \frac{1}{v(\theta_i)} \left[ \left( U_i(\beta,\lambda_{n(-i)}|\theta_i) - U_i(\beta,\lambda_{-i}|\theta_i) \right) - \left( U_i(\lambda_{n}|\theta_i) - U_i(\lambda|\theta_i) \right) \right] \\
    \leq \quad &  \int_0^1 d\theta_i \cdot \frac{1}{v(\theta_i)} \left[ \left| U_i(\beta,\lambda_{n(-i)}|\theta_i) - U_i(\beta,\lambda_{-i}|\theta_i) \right| - \left( U_i(\lambda_{n}|\theta_i) - U_i(\lambda|\theta_i) \right) \right] 
  \end{align*}
  We inspect the $\left| U_i(\beta,\lambda_{n(-i)}|\theta_i) - U_i(\beta,\lambda_{-i}|\theta_i) \right|$ term first. If $\theta \in [0,1] \setminus I_n$, then $\lambda(\theta) = \lambda_n(\theta)$ there. As a consequence, $\lambda^{-1}\beta(\theta)$ exists, and equals $\lambda_n^{-1}\beta(\theta)$; and 
  $\left| U_i(\beta,\lambda_{n(-i)}|\theta_i) - U_i(\beta,\lambda_{-i}|\theta_i) \right| = 0$. If instead $\beta(\theta) \in I_n$, $\theta \in \beta^{-1}((\underline{\theta}_k-\Delta_{kn}, \overline{\theta}+\Delta_{kn})) \subseteq (\gamma_k-2^{-k}/n,\gamma_k+2^{-k}/n)$ for some $k \in \I$. By construction, set of such $\theta$ has measure $\leq 1/n$. As $\beta \leq v$ for both the first-price and the all-pay auction, the absolute value of this difference can be at most $2v(\theta_i)$ for each $\theta_i$, which implies that 
  \begin{align*}
    \int_0^1 d\theta_i \cdot \frac{1}{v(\theta_i)} \cdot \left| U_i(\beta,\lambda_{n(-i)}|\theta_i) -  U_i(\beta,\lambda_{-i}|\theta_i)  \right| & \leq \frac{2}{n}.
  \end{align*}

  Finally, we look at the term $U_i(\lambda|\theta_i)-U_i(\lambda_n|\theta_i)$. We first note that $U_i(\lambda_n|\theta_i)$ is defined as in Section \ref{sec:uniqueness-strict}, as $\lambda_n$ is strictly increasing. Meanwhile, 
  $$U_i(\lambda|\theta_i) = \begin{cases}
    P(\theta_i,\lambda) \cdot (v(\theta_i) - \lambda(\theta)) & \textnormal{for the first-price auction,} \\
    P(\theta_i,\lambda) \cdot v(\theta_i) - \lambda(\theta) & \textnormal{for the all-pay auction,}
  \end{cases}$$ 
  where $$P(\theta_i,\lambda) = \begin{cases}
    \sum_{m = 0}^{N-1} \frac{1}{m+1} \binom{N-1}{m} \underline{\theta}_k^{N-1-m} (\overline{\theta}_k - \underline{\theta}_k)^m & \exists \ k \in \I, \theta_i \in [\overline{\theta}_k,\underline{\theta}_k], \\
    \theta_i^{N-1} & \textnormal{else.}
  \end{cases}$$
  Therefore, for the first-price auction the difference can be decomposed, 
  \begin{align*}
    & \int_0^1 d\theta_i \cdot \frac{1}{v(\theta_i)} \cdot \left( U_i(\lambda|\theta_i)-U_i(\lambda_n|\theta_i) \right) \\ = & \int_{[0,1] \setminus I} d\theta_i \cdot \frac{1}{v(\theta_i)} \cdot \theta_i^{N-1} (\lambda(\theta_i) - \lambda_n(\theta_i)) \\
     - & \sum_{k \in \I} \int_{\underline{\theta}_k}^{\overline{\theta}_k} d\theta_i \cdot \left(\theta_i^{N-1} - \sum_{m = 0}^{N-1} \frac{1}{m+1} \binom{N-1}{m} \underline{\theta}_k^{N-1-m} (\overline{\theta}_k - \underline{\theta}_k)^m\right) \\
    + & \sum_{k \in \I} \int_{\underline{\theta}_k}^{\overline{\theta}_k} d\theta_i \cdot \frac{\ell_k}{v(\theta_i)} \cdot \left(\theta_i^{N-1} - \sum_{m = 0}^{N-1} \frac{1}{m+1} \binom{N-1}{m} \underline{\theta}_k^{N-1-m} (\overline{\theta}_k - \underline{\theta}_k)^m\right) \\
    - & \sum_{k \in \I} \int_{\underline{\theta}_k}^{\overline{\theta}_k} d\theta_i \cdot \frac{1}{v(\theta_i)} \cdot \theta_i^{N-1} \cdot (\ell_k - \lambda_n(\theta_i)).
  \end{align*}
  The second term evaluates\footnote{One way to show this is by evaluating the integral, and then showing that its derivative with respect to $\overline{\theta}_k$ equals zero via the binomial theorem. As the term equals zero when $\overline{\theta}_k = \underline{\theta}_k$, we have the desired result.} to $0$, while the third term is $\leq 0$ as $\ell_k/v$ is decreasing. On the other hand, the first and the fourth term are both bounded above by 
  $$ \int_0^1 d\theta_i \cdot \frac{1}{v(\theta_i)} \cdot \theta_i^{N-1} \cdot |\lambda(\theta_i) - \lambda_n(\theta_i)|.$$
  As $\|\lambda_n - \lambda\|_\infty \leq 1/n$ and as $\lambda_n$ also satisfies no gross overbidding, this integral is in turn bounded above by $\mu_n\cdot 2M + (1-\mu_n) \cdot 1/\sqrt{n}$, where $\mu_n$ is the measure of the set of $\theta \in [0,1]$ such that $v(\theta_i) < 1/\sqrt{n}$. As $v$ is continuous and strictly increasing, $\lim_{n \rightarrow \infty} \mu_n = 0$. The bound for the all-pay auction follows from an identical line of arguments; the second and third terms are identical, while the first and fourth terms are no longer multiplied by $\theta_i^{N-1}$.
\end{proof}

By the result of Lemma \ref{lem:v-solution}, we will write our primal infinite LP for both the first-price and the all-pay auctions as 
\begin{align}
  \sup_{\sigma \geq 0} \int d\sigma(\lambda) \cdot \sum_{i \in [N]} \int_0^1 d\theta_i \cdot \alpha(\theta_i|\lambda) \textnormal{ subject to}& \label{opt:inf-weak}\\
  \int_\Lambda d\sigma(\lambda) & = 1 \tag{$\gamma$} \label{cons:inf-gamma-weak} \\
  \int_\Lambda d\sigma(\lambda) \cdot \frac{1}{v(\theta_i)}\left( U_i(\beta,\lambda_{-i}|\theta_i) - U_i(\lambda|\theta_i) \right) & \leq 0 \tag{$\epsilon(\theta_i)$} \ \forall \ i \in [N], \theta_i \in (0,1], \label{cons:inf-eq-weak}
\end{align}
where $\Lambda$ is a set of bidding functions satisfying Assumption \ref{asmptn:sup-weak}, and $\alpha(\theta|\lambda)$ is yet to be constructed. Then (\ref{opt:inf-weak}) has a dual infinite LP 
\begin{align}
  \inf_{\epsilon \geq 0, \gamma} \gamma \textnormal{ subject to} &  \label{opt:dual-inf-weak}\\
  \gamma + \sum_{i \in [N]} \int_0^1 d\theta_i \cdot \epsilon(\theta_i) \cdot \frac{1}{v(\theta_i)} \left( U_i(\beta,\lambda_{-i}|\theta_i) - U_i(\lambda|\theta_i) \right) & \geq \sum_{i \in [N]} \alpha(\theta_i|\lambda) \label{cons:inf-dual-gen-weak} \tag{$\sigma(\lambda)$}
  \ \forall \ \lambda \in \Lambda. 
\end{align}

When we discuss a continuous BCCE $\sigma$ with support in weakly increasing strategies, we shall this time simply assume that $\sigma$ is sufficiently well-behaved, such that 
\begin{align}
  & \int_\lambda d\sigma(\lambda) \cdot \sum_{i \in [N]} \int_{0}^1 d\theta_i \cdot \frac{\epsilon(\theta_i)}{v(\theta_i)} \cdot \left( U_i(\beta,\lambda_{-i} | \theta_i) - U_i(\lambda|\theta_i) \right) \label{cond:fubini-weak} \\
  = & \sum_{i \in [N]} \int_{0}^1 d\theta_i \cdot \frac{\epsilon(\theta_i)}{v(\theta_i)} \cdot \int_\lambda d\sigma(\lambda) \cdot \left( U_i(\beta,\lambda_{-i} | \theta_i) - U_i(\lambda|\theta_i) \right) \nonumber
\end{align}
holds. One way to ensure this is to assume that $\sigma$ is restricted to be a Borel probability measure with respect to discrete topology; 
when we provide our approximation results by relaxing the constrants (\ref{cons:inf-eq-weak}), we will be able to consider $\Lambda$ as a finite set of piecewise constant functions.

Then to obtain our uniqueness results, we will need to show that such a dual solution still has value $0$ for the dual problems (\ref{opt:dual-inf-FP-strict},\ref{opt:dual-inf-AP-strict}). This turns out to be the case under certain assumptions on the prior distribution, in line with our numerical results in Section \ref{sec:numerical-results}; for first-price auctions, strict concavity of the prior distribution -- or equivalently, strict convexity of $v(\theta)$ is a sufficient condition, while for all-pay auctions the conditions turn out to be a lot less restrictive.

\begin{theorem}\label{thm:thm-weak-uniq}
  Suppose for that $\Lambda$ is a set of bidding strategies satisfying Assumption \ref{asmptn:sup-weak}. Then $\epsilon(\theta_i) = 1$ and $\gamma = 0$ is a solution to (\ref{opt:dual-inf-weak}) of value $0$, whenever $v(\theta)$ is strictly increasing and differentiable on $[0,1]$
  \begin{enumerate}
    \item for the all-pay auction when we fix 
    $$\alpha(\theta_i|\lambda) = \frac{\lambda(\theta_i)-\beta(\theta_i)}{v(\theta_i)} - \int_{\beta(\theta)}^{\max\{\lambda(\theta),\beta(1)\}} d\mu \cdot \frac{\theta_i^{N-2}}{\beta^{-1}(\mu)^{N-2} \cdot v\beta^{-1}(\mu)},$$  
    \item and if $v(\theta)$ is also strictly convex, for the first-price auction when we fix 
    $$\alpha(\theta_i|\lambda) = \theta_i^{N-2} \cdot \left(\frac{\theta_i}{v(\theta_i)} \cdot (\lambda(\theta) - \beta(\theta)) - \int_{\beta(\theta)}^{\max\{\lambda(\theta),\beta(1)\}} d\mu \cdot \frac{\beta^{-1}(\mu)}{v\beta^{-1}(\mu)}\right).$$
  \end{enumerate}
  As a consequence, for such prior distributions, any continuous BCCE $\sigma$ on $\Lambda$ for which (\ref{cond:fubini-weak}) holds places probability $1$ on the equilibrium bidding strategies~$\beta$.
\end{theorem}

\begin{proof}
  We proceed as in the proof of Theorem \ref{thm:thm-fp-uniq}, and argue the result for the first-price auction; the result for the all-pay auction follows through an identical line of arguments we omit. For $0 \leq \underline{\theta} < \overline{\theta} \leq 1$, again denote by $\gamma(\lambda|\underline{\theta},\overline{\theta})$,
  $$ \gamma(\lambda|\underline{\theta},\overline{\theta}) = \int_{\underline{\theta}}^{\overline{\theta}} d\theta_1 \cdot \left( \alpha(\theta_1|\lambda) - \frac{1}{v(\theta_1)} \cdot  \left( U_1(\beta,\lambda_{-1} | \theta_1) - U_1(\lambda|\theta_1) \right) \right). $$
  We again need to show that $\gamma(\lambda|0,1)$ is maximised when $\lambda_i = \beta$ for any buyer $i$. Our construction of $\alpha$ is in fact such that, for both the first-price and the all-pay auction, 
  $$ \int_0^1 d\theta_1 \cdot \frac{\delta\alpha(\theta_1|\lambda)}{\delta\lambda} \cdot \delta\lambda(\theta_1) =  \int_{0}^{1} d\theta \cdot \frac{1}{v(\theta_1)} \cdot \frac{\delta\left( U_1(\beta,\lambda_{-1} | \theta_1) - U_1(\lambda|\theta_1) \right)}{\delta\lambda} \cdot \delta\lambda(\theta_1)$$
  for any bidding strategy $\lambda$ which is strictly increasing and satisfying Assumption \ref{asmptn:sup-weak}, and any smooth variation $\delta\lambda$. Thus it is sufficient to show that 
  \begin{equation}\int_{0}^{1} d\theta_1 \cdot \frac{1}{v(\theta_1)} \cdot \frac{\delta\left( U_1(\beta,\lambda_{-1} | \theta_1) - U_1(\lambda|\theta_1) \right)}{\delta\lambda} \cdot \delta\lambda(\theta) \leq 0 \label{eq:dual-good} \tag{$\delta\alpha$}\end{equation}
  whenever we fix $\delta\lambda = \beta - \lambda$, with equality only when $\beta = \lambda$.
  
  So once again divide the interval $[0,1]$ into finitely many subintervals, $[0,1] = \cup_{k} [\underline{\theta}_k,\overline{\theta}_k]$, as in the proofs of Theorems \ref{thm:thm-fp-uniq} and \ref{thm:thm-ap-uniq}. The proof then follows by identical case analysis; we present the case $\lambda(1) > \beta(1)$ to illustrate the arguments. By continuity of $\lambda - \beta$, since $\lambda(0) = \beta(0)$, there exists supremal (hence maximal) $\underline{\theta}$ such that $\lambda(\underline{\theta}) = \beta(\underline{\theta})$. By the supremum property, $\lambda(\theta) > \beta(\theta)$ for any $\theta > \underline{\theta}_k$. Thus, for any $\theta \in (\underline{\theta},1]$,  $\lambda^{-1}\beta(\theta) < \theta$, and the left hand side of (\ref{eq:dual-good}) equals
  \begin{align*}
      & = \int_{\underline{\theta}_k}^{1} d\theta \cdot \Bigg[-\frac{(N-1)\lambda^{-1}\beta(\theta)^{N-2}}{\lambda' \lambda^{-1} \beta(\theta) \cdot v(\theta)} \cdot (v(\theta) - \beta(\theta)) \cdot \delta\lambda(\lambda^{-1}\beta(\theta)) + \frac{\theta^{N-1}}{v(\theta)} \cdot \delta\lambda(\theta) \Bigg] \\
      & = \int_{\lambda^{-1}\beta(1)}^{1} d\theta \cdot \delta\lambda(\theta) \cdot \theta^{N-2} \cdot \left( \frac{\theta}{v(\theta)} \right) + \int_{\underline{\theta}_k}^{\lambda^{-1}\beta(1)} d\theta \cdot \delta\lambda(\theta) \cdot \theta^{N-2}\cdot \left( \frac{\theta}{v(\theta)} - \frac{\beta^{-1}\lambda(\theta)}{v\beta^{-1}\lambda(\theta)} \right) .
  \end{align*}
  For any $\theta \in (\lambda^{-1}{\beta(1)},1]$, $\delta\lambda(\theta) = \beta(\theta) - \lambda(\theta) < 0$, so the first term is negative. Meanwhile, strict convexity of $v(\theta)$ along with $v(0) = 0$ implies that $v(\theta)/\theta$ is strictly increasing on $[0,1]$, which implies that the second term is also negative.
\end{proof}

We have thus identified, in a sense, the \emph{source} of the prior dependence on the nearness of the BCCE of first-price auctions to its canonical equilibria. One shortcoming of Theorem \ref{thm:thm-weak-uniq}, however, is that the uniqueness result does not yet extend to weakly convex $v(\theta)$, which would include the uniform distribution. Also, the objective $\alpha$ is less intuitive than the modified Wasserstein-$2$ distance we considered in Section \ref{sec:uniqueness-strict}. The following result shows that for first-price auctions, slightly stronger assumptions on the convexity of $v$ allows us to exchange our objective for a weighted Wasserstein-$2$ distance, while for the all-pay auction we require only a mild regularity assumption at $\theta = 0$.

\begin{proposition}\label{prop:wasserstein-prior}
  For any bidding strategies $\lambda$ satisfying Assumption \ref{asmptn:sup-weak}:
  \begin{enumerate}
    \item For a first-price auction, if $\beta^{-1}/v\beta^{-1}$ also has a lower bound \underbar{on the magnitude} of its derivative, then $\int_0^1 d\theta \cdot \alpha(\theta|\lambda) \geq \int_0^1 d\theta \cdot \theta^{N-2} \cdot \frac{(\lambda(\theta)-\beta(\theta))^2}{2} \cdot C^{FP}$, where
    $$C^{FP} = \inf\left\{ -\frac{d}{d\mu} \left( \frac{\beta^{-1}(\mu)}{v\beta^{-1}(\mu)} \right) \ \Bigg| \ \mu \in (0,\beta(1)]  \right\} \cup \left\{ \frac{1}{v(1) \max B} \right\} > 0.$$
    This condition is satisfied whenever $\theta/v(\theta)$ has an upper bound on its derivative.
    \item For an all-pay auction, whenever $1/(\beta^{-1})^{N-2}v\beta^{-1}$ as a lower bound \underbar{on the magnitude} of its derivative, $\int_0^1 d\theta \cdot \alpha(\theta|\lambda) \geq \int_0^1 d\theta \cdot \theta^{N-2} \cdot \frac{(\lambda(\theta)-\beta(\theta))^2}{2} \cdot C^{AP}$, where
    $$C^{AP} = \inf\left\{ -\frac{d}{d\mu} \left( \frac{1}{\beta^{-1}(\mu)^{N-2} \cdot v\beta^{-1}(\mu)} \right) \ \Bigg| \ \mu \in (0,\beta(1)]  \right\} \cup \left\{ \frac{1}{v(1) \max B} \right\} > 0.$$
    This holds if $d/d\theta [1/v(\theta)]$ has a non-zero limit at $0$.
  \end{enumerate}
\end{proposition}

\begin{proof}
  For the first-price auction, whenever $\lambda(1) > \beta(1)$,
  \begin{align*}
      & \int_0^1 d\theta \cdot \frac{\delta \alpha(\theta|\lambda)}{\delta\lambda} \cdot \delta\lambda(\theta) \\
      & = \int_{\lambda^{-1}\beta(1)}^1 d\theta \cdot \delta\lambda(\theta) \cdot \theta^{N-2} \cdot \left( \frac{\theta}{v(\theta)} \right) + \int_{\underline{\theta}_k}^{\lambda^{-1}\beta(1)} d\theta \cdot \delta\lambda(\theta) \cdot \theta^{N-2} \cdot \left( \frac{\theta}{v(\theta)} - \frac{\beta^{-1}\lambda(\theta)}{v\beta^{-1}\lambda(\theta)} \right).
  \end{align*}
  When we set $\delta\lambda = \beta - \lambda$, the resulting integrand should be pointwise lesser than the integrand of 
  $$\int_0^1 d\theta \cdot \delta\lambda(\theta) \cdot \theta^{N-2} \cdot \frac{\delta}{\delta\lambda} \left( \frac{(\lambda(\theta)-\beta(\theta))^2}{2} \right) \cdot C^{FP}.$$
  Now, if $\theta > \lambda^{-1}\beta(\theta)$, then this necessitates
  \begin{align*}
      \theta^{N-2} \cdot (\beta(\theta)-\lambda(\theta)) \cdot \left( \frac{\theta}{v(\theta)}\right) & \leq -\theta^{N-2} (\beta(\theta)-\lambda(\theta))^2 \cdot C^{FP} \\
      \Rightarrow C^{FP} & \leq \frac{\theta}{v(\theta)(\lambda(\theta)-\beta(\theta))}.
  \end{align*}
  This is guaranteed whenever $C^{FP} \leq 1/(v(1) \max B)$, since $\theta/v(\theta)$ is strictly decreasing on $[0,1]$ by the strict convexity assumption on $v$. If instead $\theta < \lambda^{-1}\beta(1)$, then we require
  \begin{align*}
      \theta^{N-2} \cdot (\beta(\theta)-\lambda(\theta)) \cdot \left( \frac{\theta}{v(\theta)} - \frac{\beta^{-1}\lambda(\theta)}{v\beta^{-1}\lambda(\theta)} \right) & \leq -\theta^{N-2} \cdot (\beta(\theta)-\lambda(\theta))^2 \cdot C^{FP} \\
      \Rightarrow C^{FP} \leq -\frac{1}{\beta(\theta)-\lambda(\theta)} \cdot \left( \frac{\theta}{v(\theta)} - \frac{\beta^{-1}\lambda(\theta)}{v\beta^{-1}\lambda(\theta)} \right).
  \end{align*}
  However, by the mean-value theorem, the right hand side equals $-d/d\mu (\beta^{-1}(\mu)/v\beta^{-1}(\mu))$, evaluated at some point $\mu$ between $\lambda(\theta)$ and $\beta(\theta)$. To see the sufficient condition, 
  \begin{align*}
    \frac{d}{d\mu} \frac{\beta^{-1}(\mu)}{v\beta^{-1}(\mu)} & = \frac{1}{\beta' \beta^{-1} (\mu)} \cdot \frac{v\beta^{-1}(\mu) - \beta^{-1}(\mu) v'\beta^{-1}(\mu)}{v\beta^{-1}(\mu)^2} = \frac{1}{\beta'(\theta)} \cdot \frac{v(\theta) - \theta v'(\theta)}{v(\theta)^2} \\
    & = \frac{1}{\beta'(\theta)} \frac{d}{d\theta} \left( \frac{\theta}{v(\theta)} \right) \leq \frac{\theta}{v(\theta)} \frac{d}{d\theta} \left( \frac{\theta}{(N-1)v(\theta)} \right) \leq \frac{1}{(N-1)v(1)} \frac{d}{d\theta} \left( \frac{\theta}{v(\theta)} \right).
  \end{align*}
  Here, the second equality is by changing variables $\mu = \beta(\theta)$, the first inequality is by noting that $\theta \beta'(\theta) = (N-1) \cdot (v(\theta) - \beta(\theta))$ for a first-price auction, and the final inequality holds since $\theta/v(\theta)$ is decreasing. The case of the all-pay auction follows similarly.
\end{proof}

For a quick sanity check, we remark that the objective $\alpha$ we construct is \emph{scale invariant}, and the constant factor $C^\mathcal{A}$ in Proposition \ref{prop:wasserstein-prior} ensures our modified Wasserstein-$2$ distance becomes so as well. In particular, whenever we scale all valuations and bids by a constant $A > 0$, i.e. when we let $v \rightarrow A \cdot v$, $B \rightarrow A \cdot B$, and $\lambda \rightarrow A \cdot \lambda$, the constant $C^\mathcal{A}$ scales like $\propto 1/A^2$ while the integrand scales $(\lambda(\theta) - \beta(\theta))^2 \propto A^2$.

\subsection{Continuous Auctions with Non-Unique BCCE}\label{sec:examples}

It is worthwhile to discuss how the result of Theorem \ref{thm:thm-weak-uniq} breaks down when the prior distribution exhibits non-concavities. The main source of tension turns out to be when flat bidding profiles render our dual solution infeasible. Allowing ties, note that a weakly convex prior $F(v)$ implies a weakly concave valuation function in quantile space, $v(\theta) = F^{-1}(\theta)$. Moreover, in this case for a first-price auction,
$$\beta(\theta) = \frac{1}{\theta^{N-1}} \int_0^\theta (N-1) \theta'^{N-2} v(\theta') d\theta'.$$
Therefore, if $v$ is weakly concave, $v(\theta') \geq (\theta'/\theta) \cdot v(\theta)$ for every $\theta' \in (0,\theta)$, and as a result, for every $\theta > 0$,
$$ \beta(\theta) \geq \frac{1}{\theta^{N}} \int_0^\theta (N-1) \theta'^{N-1} v(\theta) d\theta' = \frac{N-1}{N} v(\theta). $$
Or, in other words, $v(\theta) - \beta(\theta) \leq \frac{1}{N} \cdot v(\theta)$. Denoting by $0$ the bidding strategies where every buyer bids $0$ no matter their valuation, we then have
\begin{align*}
    \int_0^1 d\theta \cdot \left( - \frac{1}{v(\theta)} \cdot (U_1(\beta,0_{-1}|\theta) - u(0|\theta) ) \right)
    \geq \int_0^1 d\theta \cdot \left(- \frac{1}{v(\theta)} \cdot \left(v(\theta) - \beta(\theta) - \frac{v(\theta)}{N}\right) \right) \geq 0,
\end{align*}
and the inequality holds strictly whenever $v$ is strictly concave. As a consequence, when the buyers' symmetric prior is weakly convex, there is no dual solution to (\ref{opt:dual-inf-weak}) of the form we seek. We can actually reflect this phenomenon in the form of primal solutions to (\ref{opt:generic-primal}), i.e. an explicit BCCE for continuous two buyer first-price auctions with strictly convex priors, though the strategies we find will fail to weakly increase. We remark that all bidding strategies considered are nevertheless symmetric, normalised at $0$ and uniformly bounded. This emphasises that our assumption of strictly increasing support for BCCEs is crucial for the uniqueness result in Theorem \ref{thm:thm-fp-uniq}.

\newcommand{\I}{\textnormal{I}}
\newcommand{\II}{\textnormal{II}}

\begin{example}
    Consider a two buyer first-price auction with prior distribution $F(v) = v^2$, i.e. $v(\theta) = \sqrt{\theta}$. Here, the equilibrium bidding strategy satisfies $\beta(\theta) = \frac{2}{3}\sqrt{\theta}$, and we will construct a BCCE where this strategy is played with probability $< 1$. Towards this end, consider bidding strategy $\lambda$, defined as  
    $$ \lambda(\theta) = \begin{cases}
        \hat{\beta}(\theta) & \theta \in [0,\delta) \\
        0 & \theta \geq \delta
    \end{cases}$$
    for some $\delta > 0$, where $\hat{\beta}$ is chosen such that buyers of type $\theta < \delta$ best respond. That is to say,
    $$ \hat{\beta}(\theta) = \arg \max_{b \in [0,\infty)} (1-\delta + \hat{\beta}^{-1}(b)) \cdot (\sqrt{\theta} - b),$$
    which in turn implies that $\hat{\beta}(\theta) = \frac{2}{3(1-\delta+\theta)}\theta^{3/2}$. We will want to find $\varepsilon > 0$ such that the correlated bidding strategy which puts weight $\varepsilon$ on both buyers employing $\lambda$ and weight $1-\varepsilon$ on both buyers employing $\beta$ forms a BCCE of the auction.

    For our purposes, it is sufficient to pick $\varepsilon = .001$ and $\delta = .05$. It remains to verify that we indeed have a BCCE. Here, note that a buyer of type $\theta < \delta$ is pointwise best responding by construction, so we only need to check the BCCE constraints for $\theta \geq \delta$. Such a buyer has expected utility 
    \begin{equation}\label{eqn:ex-util}(1-\varepsilon) \cdot \frac{\theta^{3/2}}{3} + \varepsilon \cdot (1-\delta) \cdot \frac{\sqrt{\theta}}{2} .\end{equation}
    In turn, if such a buyer deviates to bidding $b$, they have expected utility 
    \begin{equation}\label{eqn:ex-lb}
        \leq (\sqrt{\theta} - b) \cdot \left( \varepsilon + (1-\varepsilon) \cdot \frac{9}{4}b^2 \right)
    \end{equation}
    since with probability $\varepsilon$ they will win the item with probability $\leq 1$ with their bid, and with probability $(1-\varepsilon)$ the other buyer employs the equilibrium bidding strategies. Now, the maximiser $b(\theta)$ of (\ref{eqn:ex-lb}) is increasing in $\theta$ whenever $\theta \in [.05,1]$, and is in fact given by $b = \frac{\sqrt{\theta}}{3} + \frac{1}{27} \sqrt{81\theta - \frac{4}{37}}$. Then for $\theta \in [.05,1]$, it is possible to numerically verify that (\ref{eqn:ex-util}) $>$ (\ref{eqn:ex-lb}).
    
\end{example}

Another way to observe that \emph{some} regularity condition on the priors may be required for the uniqueness of BCCE in first-price auctions is by constructing continuous auctions without a unique BCCE, by an appeal to first-price auctions of complete information considered in \cite{FLN16}. In this case, the general idea is to find a first-price auction of complete information and a BCCE for it where no equilibrium constraint binds\footnote{We are grateful to Jason Hartline for the idea on the construction of this example.}:

\begin{example}\label{ex:necessity}
    Consider a two buyer first-price auction, where the valuation of each buyer $i$ is drawn i.i.d. to be $1$ or $1.2$ with probability $1/2$ each. Consider the correlated bid distribution, where each buyer submits identical bids following the cumulative distribution function
    $$ H(b) = \min \left\{ \frac{(1.2-.88)(1-.8)}{(1.2-.8)(1-b)} , \frac{1.2-.88}{1.2-b}, 1\right\}.$$
    This is a Bayesian CCE -- moreover, no coarse correlated equilibrium constraint (\ref{def:BCE}) binds. In particular, in the BCCE outcome, buyer $1$ obtains utility $\sim .224458$ while buyer $2$ gets utility $\sim .324458$. However, by definition of $H(b)$ the maximum utility buyer $1$ may get from bidding $b > 0$ is 
    $$ (1-b) \cdot H(b) \leq (1-b) \cdot \frac{.16}{1-b} = .16 < .22448.$$
    Similarly, buyer $2$ may at most obtain $.32 < .324458$ from such a unilateral deviation. Now, for a small $\epsilon$ and $\delta > 0$, we construct a probability distribution $f$ such that 
    $$f(b) = \frac{\epsilon}{1.2 + \delta} \mathbb{I}[b \in [0,1.2 + \delta]] + (1-\epsilon) \cdot (\sigma_1(b) + \sigma_{1.2}(b)) \cdot \frac{1}{2},$$
    where $\sigma_x$ is a $C^\infty$ function with support $[x-\delta,x+\delta]$ such that $\int_{x-\delta}^{x+\delta} \sigma_x(b)db = 1$. We then construct a Bayesian CCE as follows. We draw a number $\alpha \sim U[0,1]$, and correlate buyers' bids on $\alpha$. If buyer $i$'s type is drawn from the  component of $f$ corresponding to the ``smoothed'' component of our discrete uniform distribution on $\{1,1.2\}$, then buyer $i$ bids $H^{-1}(\alpha)$. Else if a buyer's type is drawn from the ``uniform $\epsilon$'' component of $f$, they simply play best responses $\beta(v)$, determined as the solution of the differential equation where $\beta(0) = 0$ and
    $$ v\left( (1-\epsilon) H'\left(\beta(v)\right)\beta'(v) + \frac{\epsilon}{1.2 + \delta} \right) - \beta'(v) \left((1-\epsilon) H\left(\beta(v)\right) + \frac{\epsilon v}{1.2 + \delta}\right) = 0 $$
    for any $v \in (0,1)$, except for two points where $H'$ is not defined. This then is by construction a BCCE for small enough $\epsilon$ and $\delta$.
\end{example}

\section{Linking the Discrete and the Continuous}\label{sec:cont-to-disc}

This section contains the culmination of our efforts, in the form of explicit distance bounds for the BCCE of discretised auctions to the canonical equilibrium of their continuous counterparts. Our approach can be summarised as follows; we first argue that we may take the support of the BCCE of a given discretised auction as an approximate BCCE for the continuous auction, via encoding a discretised bidding strategy as a piecewise constant function on $[0,1]$. We first show that, the feasibility of a dual constraint corresponding to a weakly increasing piecewise constant bidding strategy follows by considering it as the limit of a sequence of bidding strategies satisfying Assumption \ref{asmptn:sup-weak}. A strengthening of Proposition \ref{prop:inc-bids} then shows that with an appropriate dual solution, the distance maximal bidding strategy within the restricted support will be necessarily weakly increasing. Weak duality then provides us with the desired distance bounds, providing \emph{``translations''} of the results of the previous section to the discretised setting.

\subsection{From Discrete BCCE to Approximate Continuous BCCE}\label{sec:disc-to-approx-cont-BCCE}

Our first step is to discuss Definition \ref{def:discretisations} of families of discretised auctions in greater detail. We shall consider discretisations $(V_n,B_n,F_n)$ of a continuous auction $(V,B,F)$, with a set of valuations $V_n = \{v^n_0, v^n_1, \ldots, v^n_{|V_n|}\} \subseteq [0,v(1))$, a set of bids $B_n = \{b^n_0, b^n_1, \ldots, b^n_{|B_n|}\} \subseteq B$, where $v^n_0 = b^n_0 = 0$ and $b^n_{|B_n|} = \max B$. We shall consider the case when $F_n(v^n_k) = F(v^n_{k+1}) - F(v^n_{k})$ for each $v^n_k \in V_n$. We will also denote $v^n_{|V_n|+1} = v(1)$ for the maximum valuation. 

Now, given a vector of buyers' strategies, $\lambda_i : V_n \rightarrow B_n$, we extend each $\lambda_i$ to a right continuous function $\lambda_i : V \rightarrow B$ by setting 
$\lambda_i(v) = \lambda_i(v^n_k)$ whenever $v \in [v^n_k,v^n_{k+1})$ for some $k$. Of course, we also let $\lambda_i(v(1)) = \lambda_i(v^n_{|V_n|})$. Then, under the change of coordinates $\lambda_i(\theta) \equiv \lambda_i(F^{-1}(v))$, we obtain bidding strategies for buyers in quantile space that are piecewise constant.

Now, for a BCCE $\sigma$ for the discretised auction, consider the induced probability distribution over such piecewise constant bidding strategies $\lambda$ for the continuous auction. Types $\theta$ that correspond to some $v \in V_n$ still incur no regret, however, whenever $v(\theta) \in (v^n_k,v^n_{k+1})$, we have the possibility that some equilibrium constraint (\ref{cons:inf-eq}) is violated. However, if the discretisation is sufficiently fine, then $v(\theta) - v^n_k$ is small, and the incurred regret cannot be too large. The next proposition formalises this argument.

\begin{proposition}\label{prop:disc-to-approx-cont}
  Let $(V_n,B_n,F_n)$ be a $\delta$-fine family of discretisations of a continuous first-price or all-pay auction $A = (N,\valset,F^N,\bidset,u)$. Then for each $n \in \mathbb{N}$, a BCCE $\sigma^n$ for the discretised auction $A_n$ is a $\min\{v,4\delta(n)\}$-approximate BCCE for $A$.
\end{proposition}

\begin{proof}
  Fix a buyer $i$, type $\theta_i$, and some bid $b'_i \in B$. Following Definition \ref{def:cont-BCCE}, note that for each bidding strategy $\lambda$, each buyer $j$ of type $\theta_j$ bids $\lambda_j(v_{k_j}^n)$ for a $k_j$ which satisfies $v(\theta_j) \in [v_{k_j}^n,v^n_{k_j+1}]$ with $v(\theta_j) = v_{k_j+1}^n$ if and only if $\theta_j = 1$. We thus want to show 
  \begin{align}
    \sum_{\lambda \in \times_{i \in [N]} B_n^{V_n}} & \sigma^n(\lambda) \cdot U_i(b'_i,\lambda_{-i}|\theta_i) - U_i(\lambda|\theta_i) \label{eq:disc-bound-regret}\\
    = \sum_{\lambda \in \times_{i \in [N]} B_n^{V_n}, v_j \in V_n^{N-1}} & \sigma^n(\lambda) \cdot F_{n(-i)}(v_{-i}) \cdot \Big[\left(x_i(b'_i,\lambda_{-i}(v_{-i})) - x_i(\lambda(v(\theta_i)),\lambda_{-i}(v_{-i})) \right) \cdot v(\theta_i) \ldots \nonumber \\ & - p_i(b'_i,\lambda_{-i}(v_{-i})) + p_i(\lambda(v(\theta_i)),\lambda_{-i}(v_{-i}))\Big] \leq \min\{v(\theta_i),4\delta(n)\}. \nonumber
  \end{align}
  Given $\theta_i$, either there exists $k$ such that $v(\theta_i) \in [v_k^n,v_{k+1}^n)$, or $\theta_i = 1$ which implies that for $k = |V_n|$, $v(\theta_i) \in [v_k^n,v_{k+1}^n]$. Furthermore, by the assumption $b^n_{|B_n|} = \max B$ and finiteness of $B_n$, there exists a minimum $\hat{b} \in B_n$ such that $\hat{b} \geq b'_i$. Then as $\sigma^n$ is a BCCE for the discretised auction $A_n$,
  \begin{equation}\label{eq:disc-good-regret}
    \sum_{\lambda \in \times_{i \in [N]} B_n^{V_n}} \sigma^n(\lambda) \cdot U_i(\hat{b},\lambda_{-i}|F(v_k^n)) - U_i(\lambda|F(v_k^n)) \leq 0.
  \end{equation}
  We thus want to compare (\ref{eq:disc-bound-regret}) with (\ref{eq:disc-good-regret}) to obtain our desired bound. Now, as $x_i$ satisfies Assumption \ref{asmptn:monotone} and Assumption \ref{asmptn:zero}, whenever $b'_i > v(\theta_i)$ (\ref{eq:disc-bound-regret}) is immediately satisfied; buyer $i$ of type $\theta_i$ obtains utility at least as much as that of valuation $v_k^n$, which must be $\geq 0$ by the no-regret condition against the zero bid, but any overbidding deviation $b'_i$ can only result in non-positive utility. Therefore we may restrict attention to $b'_i \leq v(\theta_i)$. In this case, for the first-price auction, a fixed vector $\lambda$ of bidding strategies for buyers, and a fixed valuation profile $v_{-i}$ for buyers other than $i$,
  \begin{align*}
    x_i(b'_i,\lambda_{-i}(v_{-i})) \cdot v(\theta_i) - p_i(b'_i, \lambda_{-i}(v_{-i})) & = x_i(b'_i,\lambda_{-i}(v_{-i})) \cdot (v(\theta_i) - b'_i) \\
    & \leq x_i(\hat{b},\lambda_{-i}(v_{-i})) \cdot (v(\theta_i) - b'_i) \\
    & = x_i(\hat{b},\lambda_{-i}(v_{-i})) \cdot  \left[ (v_k^n - \hat{b}) + (\hat{b} - b'_i) + (v(\theta_i) - v_k^n) \right] \\
    & \leq x_i(\hat{b},\lambda_{-i}(v_{-i})) \cdot  (v_k^n - \hat{b}) + 4\delta(n).
  \end{align*}
  Here, the first inequality is by monotonicity of $x_i$ in buyer $i$'s bid, the second equality is by adding and subtracting $\hat{b}$ and $v_k^n$ to $(v(\theta_i) - b'_i)$, and the final inequality follows since $x_i : \bidset \rightarrow [0,1]$ and the assumption that the discretisation is $\delta$-fine. Specifically, $\delta$-fineness of the discretisation implies that $v_{k+1}^n - v_k^n$ and $\hat{b} - b'_i$ are both $\leq 2\delta(n)$.

  The same inequality holds for the all-pay auction through identical arguments, by noting that $p_i(\hat{b},\lambda_{-i}(v_{-i})) = \hat{b}$ and $p_i(\hat{b},\lambda_{-i}(v_{-i})) = b'_i$, while $x_i$ does not change between first-price and all-pay auctions. Taking expectations over $\lambda$ and $\theta_{-i}$, we conclude that 
  $$ \sum_{\lambda \in \times_{i \in [N]} B_n^{V_n}}  \sigma^n(\lambda) \cdot U_i(b'_i,\lambda_{-i}|\theta_i) \leq \sum_{\lambda \in \times_{i \in [N]} B_n^{V_n}} \sigma^n(\lambda) \cdot U_i(\hat{b},\lambda_{-i}|F(v_k^n)) + 4\delta(n)$$
  for both first-price and all-pay auctions. Meanwhile, as $v(\theta_i) \geq v_{k}^n$ and $\lambda_i(v(\theta_i)) = \lambda_i(v_k^n)$, we have 
  $ -U_i(\lambda|\theta_i) \leq -U_i(\lambda|F(v_k^n))$.
  Combining the two, we conclude that 
  $$ \sum_{\lambda \in \times_{i \in [N]} B_n^{V_n}} \sigma^n(\lambda) \cdot U_i(b'_i,\lambda_{-i}|\theta_i) - U_i(\lambda|\theta_i) \leq 4\delta(n).$$
  Meanwhile, as $x_i(b'_i,\lambda_{-i}(v_{-i})) \leq 1$, we have $U_i(b'_i,\lambda_{-i}|\theta_i) \leq v(\theta_i)$ and, by Assumption \ref{asmptn:zero}, $U_i(\lambda|\theta_i) \geq 0$. Together, we also infer that  
  $$ \sum_{\lambda \in \times_{i \in [N]} B_n^{V_n}} \sigma^n(\lambda) \cdot U_i(b'_i,\lambda_{-i}|\theta_i) - U_i(\lambda|\theta_i) \leq v(\theta_i).$$
\end{proof}

Proposition \ref{prop:disc-to-approx-cont} implies that any BCCE $\sigma^n$ of a discretised auction is also an approximate BCCE for the continuous auction. Therefore, $\sigma^n$ is a solution to (\ref{opt:inf-LP-primal-gen}), which implies that to bound its distance to the canonical equilibrium of the auction, we may evaluate the value of (\ref{opt:dual-inf-2}). Following our analysis in Section \ref{sec:theory-cont}, we will want to impose the modified Wasserstein-$2$ distance considered in Proposition \ref{prop:wasserstein-prior} as the objective of (\ref{opt:inf-LP-primal-gen}), making the appropriate assumptions on the prior. We will also again relax (\ref{opt:inf-LP-primal-gen}) by only enforcing the constraints $\epsilon_i(\theta_i,\beta(\theta_i))$, and allow ourselves the freedom to modify the support $\Lambda_n$ as desired to a finite set of bidding strategies.

This results in our final infinite primal LP,
\begin{align}\label{opt:LP-disc-cont}
  \max_{\sigma \geq 0} \sum_{\lambda \in \Lambda_n} \sigma(\lambda) \cdot \sum_{i \in [N]} \int_0^1 d\theta_i \cdot \theta^{N-2} \cdot \frac{(\lambda_i(\theta_i)-\beta(\theta_i))^2}{2} & \textnormal{ subject to} \\
  \sum_{\lambda \in \Lambda_n} \sigma(\lambda) & = 1 \tag{$\gamma$} \label{cons:inf-disc-cont-gamma} \\
  \forall \ i \in [N], \theta_i \in [0,1], \sum_{\lambda \in \Lambda_n} \sigma(\lambda) \cdot \left( U_i(\beta,\lambda_{-i}|\theta_i) - U_i(\lambda|\theta_i) \right) & \leq \min\{v(\theta_i),4\delta(n)\}. \tag{$\epsilon_i(\theta_i)$} \label{cons:inf-disc-cont-eq}
\end{align}
To bound the value of (\ref{opt:LP-disc-cont}), we shall in fact want to use the very same dual solution we used in Section \ref{sec:weak-uniqueness}. Therefore, we will want to find $\epsilon_i(\theta_i)$ which is, setting $\gamma = 0$, a solution of
\begin{align}
  \inf_{\epsilon \geq 0, \gamma} \gamma + \sum_{i \in [N]} \int_0^1 d\theta_i \cdot \frac{\epsilon_i(\theta_i)}{v(\theta_i)} \cdot \min\{v(\theta_i),4\delta(n)\} \textnormal{ subject to} &  \label{opt:dual-disc-cont-inf}\\
  \gamma + \sum_{i \in [N]} \int_0^1 d\theta_i \cdot\frac{\epsilon_i(\theta_i)}{v(\theta_i)}  \left( U_i(\beta,\lambda_{-i}|\theta_i) - U_i(\lambda|\theta_i) \right)  \geq \sum_{i \in [N]} \int_0^1 d\theta_i &  \cdot \frac{\theta^{N-2}(\lambda_i(\theta_i)-\beta(\theta_i))^2}{2} \nonumber \\
  \ \forall \ \lambda \in \Lambda_n. & \label{cons:inf-disc-cont-dual-gen} \tag{$\sigma(\lambda)$}
\end{align}

\subsection{Black-Box Distance Bounds for Discretised Auctions}\label{sec:bounds-theory}

Comparing (\ref{opt:dual-disc-cont-inf}) with (\ref{opt:dual-inf-weak}), the objective now incurs a penalty for the magnitude of $\epsilon_i$, and the constraint set $\Lambda$ is exchanged for $\Lambda_n$. As a consequence, for an appropriate feasible dual solution, we obtain our desired distance bounds immediately.

\begin{proposition}\label{prop:final-bounds-2}
  Suppose that $\epsilon_i(\theta_i) = \epsilon$, $\gamma = 0$ is a feasible solution of (\ref{opt:dual-disc-cont-inf}) for each $n \in \mathbb{N}$. Then for any sequence $\sigma^n$ of BCCEs of a $\delta$-fine family of discretised auctions,
  $$  \sum_{\lambda \in \Lambda_n} \sigma(\lambda) \cdot \sum_{i \in [N]} \int_0^1 d\theta_i \cdot \theta^{N-2} \cdot \frac{(\lambda_i(\theta_i)-\beta(\theta_i))^2}{2} \leq N\epsilon \left[ F(4\delta(n)) +  \int_{F(4\delta(n))}^1 d\theta_i \cdot \frac{4\delta(n)}{v(\theta_i)} \right],$$ 
  which $\rightarrow 0$ as $n \rightarrow \infty$.
\end{proposition}

\begin{proof}
  Each $\Lambda_n$ is a finite family of step functions, thus $\sigma^n$ is a probability distribution over a finite set; i.e. defines a Borel measure with respect to the discrete topology on $\Lambda_n$. Moreover, for each fixed $\lambda \in \Lambda_n$, 
  $$\frac{1}{v(\theta_i)}  \left( U_i(\beta,\lambda_{-i}|\theta_i) - U_i(\lambda|\theta_i) \right)$$
  is continuous over $\theta_i \in [0,1]$ except potentially at finitely many points. As a consequence, the conditions of the Fubini-Tonelli theorem are satisfied and weak duality holds. The bound is then obtained by evaluating the integral, 
  \begin{align*}
    & \sum_{\lambda \in \Lambda_n} \sigma(\lambda) \cdot \sum_{i \in [N]} \int_0^1 d\theta_i \cdot \theta^{N-2} \cdot \frac{(\lambda_i(\theta_i)-\beta(\theta_i))^2}{2} \\ 
    \leq & \sum_{\lambda \in \Lambda_n} \sigma(\lambda) \cdot \sum_{i \in [N]} \int_0^1 d\theta_i \cdot\frac{\epsilon}{v(\theta_i)}  \left( U_i(\beta,\lambda_{-i}|\theta_i) - U_i(\lambda|\theta_i) \right) \\
    =& \sum_{i \in [N]} \int_0^1 d\theta_i \cdot\frac{\epsilon}{v(\theta_i)} \cdot \sum_{\lambda \in \Lambda_n} \sigma(\lambda) \cdot  \left( U_i(\beta,\lambda_{-i}|\theta_i) - U_i(\lambda|\theta_i) \right) \\
    \leq& \sum_{i \in [N]} \int_0^1 d\theta_i \cdot\frac{\epsilon}{v(\theta_i)} \cdot \min\{v(\theta_i),4\delta(n)\} \\
    =& N\epsilon \left( \int_0^{v^{-1}(4\delta(n))} d\theta_i + \int_{v^{-1}(4\delta(n))}^1 d\theta_i \cdot \frac{4\delta(n)}{v(\theta_i)} \right).
  \end{align*} 
  By definition, $v(\theta) = F^{-1}(\theta)$ for the cumulative distribution function $F$ of the symmetric prior distribution of the continuous auction, which gives us the desired dependence on $n$. As $\delta(n) \rightarrow 0$ and $F$ is continuous with $F(0) = 0$, we have the desired result.
\end{proof}

All that remains, then, is to show that such $\epsilon_i, \gamma$ is feasible for $\Lambda_n$ for \underbar{\emph{every}} $n \in \mathbb{N}$. Following our discussion at the beginning of Section \ref{sec:disc-to-approx-cont-BCCE}, we shall assume that each $\Lambda_n$ contains bidding strategies $\lambda$ induced by the set of bidding strategies in the discretised auction; maintaining our symmetry assumption, restricted to those which are symmetric amongst the buyers. The resulting strategies are right-continuous, piecewise constant, and discontinuous at $\leq |V_n|$ points on $[0,1]$, and we will assume that any element of $\Lambda_n$ satisfies this. By Proposition \ref{prop:zero-bid} we will restrict attention to strategies normalised to $\lambda(0) = 0$ for the discretised auctions; for the corresponding strategies in the continuous auction, this implies that $\lambda(\theta_i) = 0$ for any $\theta \in [0,F(v^n_1))$. Our assumptions on $\Lambda_n$ can then be summarised:

\begin{assumption}\label{asmptn:final}
  For the continuous auction $A$ and its discretisation $A_n$, $\Lambda_n \subseteq \times_{i \in N} B_n^{[0,1]}$ contains bidding strategies $\lambda$ which satisfy:
  \begin{enumerate}
    \item (Symmetry) For each buyer $i,j$, $\lambda_i = \lambda_j$.
    \item (Discrete Containment) Let $\lambda^n \in \times_{i \in N} B_n^{V_n}$ be symmetric and normalised. Then there exists $\lambda \in \Lambda_n$ such that for any buyer $i$, and any type $v^n_k \in V_n$, $\lambda_i(\theta_i) = \lambda^n_i(F(v^n_k))$ for any $\theta_i \in [F(v^n_k),F(v^n_{k+1}))$. Moreover, $\lambda_i$ is continuous at $1$.
    \item (Normalisation) For each buyer $i$, $\lambda_i(\theta_i) = 0$ for any $\theta_i \in [0,F(v^n_1))$.
  \end{enumerate}
\end{assumption}

In Assumption \ref{asmptn:final}, we are finally able to drop any assumption of monotonicity. This will be accomplished by an extension of Proposition \ref{prop:inc-bids} to bidding strategies in $\Lambda_n$. However, the uniform prior assumption in Proposition \ref{prop:inc-bids} will necessitate us to consider \emph{monotone rearrangements} of any of its elements \emph{for the continuous auction} in our proof.

\begin{definition}
  Given $\lambda \in \Lambda_n$, for each buyer $i$, suppose the \emph{inverse image} $\lambda_i^{-1}([0,1])$ is a finite ordered subset $\{b^{in}_0, \ldots, b^{in}_{C(i)}\}$ of $B_n$ with $b^{in}_0 = 0$. For each integer $0 \leq k \leq C(i)$, let $\phi_{ik}$ be the measure of $\lambda_i^{-1}(b^{in}_k)$, i.e. the sum of the lengths of its constituent intervals. Then the \textbf{monotonisation} of $\lambda_i$, $\mon \lambda_i$, is defined such that 
  $$ \forall \ 0 \leq k \leq C, \forall \ \theta_i \in \left[ \sum_{k' = 0}^{k-1} \phi_{ik}, \sum_{k' = 0}^{k} \phi_{ik} \right), \mon \lambda_i(\theta_i) = b^{in}_k,\textnormal{ and }\mon \lambda_i(1) = b^{in}_{C(i)}.$$ We extend the monotonisation operator to $\Lambda_n$ by letting $(\mon \lambda)_j = \mon \lambda_j$ for any buyer $j$. 
\end{definition}

To prove our extension result, we first want to show an analogue of Lemma \ref{lem:v-solution}; what we want to accomplish is to convert a $\lambda$ which satisfies Assumption \ref{asmptn:final} to a symmetric bidding strategy for which each $\lambda_i$ is piecewise strictly increasing and normalised, when checking for dual feasibility.

\begin{lemma}\label{lem:prelim-inc}
  Suppose, that $\lambda$ satisfies Assumption \ref{asmptn:final}. Then there exists a symmetric and normalised sequence of bidding strategies $\lambda_m$, converging uniformly to $\lambda$ on $[0,1]$, such that each $\lambda_{mi}$ is differentiable with a strictly positive derivative at any point of continuity of $\lambda_i$ on $(0,1)$. Moreover, for both the first-price auction and the all-pay auction, 
  \begin{equation}
    \lim_{m\rightarrow \infty} \int_0^1 d\theta_i \cdot \frac{1}{v(\theta_i)} \cdot \left( U_i(\beta,\lambda_{m(-i)}|\theta_i) - U_i(\lambda_m|\theta_{i}) \right) \leq \int_0^1 d\theta_i \cdot \frac{1}{v(\theta_i)} \cdot \left( U_i(\beta,\lambda_{-i}|\theta_i) - U_i(\lambda|\theta_i) \right). \label{con:good-conv-disc}
  \end{equation}
\end{lemma}

\begin{proof}
  For convenience, suppose that the points of discontinuity of each $\lambda_i$ are rational numbers; if not, by the density of the rationals we may modify each $\lambda_i$ to be so at an arbitrarily small change on both sides of (\ref{con:good-conv-disc}) and consider the limits of such $\lambda_i$. We then divide $[0,1)$ into subintervals $[\phi_k,\phi_{k+1})$ of length $1/\ell$ for some $\ell \in \mathbb{N}$, such that $\lambda_i(\theta_i)$ is constant on $[\phi_k,\phi_{k+1})$ for each buyer $i$ and each $k$. Now, choose for each buyer $i$ a permutation $s_i$ of $[\ell]$, such that: 
  \begin{enumerate}
    \item For each $\theta_i = \phi_k + \Delta\theta_i \in [\phi_k,\phi_{k+1})$, $\mon \lambda_i(\theta_i) = \lambda_i(\phi_{s_i^{-1}(k)})$.
    \item For each $b^n \in B_n$, $s_i$ is increasing on $\{ k \in [\ell] \ | \ \lambda_i(\phi_k) = b^n\}$.
  \end{enumerate}
  By condition (1), $s_i$ induces a bijection $\mu_i : [0,1) \rightarrow [0,1)$, defined as $\mu_i(\phi_k + \Delta\theta_i) = \phi_{s_i(k)} + \Delta\theta_i$ for each $\Delta\theta_i \in [0,1/\ell)$, such that $\mon \lambda_i(\theta_i) = \lambda_i(\mu_i(\theta_i))$. Condition (2) guarantees that if $\theta_i < \theta'_i$ and $\lambda_i(\theta_i) = \lambda_i(\theta'_i)$, then $\mu_i(\theta_i) < \mu_i(\theta'_i)$. Our goal is then to modify $\lambda_i$ in a small amount, such that the bid of $\theta'_i$ always wins against that of $\theta_i$; and $\mu_i$ should be faithful to this phenomenon. 
  
  Towards this end, for each $m \in \mathbb{N}$, define a scale constant $\nu_m$ such that 
  $$\nu_m \leq \frac{1}{2m \cdot \min_{\kappa \neq \kappa'} |b^n_{\kappa} - b^n_{\kappa'}|}.$$
  Moreover, since $\beta$ is a continuous bijection between $[0,1]$ and $[0,\beta(1)]$, $\beta^{-1}$ is also a continuous bijection; hence it is uniformly continuous. Pick $\nu_m$ also small enough such that for any $\beta(\theta_i), \beta(\theta'_i) \in [0,\beta(1)]$, if $|\beta(\theta_i) - \beta(\theta'_{i})| \leq \nu_m$, then $|\theta_i - \theta'_i| \leq 1/2m|V_n|$. We then let
  $$ \lambda_{mi}(\theta_i) = \lambda_i(\theta_i) + \nu_m \cdot \frac{v(\theta_i)}{v(1)}.$$

  We compare the terms in (\ref{con:good-conv-disc}) separately, as in the proof of Lemma \ref{lem:v-solution}. We first note that $U_i(\beta|\lambda_{m(-i)}|\theta_i) = U_i(\beta|\lambda_{-i}|\theta_i)$ whenever $\beta(\theta_i) \notin \lambda_i([0,1])+[0,\nu_m]$ by construction. Moreover, $\beta^{-1}[\lambda_i([0,1])+[0,\nu_m]]$ is a union of at most $|V_n|$ closed intervals, of total measure at most $1/2m$. As a consequence, 
  $$ \int_0^1 d\theta_i \cdot \frac{1}{v(\theta_i)} \cdot \left| U_i(\beta,\lambda_{n(-i)}|\theta_i) - U_i(\beta|\theta_i) \right| \leq \frac{1}{2m}.$$

  It remains to show how the $U_i(\lambda|\theta_i)$ term compares to $U_i(\lambda_m|\theta_i)$. Note that the probability a bid $\lambda_i(\theta_i)$ wins the item for buyer $i$ remains unchanged when buyers $-i$ change their bidding strategies to $\mon \lambda_{-i}$, and likewise for the expected payment. Thus, if $\lambda_i(\theta_i) = L$, buyer $i$ wins the item with probability 
  $$ \sum_{m' = 0}^{N-1} \frac{1}{m'+1} \binom{N-1}{m'} \underline{\theta}^{N-1-m'} (\overline{\theta} - \underline{\theta})^{m'}$$
  where $[\underline{\theta},\overline{\theta}]$ is the closure of the interval $\mu_i\lambda^{-1}(L)$. Meanwhile, under bidding strategies $\lambda_m$, buyer $i$ wins the item with probability $\mu_i(\theta_i)^{N-1}$; $\lambda_m$ is injective into its image by choice of $\nu_m$. Thus, adapting the argument in Lemma \ref{lem:v-solution} and noting that $d\mu_i(\theta_i) = d\theta_i$ everywhere except for a finite set on $[0,1)$, for the first price auction
  \begin{align*}
    & \int_{\lambda^{-1}(L)} d\theta_i \cdot \frac{1}{v(\theta_i)} \cdot \left( U_i(\lambda|\theta_i)-U_i(\lambda_m|\theta_i) \right) \\ 
    = & \int_{\lambda^{-1}(L)} d\theta_i \cdot \frac{1}{v(\theta_i)} \cdot \Bigg[  \left( \sum_{m' = 0}^{N-1} \frac{1}{m'+1} \binom{N-1}{m'} \underline{\theta}^{N-1-m'} (\overline{\theta} - \underline{\theta})^{m'} \right) \cdot (v(\theta_i) - \lambda(\theta_i)) \ldots \\ & -  \mu_i(\theta_i)^{N-1} (v(\theta_i) - \lambda_m(\theta_i)) \Bigg]\\
    = & \int_{\underline{\theta}}^{\overline{\theta}} d\theta_i \cdot \frac{1}{v\mu_i^{-1}(\theta_i)} \cdot \Bigg[  \left( \sum_{m' = 0}^{N-1} \frac{1}{m'+1} \binom{N-1}{m'} \underline{\theta}^{N-1-m'} (\overline{\theta} - \underline{\theta})^{m'} \right) \cdot (\cancel{v\mu_i^{-1}(\theta_i)} - L) \ldots \\ & -  \theta_i^{N-1} (\cancel{v\mu_i^{-1}(\theta_i)} - L - \nu_m \cdot v\mu_i^{-1}(\theta_i) / v(1)) \Bigg]\\
    = & \int_{\underline{\theta}}^{\overline{\theta}} d\theta_i \cdot \frac{L}{v\mu_i^{-1}(\theta_i)} \cdot \left( \theta_i^{N-1} - \sum_{m' = 0}^{N-1} \frac{1}{m'+1} \binom{N-1}{m'} \underline{\theta}^{N-1-m'} (\overline{\theta} - \underline{\theta})^{m'} \right)  + \int_{\underline{\theta}}^{\overline{\theta}} d\theta_i \cdot \frac{\nu_m \theta_i^{N-1}}{v(1)}.
  \end{align*}
  In the last expression, as in Lemma \ref{lem:v-solution}, the first term is $\leq 0$; this time since $\mu_i^{-1}$ is increasing on $\lambda^{-1}(L)$ by construction. Meanwhile the second term is $O(1/m)$ by choice of $\nu_m$, which implies the desired result. The case of the all-pay auction follows through identical arguments.
\end{proof}

Lemma \ref{lem:prelim-inc} allows us to focus on bidding strategies which only fail to increase at their step discontinuities. Our \emph{transport lemma} then shows that, to bound the modified Wasserstein-$2$ distance of BCCE for a discretised auction to the equilibrium of the continuous auction, under suitable assumptions on the prior distribution we only need to focus on weakly increasing bidding strategies.

\begin{lemma}[Transport Lemma]\label{lem:transport}
  Suppose that $\lambda \in \Lambda_n$ satisfies Assumption \ref{asmptn:final}, then:
  \begin{enumerate}
    \item For the first-price auction,
    \begin{align*}
      & \int_0^1 d\theta_i \cdot \frac{1}{v(\theta_i)} \left( U_i(\beta,\lambda_{-i}|\theta_i) - U_i(\beta,\mon \lambda_{-i}|\theta_i) - U_i(\lambda|\theta_i) + U_i(\mon \lambda |\theta_i) \right)\\ \geq & \int_0^1 d \theta_i \cdot \theta^{N-2} \cdot \left( \frac{(\lambda_i(\theta_i)-\beta(\theta_i))^2}{2} -  \frac{(\mon \lambda_i(\theta_i)-\beta(\theta_i))^2}{2} \right)  \cdot D^{FP}, \\
      & \textnormal{ where } D^{FP} = \inf_{\theta \in (0,1]} \frac{v'(\theta) \theta^2}{(2N-3) v(\theta)^3}.
    \end{align*} 
    A sufficient condition for $D^{FP} > 0$ is the convexity of $v(\theta)$; equivalently $v'(\theta)\theta \geq v(\theta)$ on $[0,1]$ and the concavity of the prior distribution.
    \item For the all-pay auction,
    \begin{align*}
      & \int_0^1 d\theta_i \cdot \frac{1}{v(\theta_i)} \left( U_i(\beta,\lambda_{-i}|\theta_i) - U_i(\beta,\mon \lambda|\theta_i) - U_i(\lambda_{-i}|\theta_i) + U_i(\mon \lambda |\theta_i) \right) \\ 
      \geq & \int_0^1 d \theta_i \cdot \theta^{N-2} \cdot \left( \frac{(\lambda_i(\theta_i)-\beta(\theta_i))^2}{2} -  \frac{(\mon \lambda_i(\theta_i)-\beta(\theta_i))^2}{2} \right)  \cdot D^{AP}, \\
      & \textnormal{ where } D^{AP} = \inf_{\theta \in (0,1]} \frac{ v'(\theta)}{(2N-3)v(\theta)^3}.
    \end{align*} 
    As a consequence, a sufficient condition for $D^{AP} > 0$ to hold is for $1/v(\theta)$ to have an upper bound on its derivative on $(0,1]$.
  \end{enumerate}
\end{lemma}

\begin{proof}
  We will again assume that the points of discontinuity of each $\lambda_i$ are contained in the rationals; the case of irrational points of discontinuity are then covered by density arguments as in the proof of Lemma \ref{lem:prelim-inc}. Given $\lambda$, take a sequence $\lambda_m$ as in the statement of Lemma \ref{lem:prelim-inc}. If the result holds for every such $\lambda_m$, then the desired inequality holds since the Wasserstein-$2$ distance we consider is continuous in the $\sup$-norm of $\lambda$. Since the inequality holds trivially whenever $\lambda_m = \mon \lambda_m$, suppose this is not the case.
  
  Just as in the proof of Lemma \ref{lem:prelim-inc}, divide the interval $[0,1)$ to subintervals of length $1/\ell$, and identically define the permutation $s_i : [\ell] \rightarrow [\ell]$ and the induced bijection $\mu_i : [0,1) \rightarrow [0,1)$. Then, by Lemma \ref{lem:rearr}, there exists $k < k'$ such that $s_i(k) \geq k'$ and $s_i(k') \leq k$. Define the transport function $t : [0,1) \rightarrow [0,1)$,
  $$ t(\theta_i) = \begin{cases}
    \theta_i & \theta_i \notin [\phi_k,\phi_{k+1}) \cup [\phi_{k'},\phi_{k'+1}), \\
    \theta_i - \phi_k + \phi_{k'} & \theta_i \in [\phi_k,\phi_{k+1}), \\
    \theta_i - \phi_{k'} + \phi_k & \theta_i \in [\phi_{k'},\phi_{k'+1}).
  \end{cases}$$
  Intuitively, we transport the bids between the intervals $[\phi_k,\phi_{k+1})$ and $[\phi_{k'},\phi_{k'+1})$, letting $\lambda'_m = \lambda_m \circ t$. Note that the transport is order correcting for the bidding strategy $\lambda_m$; as $\mon \lambda_m = \lambda_m \mu$,
  $$ \lambda_m(\phi_{k'}) = \mon \lambda_m(\phi_{s_i(k')}) < \mon \lambda_m(\phi_{s_i(k)}) = \lambda_m(\phi_k).$$
  As a consequence, $\mon \lambda = \lambda \circ t_1 t_2 \ldots t_M$ for some finite sequence of transport functions $t_s$. If we can show that, for a given transport function,
  \begin{align*}
    & \int_0^1 d\theta_i \cdot \frac{1}{v(\theta_i)} \left( U_i(\beta,\lambda_{m(-i)}|\theta_i) - U_i(\beta,(\lambda_m \circ t)_{-i}|\theta_i) - U_i(\lambda_{m}|\theta_i) + U_i(\lambda_m \circ t |\theta_i) \right)  \\ 
    \geq & \int_0^1 d \theta_i \cdot \theta^{N-2} \cdot \left( \frac{(\lambda_{mi}(\theta_i)-\beta(\theta_i))^2}{2} - \frac{((\lambda_m\circ t)_i(\theta_i)-\beta(\theta_i))^2}{2}  \right) \cdot D^{\mathcal{A}}, 
  \end{align*} 
  then an inductive argument implies our result via evaluation of a telescoping sum. Note that the $U_i(\beta,\lambda_{m(-i)})$ terms cancel out on the right hand side. We thus only need to bound the change in the $U_i(\lambda_m|\theta_i)$ terms. Further note that the utilities and the Wasserstein-$2$ distance integrand change only for types $\theta_i$ such that $t(\theta_i) \neq \theta_i$, which implies that we need to show that 
  \begin{align}
    & \int_{\theta_i \neq t(\theta_i)} d\theta_i \cdot \frac{1}{v(\theta_i)} \left(U_i(\lambda_m \circ t |\theta_i)- U_i(\lambda_{m}|\theta_i)\right) \nonumber \\ \geq & \int_{\theta_i \neq t(\theta_i)} d\theta_i \cdot \theta^{N-2} \cdot  \left( \frac{(\lambda_{mi}(\theta_i)-\beta(\theta_i))^2}{2} - \frac{((\lambda_m\circ t)_i(\theta_i)-\beta(\theta_i))^2}{2}  \right) \cdot D^{\mathcal{A}}. \label{eq:transport-intermediate}
  \end{align}
  The right hand side then can be expanded and bounded above,
  \begin{align}
    &\int_0^{1/\ell} d\theta \cdot D^\mathcal{A} \cdot \Bigg[ \left( \beta(\phi_{k'}+\theta) (\phi_{k'}+\theta)^{N-2} - \beta(\phi_k + \theta)(\phi_k + \theta)^{N-2}\right) (\lambda_m(\phi_k+\theta) - \lambda_m(\phi_{k'}+\theta)) \nonumber \\
    &\ldots - \left( (\phi_{k'}+\theta)^{N-2} - (\phi_k+\theta)^{N-2} \right) \left( \frac{\lambda_m(\phi_k+\theta)^2 -\lambda_m(\phi_{k'}+\theta)^2 }{2} \right) \Bigg] \leq \ldots \nonumber \\
    & \int_0^{1/\ell} d\theta \cdot D^\mathcal{A} \cdot \left( \beta(\phi_{k'}+\theta) (\phi_{k'}+\theta)^{N-2} - \beta(\phi_k + \theta)(\phi_k + \theta)^{N-2}\right) (\lambda_m(\phi_k+\theta) - \lambda_m(\phi_{k'}+\theta)). \label{eq:Wasserstein-compare}
  \end{align}

  Now consider the first-price auction. Expanding the left hand side of (\ref{eq:transport-intermediate}),
  \begin{align}
    & \int_{\theta_i \neq t(\theta_i)} d\theta_i \cdot \frac{1}{v(\theta_i)} \left(U_i(\lambda_m \circ t |\theta_i)- U_i(\lambda_{m}|\theta_i)\right) \nonumber \\ = & \int_0^{1/\ell} d\theta_i \cdot \left(\frac{1}{v(\phi_k + \theta)} - \frac{1}{v(\phi_{k'} + \theta)}\right)\left( \lambda_m(\phi_k + \theta) (\phi_{s_i(k)}+\theta)^{N-1} - \lambda_m(\phi_{k'}+\theta)(\phi_{s_i(k')}+\theta)^{N-1}  \right) \nonumber \\
    \geq & \int_0^{1/\ell} d\theta_i \cdot \left(\frac{1}{v(\phi_k + \theta)} - \frac{1}{v(\phi_{k'} + \theta)}\right)\left( \lambda_m(\phi_k + \theta) (\phi_{k'}+\theta)^{N-1} - \lambda_m(\phi_{k'}+\theta)(\phi_{k}+\theta)^{N-1}  \right).\label{eq:first-price-compare-1}
  \end{align}
  The inequality here holds by our choice of an order correcting swap which satisfies $s_i(k') \leq k < k' \leq s_i(k)$. We thus want to show that (\ref{eq:first-price-compare-1}) is no less than (\ref{eq:Wasserstein-compare}), which we again accomplish by bounding the integrands pointwise. Writing $\theta^{\pm}$ for the higher and lower types as before, and using the shorthand $\lambda_m(\phi_{k}+\theta) = \lambda^+, \lambda_m(\phi_{k'}+\theta) = \lambda^{-}$, we need to show that 
  $$ \left( \frac{1}{v(\theta^-)} - \frac{1}{v(\theta^+)}\right) (\theta^{+ \ N-1} \lambda^+ - \theta^{- \ N-1} \lambda^{-}) \geq D^{FP} \cdot (\beta(\theta^+) \theta^{+ \ N-2} - \beta(\theta^-) \theta^{- \ N-2}) (\lambda^+ - \lambda^-)$$
  for any $\lambda^+ \geq \lambda^-$ and $\theta^+ \geq \theta^- > 0$. Both sides of the inequality are linear in $\lambda^+$, and equality holds when $\lambda^+ = \lambda^-$. So our desired inequality holds for $D^{FP} > 0$ if and only if 
  $$ \left( \frac{1}{v(\theta^-)} - \frac{1}{v(\theta^+)}\right) \theta^{+ \ N-1} \geq D^{FP} \cdot (\beta(\theta^+) \theta^{+ \ N-2} - \beta(\theta^-) \theta^{- \ N-2})$$
  for any $\theta^+ \geq \theta^- > 0$. Again, equality holds whenever $\theta^+ = \theta^-$, so a sufficient condition is for the left hand side to be increasing faster than the right hand side in $\theta^+$. Differentiating both sides, we get 
  $$ (N-1) \left( \frac{1}{v(\theta^-)} - \frac{1}{v(\theta^+)}\right) \theta^{+ \ N-2} - \theta^{+ \ N-1} \frac{d}{d\theta^+} \left( \frac{1}{v(\theta^+)}\right) \geq D^{FP} \cdot \frac{d}{d\theta^+} \left( \beta(\theta^+) \theta^{+ \ N-2} \right).$$
  A sufficient condition for this to hold is 
  $$ D^{FP} \leq \inf_{\theta \in (0,1]} \frac{\theta^{N-1} \frac{d}{d\theta} \left( \frac{-1}{v(\theta)}\right)}{\frac{d}{d\theta} \left( \beta(\theta) \theta^{N-2} \right)},$$
  as the first term on the left hand side is positive and $\beta(\theta)\theta^{N-2}$ is increasing in $\theta$. We remark that $\beta(\theta) \leq v(\theta)$ and $\theta \beta'(\theta) = (N-1) (v(\theta)- \beta(\theta)) \leq (N-1)v(\theta)$. This implies the upper bound on the denominator,
  \begin{align*}
    \frac{d}{d\theta} (\beta(\theta)\theta^{N-2}) & = \beta'(\theta) \theta^{N-2} + (N-2) \beta(\theta) \theta^{N-3} \\ & \leq \theta^{N-3} \left(N-1 + N-2\right) v(\theta) = \theta^{N-3} \left(2N-3\right) v(\theta).
  \end{align*}
  We thus obtain our final bound on $D^{FP}$, 
  $$ D^{FP} \leq \inf_{\theta \in (0,1]} \frac{v'(\theta) \theta^2}{(2N-3) v(\theta)^3} = \inf_{\theta \in (0,1]} \frac{1}{2N-3} \cdot \frac{\theta v'(\theta)}{v(\theta)} \cdot \frac{\theta}{v(\theta)} \cdot \frac{1}{v(\theta)}.$$
  Whenever $v$ is convex, $\theta v'(\theta) / v(\theta) \geq 1$ and $\theta/v(\theta), 1/v(\theta) \geq 1/v(1)$, which provides the sufficient condition.
\end{proof}

By the transport lemma, potentially up to a constant factor deterioration of the Wasserstein-$2$ distance we consider, we may restrict attention to strictly increasing bidding strategies on $[0,1]$ with step discontinuities to check the validity of our dual solution. For the final step of our dual feasibility proof, we will want to smooth out these discontinuities at vanishing loss, which will allow us to invoke Theorem \ref{thm:thm-weak-uniq} and Proposition \ref{prop:wasserstein-prior} from Section \ref{sec:weak-uniqueness}. The high level idea is that, in our previous proofs of uniqueness, steepness of the bidding strategies were cancelled out by our change of variables $\theta \rightarrow \lambda^{-1}\beta(\theta)$; so replacing step discontinuities with very steep smooth curves should not hurt our result.

\begin{lemma}\label{lem:smoothing}
  Suppose for a continuous first-price or all-pay auction $A$ that $\lambda$ are buyers' bidding strategies, such that $\lambda_i : [0,1] \rightarrow B$ is strictly increasing and satisfies Assumption \ref{asmptn:sup-weak} except for finitely many step discontinuities. Further suppose that all step discontinuities of $\lambda_i$ are disjoint from $[0,F(v^n_1))$ for some $v^n_1 > 0$. Then there exists a sequence $\lambda_m$ satisfying Assumption \ref{asmptn:sup-weak} such that the dual constraint (\ref{cons:inf-disc-cont-dual-gen}) is faithful to its limit on both sides,
    \begin{align*}
      \lim_{m\rightarrow \infty} \int_0^1 d\theta_i \cdot \frac{1}{v(\theta_i)} \cdot \left( U_i(\beta,\lambda_{m(-i)}|\theta_i) - U_i(\lambda_m|\theta_{i}) \right) & = \int_0^1 d\theta_i \cdot \frac{1}{v(\theta_i)} \cdot \left( U_i(\beta,\lambda_{-i}|\theta_i) - U_i(\lambda|\theta_i) \right), \\
      \lim_{m\rightarrow \infty} \int_0^1 d\theta_i \cdot \frac{\theta^{N-2}(\lambda_{mi}(\theta_i)-\beta(\theta_i))^2}{2} & =  \int_0^1 d\theta_i \cdot \frac{\theta^{N-2}(\lambda_{i}(\theta_i)-\beta(\theta_i))^2}{2}.
    \end{align*}
\end{lemma}

\begin{proof}
  Let $\{d_1,\ldots,d_{|\mathcal{D}|}\} = \mathcal{D} \subseteq (0,1]$ be the ordered finite set of discontinuities of $\lambda_i$, and denote $d_0 = 0$. Let $d = \min_{1 \leq k \leq |\mathcal{D}|} d_k - d_{k-1}$. For each $m$, define $\lambda_m$ such that for any buyer $i$,
  $$ \lambda_{mi}(\theta) = \begin{cases}
    \lambda_i(\theta) & \theta \notin (d_k-d/2m,d_k) \textnormal{ for any } d_k \in \mathcal{D}, \textnormal{ and}\\
    \lambda_i(\theta) + c_{km}(\theta) & \exists \ d_k \in \mathcal{D}, \theta \in (d_k-d/2m,d_k),
  \end{cases}$$
  where $c_{km}(\theta)$ is the integral of a $C^\infty$ bump function with support $(d_k-d/2m,d_k)$ such that $c_{km}(d_k) = \lim_{h \downarrow 0} \lambda_i(d_k+h) - \lambda_i(d_k-h)$.

  We first evaluate the easier limit. To obtain the bound on the Wasserstein-$2$ distance, note that the measure of the set on which $\lambda_i \neq \lambda_{mi}$ equals $d|\mathcal{D}|/2m$, on which $0 \leq \lambda_m(\theta) - \lambda(\theta) \leq \max B$. This implies a bound $(\max B)^2 d |\mathcal{D}| / 2m$, which $\rightarrow 0$ as $m \rightarrow \infty$. Then to bound the change in expected utilities, we first note that 
  \begin{align*}
    & \left|\int_0^1 d\theta_i \cdot \frac{1}{v(\theta_i)} \cdot \left( U_i(\beta,\lambda_{m(-i)}|\theta_i) - U_i(\lambda_m|\theta_{i}) - U_i(\beta,\lambda_{-i}|\theta_i) + U_i(\lambda|\theta_i) \right) \right|\\
    \leq & \int_0^1 d\theta_i \cdot \frac{1}{v(\theta_i)} \cdot \left( \left| U_i(\beta,\lambda_{m(-i)}|\theta_i) - U_i(\beta,\lambda_{-i}|\theta_i) \right| + \left|U_i(\lambda_m|\theta_{i}) - U_i(\lambda|\theta_i)\right| \right) 
  \end{align*}
  The integral of the second term is then bounded as $\lambda_m(\theta) - \lambda(\theta) \neq 0$ on a set of measure $d|\mathcal{D}|/2m$. Then, since both $\lambda_m$ and $\lambda$ are strictly increasing,
  $$
    \frac{1}{v(\theta_i)} \cdot (U_i(\lambda_m|\theta_{i}) - U_i(\lambda|\theta_i)) = \frac{1}{v(\theta_i)} \cdot \theta_i^{N-2} (\lambda_{mi}(\theta_i) - \lambda_i(\theta_i)) \\
     \leq \frac{2}{v(F(v^n_1)/2)} \cdot (\max B)^2.$$
  We remark that the bound $1/v(\theta_i)$ follows since the set of discontinuities of $\lambda_i$ are contained in $[F(v^n_1),1]$, which implies that $d_k - d/2m \geq F(v^n_1)/2$. This implies a bound on the integral,
  $$ \int_0^1 d\theta_i \cdot \frac{1}{v(\theta_i)} \cdot \left|U_i(\lambda_m|\theta_{i}) - U_i(\lambda|\theta_i)\right| \leq \frac{d|\mathcal{D}|(\max B)^2}{m\cdot v(F(v^n_1)/2)} \rightarrow 0 \textnormal{ as } m \rightarrow \infty.$$
  To see the bound on the first term, note that the probability a bid $\beta(\theta_i)$ wins the item changes whenever $\beta(\theta_i) \in \lambda_i\left(\cup_k (d_k-d/2m,d)\right)$, in which case the change in probability is at most 
  $$d_k^{N-1} - (d_k-d/2m)^{N-1} \leq (N-1) \frac{d\cdot d_k^{N-2}}{2m} \leq \frac{(N-1)d}{2m}.$$ Therefore, for both the first-price and the all-pay auction,
  \begin{align*}
    \int_0^1 d\theta_i \cdot \frac{1}{v(\theta_i)} \cdot \left| U_i(\beta,\lambda_{m(-i)}|\theta_i) - U_i(\beta,\lambda_{-i}|\theta_i) \right| & \leq \int_0^1 d\theta_i \cdot \frac{1}{v(\theta_i)} \cdot \left( \frac{v(\theta_i) \cdot (N-1)d}{2m} \right),
  \end{align*}which is at most $(N-1)d/2m$, and again $\rightarrow 0$ as $m \rightarrow \infty$.
\end{proof}

Our main result then follows by piecing together the lemmata above in a suitable fashion.

\begin{theorem}\label{thm:main-disc-result}
  For a continuous auction $A$ of format $\mathcal{A}$, let $A_n$ be a $\delta$-fine family of its discretisations. If $A$ is either 
  \begin{enumerate}
    \item an all-pay auction such that $\frac{d}{d\theta}\left(\frac{1}{v(\theta)}\right)$ has a negative upper bound on $(0,1]$, or 
    \item a first-price auction such that $\theta/v(\theta)$ is strictly decreasing with an upper bound on its derivative,
  \end{enumerate}
  then for any BCCE $\sigma^n$ of $A_n$ with symmetric support,
  $$\sum_{\lambda \in \Lambda_n} \sigma(\lambda) \cdot \sum_{i \in N} \int_0^1 d\theta_i \cdot \theta^{N-2} \cdot \frac{(\lambda_i(\theta_i)-\beta(\theta_i))^2}{2} \leq \frac{N\left[ F(4\delta(n)) +  \int_{F(4\delta(n))}^1 d\theta_i \cdot \frac{4\delta(n)}{v(\theta_i)} \right]}{\min\{C^\mathcal{A},D^\mathcal{A}\}}.$$
\end{theorem}

\begin{proof}
  Let $\Lambda_n$ be the set of bidding strategies for the continuous action induced by the symmetric bidding strategies of $A_n$. Then $\sigma^n$ has its support contained in $\Lambda_n$, which satisfies Assumption \ref{asmptn:final}. By Proposition \ref{prop:final-bounds-2}, we want to show that $\epsilon_i = 1 / \min\{C^\mathcal{A},D^\mathcal{A}\}$ and $\gamma = 0$ is a feasible dual solution for (\ref{opt:dual-disc-cont-inf}). By the transport lemma, since $\epsilon_i \geq 1/D^\mathcal{A}$, we may restrict attention to the case when $\lambda$ is weakly increasing. The dual constraints associated with such a $\lambda$ is then satisfied if it is satisfied for each $\lambda_m$ defined in the proof of Lemma \ref{lem:prelim-inc}. Each such $\lambda_m$ is strictly increasing by construction; this implies that the associated dual constraint holds if it holds for each $\lambda_{mk}$ defined following the proof of Lemma \ref{lem:smoothing}. Since $\epsilon_i \geq 1/C^\mathcal{A}$, this is assured by Theorem \ref{thm:thm-weak-uniq} and Proposition \ref{prop:wasserstein-prior}.
\end{proof}

To illustrate how Theorem \ref{thm:main-disc-result} can be applied, we compute explicitly the implied Wasserstein-$2$ distance bounds for the discretised power law distribution on $[0,1]$ with parameter $\alpha > 0$, which has a cumulative distribution function $F(v) = v^\alpha$. Towards this end, it is helpful to note down the sufficient conditions of Theorem \ref{thm:main-disc-result} in valuation space. We remark  
\begin{align*}
    \frac{d}{d\theta}\left(\frac{1}{v(\theta)}\right) & = \frac{-v'(\theta)}{v^2(\theta)} = \frac{-1}{f(v)v^2}, \textnormal{ and } \\
    \frac{d}{d\theta}\left(\frac{\theta}{v(\theta)}\right) & = \frac{v(\theta) - \theta v'(\theta)}{v^2(\theta)} = \frac{v f(v) - F(v)}{f(v)v^2}
\end{align*}
for any pair $\theta \in (0,1]$ and $v \in (0,v(1)]$ such that $\theta = F(v)$, where $F$ is the symmetric prior distribution of the continuous auction and $f$ its probability density function. Thus, for the all-pay auction, we may recover an approximation constant as long as $f(v) \leq K/v^2$ for some constant $K$, an assumption that holds for any reasonable probability density function. For the first-price auction, a sufficient condition is instead a strictly negative second derivative of $F(v)$ on $[0,v(1)]$. Thus, for any $\alpha > 0$, the power law distribution satisfies the sufficient condition for the all-pay auction, and for any $\alpha < 1$ it does so for the first-price auction as well. The Wasserstein-$2$ distance bounds in this case then follow simply from calculating the associated constants $C^{\mathcal{A}}, D^\mathcal{A}$.

\begin{proposition}\label{prop:bounds}
    Let $A_n$ be a $\delta$-fine family of discretised auctions, with a power law prior distribution with parameter $\alpha$. Then any BCCE $\sigma^n$ of $A_n$ with symmetric support satisfies 
    $$ \sum_{\lambda \in \Lambda_n} \sigma(\lambda) \cdot \sum_{i \in N} \int_0^1 d\theta_i \cdot \theta^{N-2} \cdot \frac{(\lambda_i(\theta_i)-\beta(\theta_i))^2}{2} \leq K(\alpha,N) N \begin{cases}
        4\delta(n) (1-\ln(4\delta(n)) & \alpha = 1 \\
        4\delta(n) \left( \frac{\alpha}{\alpha-1} \right) - \frac{(4\delta(n))^\alpha}{\alpha-1} & \alpha \neq 1
    \end{cases},$$
    \begin{enumerate}
        \item where if $A_n$ is an all-pay auction and $\alpha > 0$, 
       $K(\alpha,N) = \max\{ \max B, \alpha(2N-3) \}$, and
        \item if $A_n$ is a first-price auction and $\alpha \in (0,1)$, $K(\alpha, N) = \max\left\{ \max B, \alpha(2N-3), \frac{\alpha(N-1)}{1-\alpha} \right\}$.
    \end{enumerate}
\end{proposition}

We defer the associated calculations to the appendix, as they follow from a straightforward application of Theorem \ref{thm:main-disc-result}. As a consequence, for fixed $\alpha > 0$ and naive discretisations $V_n = B_n = \{0,1/n,\ldots,1-1/n\}$ of the valuation and bidding spaces, we infer a $\tilde{O}(1/n^{\min\{\alpha,1\}})$ distance bound for the all-pay auction, and an $\tilde{O}(1/n^\alpha)$ bound for the first-price auction whenever $\alpha \in (0,1)$. Moreover, implementing Algorithm \ref{alg:self-play} allows us to compute a BCCE with the associated distance guarantees.

\section{Future Directions \& Open Questions}

We provide an appropriate framework to study how close the learnable outcomes of an auction is to the unique, a priori known equilibrium of its continuous counterpart. Our results showcase both the strength and the limitations of no-regret learning in explaining convergence behaviour. Experimentally, we observe that it explains convergence in all-pay auctions generally, and for first-price auctions, in settings with concave priors. We prove this convergence for deterministic self-play via explicit dual solutions. Our primal-dual approach allows us not only to establish the convergence properties of no-regret learning with deterministic self-play, but obtain uniqueness of equilibrium results for the BCCE of the continuous auction in line with traditional economics literature on auctions, and identify \emph{how} uniqueness of equilibrium may fail by inspecting when dual solutions are no longer valid. We remark that our variational approach in our proofs uses an equilibrium known to be unique, but only through its defining differential equations. We thus expect our approach to be extendable in a larger variety of auctions and contest with unique equilibria defined as solutions to differential equations, such as asymmetric single-item auctions, Tullock contests, or Bertrand oligopoly models.

However, we currently only possess experimental results for the setting of a fixed number of competing buyers learning to bid in the auction by repeated play while using no-regret learning algorithms. The sole barrier here is the symmetry assumption\footnote{We emphasise again that the supersymmetry constraints in numerical work are without loss of generality, and are distinct from the symmetric support assumption we make. As a consequence, our experimental results are valid when buyers implement their no-regret learning algorithms independently.}. While symmetry is also a standard assumption in economics literature, the form we use allows us only to evaluating the convergence guarantees of deterministic self-play. We remark that all other aspects of Assumption \ref{asmptn:final} -- i.e. discrete containment and normalisation -- hold without loss of generality for the BCCE of discretised auctions. Lifting the symmetry assumption is thus the sole remaining challenge in providing bounds for no-regret learning with general ex-interim utility feedback. 

Our work points to a deeper question, though, and perhaps there are alternate ways to build upon our framework than addressing the abovementioned gaps in our technical knowledge. Our results point out that the no-regret property in of itself is an insufficient explanation for the convergence behaviour observed in previous experimental work \cite{BBLS17,BFHKS21,soda2023}. In particular, even our symmetric first-price auctions have off-equilibrium BCCE when the prior c.d.f. is non-concave, while convergence is observed in a much wider class of auctions. This implies that a convex description of learnable outcomes, i.e. the (convex hull of) histories of play reached by these algorithms, is a smaller set than the set of BCCE of these auctions. Furthermore, while experimental results suggest that these auctions have a unique agent-normal form Bayesian CE, these algorithms in general do not have the no-internal regret property. These point out that these learning algorithms may satisfy an \emph{alternate, stronger form of coarse correlated equilibrium}, or namely, a refinement of the no-external regret property which is nevertheless independent from internal regret. \emph{In other words, we are missing constraints in our primal formulation.} A recent result of \citet{piliouras2022evolutionary} points in this direction, where they show replicator dynamics satisfies the strongest possible form of swap regret in $2 \times 2$ normal form games. 

\subsubsection*{Acknowledgments} 
  We are thankful to Jason Hartline for advice on the construction of Example \ref{ex:necessity}. Thanks go to Stephan Lunowa, Matthias Oberlechner, Fabian Pieroth, Evgeny Shindin, and Gideon Weiss. The project was funded by the Deutsche Forschungsgemeinschaft (DFG, German Research Foundation) under Grant No. BI 1057/9.

\bibliographystyle{plain}
\bibliography{bibliography}

\appendix

\section{Notations Table}
\small
\begin{tabular}{p{4cm} p{1cm} p{9cm}}
		\hline
		General & & \\ 
		\hline
		$\mathbb{N}$ & & Set of natural numbers, in the convention without $0$ \\
        $[N]$ & & For $N \in \mathbb{N}$, the set $\{1,2,...,N\}$ \\
		$\mathbb{R}$ & & Set of real numbers \\
		$\mathbb{R}_+$ & & Set of non-negative real numbers \\
		$\Delta(S)$ & & Set of probability distributions over $S$ \\
		$S \rightarrow T$ & & Set of functions $S \rightarrow T$ \\
		$\otimes_{\ell \in \Lambda}$ & & For $s_\ell : S_\ell \rightarrow \mathbb{R}$ and $x_\ell \in S_\ell$, the map $(x_\ell)_{\ell \in \Lambda} \mapsto \prod_{\ell \in \Lambda} s_\ell(x_\ell)$ \\
		$\mathbb{I}[s = s'], \delta_{ss'}$ & & Indicator function, equals $1$ if $s = s'$ and $0$ otherwise \\
		\hline
		Games \& Auctions & &  \\
		\hline
        Normal Form Game & & \\ 
        \hline
        $\Gamma$ && A normal form game \\
        $N$ & $i,j$ & Number, set of agents / players \\
        $\Theta = \times_{i \in [N]} \Theta_i$ & $\theta$ & Set of agents' types \\
        $F^N \in \Delta(\Theta)$ & & Prior distribution over types \\ 
        $A = \times_{i \in [N]} A_i$ & $a$ & Set of action profiles / outcomes \\
        $u_i : A \times \Theta_i \rightarrow \mathbb{R}$ & & Utility function for buyer $i$ \\
        $s_i \in \Delta(A_i)^{\Theta_i}$ & & Mixed strategy for agent $i$ \\
        $\sigma \in \Delta(\times_{i\in [N]} A_i^{\Theta_i})$ & & Correlated distribution over pure strategies \\
        \hline
        Auctions & &\\
        \hline
         $A$ && An auction \\
         $N, [N]$ & $i,j$ & Number, set of buyers \\
        $\valset = \times_{i \in [N]} V$ & $v$ & Set of buyers' valuation profiles \\
        $F^N \in \Delta(\valset)$ & & Prior distribution over valuations, i.i.d. from $F \in \Delta(V)$ \\
        $\bidset = \times_{i \in [N]} B$ & $b$ & Set of bidding profiles \\
        $x_i, p_i : \bidset \rightarrow \mathbb{R}$ & & Allocation and payment rules \\
        $u_i : \bidset \times V \rightarrow \mathbb{R}$ & & Utility function for buyer $i$, $u_i(b|v) = x_i(b) \cdot v - p_i(b)$ \\
        $si, \lambda_i \in \Delta(B)^V$ & & Mixed strategy for buyer $i$, $\lambda$ often taken to be pure \\
        $\sigma \in \Delta(\times_{i\in [N]} B^V)$ & & Correlated distribution over pure bidding strategies \\
        $ \Lambda \subseteq \times_{i \in [N]}$ & & A subset of bidding strategies, with equality until Section \ref{sec:theory-cont} \\
        $\mathcal{A}$ & $FP, AP$ & Auction format\\
        $\beta, \lambda^*$ & & Equilibrium bidding strategies \\
        \hline
        Common Subscripts & & \\
        \hline
        $i,j$ & & Quantity relating to buyer / agent $i, j$ \\
        $-i, (-i)$ & & Quantity relating to buyers / agents other than $i$ \\
        $n$ & & Quantity relating to a discretised auction, indexed by $n$ \\
        \end{tabular}
        
        \noindent\begin{tabular}{p{5cm} p{9.2cm}}
        \hline
		Extended Form LP, Finite  & \\ 
		\hline
		$\sigma \in \Delta(\Lambda_n)$  & Primal variable, probability distribution over $\Lambda_n$ \\
        $c_{n\lambda}$  & Linear objective, non-negative and zero only when $\lambda = \beta$ \\
        $\epsilon(v_i, b'_i)$  & Dual variable, primal BCCE constraint \\
        $\gamma$  & Dual variable, primal constraint that $\sum_{\lambda \in \Lambda_n} \sigma(\lambda) = 1$ \\
        \hline
        LP Relaxation & \\
        \hline
        sub/superscript $[\ell]$  & Quantity relates to subset of buyers $[\ell]$ \\
        $\sigma^{[\ell]} \in \Delta(B_n^\ell)^{V_n^\ell}$  & \text{Marginal distribution of bids for buyers $[\ell]$, given valuations} \\
        \hline
        Extended Form LP, Infinite & \\
        \hline
        $\theta_i \in [0,1]$ & Type of buyer $i$, $\theta_i = F^{-i}(v_i)$ \\
        $c : \Lambda \rightarrow \mathbb{R}$ & Linear objective, often taken to be $\sum_{i \in [N]} \int_0^1 d\theta_i \cdot \alpha(\theta_i|\lambda)$ for some function $\alpha : [0,1] \times \Lambda \rightarrow \mathbb{R}$, equals zero only when $\lambda_i = \beta$ for every buyer $i$ \\
        $\sigma \in \Delta(\Lambda)$ & A Borel probability measure over support $\Lambda$ of a BCCE \\
        $\epsilon(\theta_i,b'_i)$ & BCCE constraint for buyer $i$, type $\theta_i$ and bid $b'_i$. We commonly enforce $\epsilon_i(\theta_i,\beta(\theta_i)) \equiv \epsilon(\theta_i)$, and drop all other BCCE constraints. \\
        $\gamma$ & The primal constraint $\int_\Lambda d\sigma(\lambda) = 1$ \\
        $U_i : \Lambda \times [0,1] \rightarrow \mathbb{R}$ & Expected utility shorthard, $U_i(\lambda|\theta_i) = \int_{[0,1]^{N-1}} d\theta_{-i} \cdot u_i(\lambda(\theta)|\theta_i)$. May for convenience replace buyer $i$'s bidding strategy for a fixed constant bid. \\
\end{tabular}
\normalsize
\section{Further Proofs}

\subsection{Section \ref{sec:uniqueness-strict}}

\begin{proposition*}[Proposition \ref{prop:weak-duality-strict}, all-pay auction]
    Let $\sigma$ be a Borel measure on $\Lambda$, and $\epsilon : [0,1] \rightarrow \mathbb{R}_+$ a continuous function. Then for an all-pay auction, 
  $$ \int_{\profset \times [0,1]} d(\sigma \times \theta_i) \cdot \epsilon(\theta_i) \cdot \left| U_i(\beta,\lambda_{-i} | \theta_i) - U_i(\lambda|\theta_i) \right| < \infty.$$
  As a consequence, (\ref{cond:fubini}) holds, hence weak duality too for the primal-dual pair (\ref{opt:inf-AP-strict})-(\ref{opt:dual-inf-AP-strict}).
\end{proposition*}

\begin{proof}
    For the all-pay auction, we have 
    \begin{align*}
        \left| U_i(\beta,\lambda_{-i}| \theta_i) - U_i(\lambda|\theta_i) \right| & \leq \left| U_i(\beta,\lambda_{-i} | \theta_i) \right| + \left| U_i(\lambda|\theta_i) \right| \\
        & \leq \left| P_\beta(\lambda,\theta_i)^{N-1} \cdot v(\theta_i)\right| + \left|\beta(\theta_i) \right| + \left| \theta_i^{N-1} \cdot v(\theta_i) \right| + \left| \lambda(\theta_i) \right| \\
        & \leq v(\theta_i) + \beta(1) + v(\theta_i) + M.
    \end{align*}
    The finiteness of the iterated integral 
    $\int_\profset d\sigma(\lambda) \cdot \int_0^1 d\theta_i \cdot \epsilon(\theta_i) \cdot \left| U_i(\beta,\lambda_{-i} | \theta_i) - U_i(\lambda|\theta_i) \right|$
    then follows since $v$ is bounded, hence integrable.
\end{proof}

\begin{theorem*}[Theorem \ref{thm:thm-ap-uniq}]
    Suppose for an all-pay auction that $\Lambda$ is a uniformly bounded set of symmetric, strictly increasing, differentiable \& normalised bidding strategies containing the equilibrium bidding strategies $\beta$. Then $\epsilon(\theta_i) = \theta^{N-2}\cdot (\beta(1)-\beta(\theta_i))$ is a solution to (\ref{opt:dual-inf-AP-strict}) of value $0$. As a consequence, if a Borel measure $\sigma$ on $\Lambda$ is a continuous BCCE, it places probability $1$ on the equilibrium bidding strategies~$\beta$.
\end{theorem*}

\begin{proof}

Again, given an appropriate dual solution, Corollary \ref{cor:weak-dual-consequence} ensures that any open set not containing the equilibrium bidding strategies is assigned zero probability by $\sigma$. Proceeding as in the proof of Theorem \ref{thm:thm-fp-uniq}, for $0 \leq \underline{\theta} < \overline{\theta} \leq 1$, denote by $\gamma(\lambda|\underline{\theta},\overline{\theta})$,
$$ \gamma(\lambda|\underline{\theta},\overline{\theta}) = \int_{\underline{\theta}}^{\overline{\theta}} d\theta_1 \cdot \left( \alpha(\theta_1|\lambda) - \epsilon(\theta_1) \cdot  \left( U_1(\beta,\lambda_{-1} | \theta_1) - U_1(\lambda|\theta_1) \right) \right), $$
and again divide the interval $[0,1]$ into subintervals, $[0,1] = \cup_{k} [\underline{\theta}_k,\overline{\theta}_k]$, such that for any $k$:
\begin{enumerate}
  \item $\overline{\theta}_k < 1$, or
  \item $\overline{\theta}_k = 1$ and $\beta(\theta) \geq \lambda(\theta)$ for any $\theta \in [\underline{\theta}_k,1]$, or
  \item $\overline{\theta}_k = 1$ and $\beta(\theta) \leq \lambda(\theta)$ for any $\theta \in [\underline{\theta}_k,1]$,
\end{enumerate}
where again for each subinterval $\beta(\underline{\theta}_k) = \lambda(\underline{\theta}_k)$, and also $\beta(\overline{\theta}_k) = \lambda(\overline{\theta}_k)$ so long as $\overline{\theta}_k \neq 1$.

To establish that $\beta - \lambda$ is a ascent direction, proceed by case analysis. For a fixed interval $k$, first suppose that $\overline{\theta}_k < 1$. Taking the first-order variation of $\gamma(\lambda|\underline{\theta}_k,\overline{\theta}_k)$, we get
\begin{align*}
    & \delta\gamma(\lambda | \underline{\theta}_k , \overline{\theta}_k ) = \int_{\underline{\theta}_k}^{\overline{\theta}_k} d\theta \cdot \Bigg[\theta^{N-2} \cdot (\lambda(\theta)-\beta(\theta)) \cdot \delta\lambda(\theta) \ldots \\ &  + \epsilon(\theta) \cdot \frac{(N-1)\lambda^{-1}\beta(\theta)^{N-2}}{\lambda' \lambda^{-1} \beta(\theta)} \cdot v(\theta) \cdot \delta\lambda(\lambda^{-1}\beta(\theta)) - \epsilon(\theta) \cdot \delta\lambda(\theta) \Bigg].
\end{align*}
We want all terms of the integrand to be multiplied by $\delta\lambda(\theta)$, the pointwise variation in $\lambda$. Therefore, for the term 
$ \int_{\underline{\theta}_k}^{\overline{\theta}_k} d\theta \cdot \epsilon(\theta) \cdot \frac{(N-1)\lambda^{-1}\beta(\theta)^{N-2}}{\lambda' \lambda^{-1} \beta(\theta)} \cdot v(\theta) \cdot \delta\lambda(\lambda^{-1}\beta(\theta)) $
we again make a change of variables,
\begin{align*}
    \theta' & = \lambda^{-1} \beta (\theta), \\
    d\theta' & = \frac{\beta'(\theta)}{\lambda'\lambda^{-1}\beta(\theta)} d\theta.
\end{align*}
By the assumption, we have $\underline{\theta}_k = \lambda^{-1}\beta(\underline{\theta}_k)$ and $\overline{\theta}_k = \lambda^{-1}\beta(\overline{\theta}_k)$. Moreover, for the all-pay auction, we have $\beta'(\theta) = (N-1) \cdot \theta^{N-2} \cdot v(\theta)$. We thus have, after making the change of variables and relabelling $\theta' \leftarrow \theta$,
\begin{align*}
    \delta\gamma(\lambda | \underline{\theta}_k , \overline{\theta}_k ) = \int_{\underline{\theta}_k}^{\overline{\theta}_k} d\theta \cdot \delta\lambda(\theta) \cdot \left[\theta^{N-2}\cdot(\lambda(\theta)-\beta(\theta)) + \epsilon\beta^{-1}\lambda(\theta) \cdot \frac{\theta^{N-2}}{\beta^{-1}\lambda(\theta)^{N-2}} - \epsilon(\theta) \right].
\end{align*}
This time, we shall immediately fix $\epsilon(\theta) = \theta^{N-2} \cdot (\beta(1)-\beta(\theta))$. Then we have 
\begin{align*}
  \epsilon\beta^{-1}\lambda(\theta)\cdot \frac{\theta^{N-2}}{\beta^{-1}\lambda(\theta)^{N-2}} - \epsilon(\theta) & = \beta^{-1}\lambda(\theta)^{N-2} (\beta(1)-\lambda(\theta)) \cdot \frac{\theta^{N-2}}{\beta^{-1}\lambda(\theta)^{N-2}} - \theta^{N-2} (\beta(1)-\beta(\theta)) \\
  & = \theta^{N-2} \cdot (\beta(\theta)-\lambda(\theta)).
\end{align*}
As a consequence, the variation in direction $\beta-\lambda$ equals 
\begin{align*}
  \delta\gamma(\lambda | \underline{\theta}_k , \overline{\theta}_k ) & = \int_{\underline{\theta}_k}^{\overline{\theta}_k} d\theta \cdot (\beta(\theta)-\lambda(\theta)) \cdot \left[\theta^{N-2}\cdot(\lambda(\theta)-\beta(\theta)) + \epsilon\beta^{-1}\lambda(\theta) \cdot \frac{\theta^{N-2}}{\beta^{-1}\lambda(\theta)^{N-2}} - \epsilon(\theta) \right] \\
  & = \int_{\underline{\theta}_k}^{\overline{\theta}_k} d\theta \cdot (\beta(\theta)-\lambda(\theta)) \cdot 0 = 0.
\end{align*}

Now suppose instead that $\overline{\theta}_k = 1$, and moreover, $\beta(1) \neq \lambda(1)$. As before; first consider the subcase where $\lambda(1) > \beta(1)$ (and hence, by assumption, $\lambda(\theta) \geq \beta(\theta)$ for any $\theta > \underline{\theta}_k$). 
In this case, we have $\lambda^{-1}\beta(1) < 1$, and 
\begin{align*}
    \delta\gamma(\lambda | \underline{\theta}_k , 1 ) & = \int_{\underline{\theta}_k}^{1} d\theta \cdot \Bigg[\frac{\delta \alpha(\theta|\lambda)}{\delta\lambda} \cdot \delta\lambda(\theta) \ldots \\ & + \epsilon(\theta) \cdot \frac{(N-1)\lambda^{-1}\beta(\theta)^{N-2}}{\lambda' \lambda^{-1} \beta(\theta)} \cdot v(\theta) \cdot \delta\lambda(\lambda^{-1}\beta(\theta)) - \epsilon(\theta) \cdot \delta\lambda(\theta) \Bigg] \\
    & = \int_{\lambda^{-1}\beta(1)}^{1} d\theta \cdot \delta\lambda(\theta) \cdot \left[\theta^{N-2} \cdot (\beta(1) - \beta(\theta))  - \epsilon(\theta) \right] \ldots \\
    & + \int_{\underline{\theta}_k}^{\lambda^{-1}\beta(1)} d\theta \cdot \delta\lambda(\theta) \cdot \left[\theta^{N-2}\cdot(\lambda(\theta)-\beta(\theta)) + \epsilon\beta^{-1}\lambda(\theta) \cdot \frac{\theta^{N-2}}{\beta^{-1}\lambda(\theta)^{N-2}} - \epsilon(\theta) \cdot \theta^{N-2} \right].
\end{align*}
Recall now that under our assumptions, $\beta(\theta) - \lambda(\theta) \leq 0$ for any $\theta \in (\underline{\theta}_k, \overline{\theta}_k)$. When we consider the variation in direction $\beta-\lambda$, the second term is non-negative by the previous case considered. Meanwhile, the first term also equals $0$ by the definition of $\epsilon$.

Finally suppose that $\lambda(1) < \beta(1)$. In this case, for $\theta > \beta^{-1}\lambda(1)$, a buyer wins with probability $1$ after deviating to $\beta(\theta)$, and we have
\begin{align*}
    \gamma(\lambda|\underline{\theta}_k, 1)
    & = \int_{\underline{\theta}_k}^1 d\theta \cdot \left( \alpha(\theta|\lambda) + \epsilon(\theta) \cdot (\theta^{N-1} v(\theta) - \lambda(\theta)) \right) \\ & - \int_{\underline{\theta}_k}^{\beta^{-1}\lambda(1)} d\theta \cdot \epsilon(\theta) \cdot (\lambda^{-1}\beta(\theta)^{N-1} v(\theta) - \beta(\theta) ) \\ & - \int_{\beta^{-1}\lambda(1)}^1 d\theta \cdot \epsilon(\theta) \cdot (v(\theta) - \beta(\theta))
\end{align*}
Again, considering its first-order variation, we get
\begin{align*}
    \delta\gamma(\lambda | \underline{\theta}_k , 1 ) & = \int_{\underline{\theta}_k}^1 d\theta \cdot \left(\theta^{N-2} \cdot (\lambda(\theta)-\beta(\theta)) \cdot \delta\lambda(\theta) - \epsilon(\theta) \cdot  \delta\lambda(\theta) \right) \\
    & + \int_{\underline{\theta}_k}^{\beta^{-1}\lambda(1)} d\theta \cdot \epsilon(\theta) \cdot \frac{(N-1) \lambda^{-1}\beta(\theta)^{N-2}}{\lambda' \lambda^{-1}\beta(\theta)} v(\theta) \cdot  \delta\lambda( \lambda^{-1} \beta(\theta) ) \\
    & - \frac{1}{\beta' \beta^{-1} \lambda(1)} \cdot \epsilon\beta^{-1}\lambda(1) \cdot (v\beta^{-1}\lambda(1) - \lambda(1)) \cdot \delta\lambda(1) \\
    & + \frac{1}{\beta' \beta^{-1} \lambda(1)} \cdot  \epsilon\beta^{-1}\lambda(1) \cdot (v\beta^{-1}\lambda(1) - \lambda(1)) \cdot \delta\lambda(1).
\end{align*}
The latter two terms of course cancel out, and after the change of variables as before, we get
\begin{align*}
    & \delta\gamma(\lambda | \underline{\theta}_k , 1 ) = \int_{\underline{\theta}_k}^1 d\theta \cdot \delta\lambda(\theta) \cdot \left(\theta^{N-2} \cdot (\lambda(\theta)-\beta(\theta))  - \epsilon(\theta) + \epsilon\beta^{-1}\lambda(\theta) \cdot \frac{\theta^{N-2}}{\beta^{-1}\lambda(\theta)^{N-2}} \right) .
\end{align*}
We thus have a reduction from this case to the first case we considered.

\end{proof}

\subsection{Section \ref{sec:weak-uniqueness}}

\begin{proposition*}[Proposition \ref{prop:wasserstein-prior}, all-pay auction]
  For any bidding strategies $\lambda$ satisfying Assumption \ref{asmptn:sup-weak} in an all-pay auction, whenever $1/(\beta^{-1})^{N-2}v\beta^{-1}$ as a lower bound \underbar{on the magnitude} of its derivative, $\int_0^1 d\theta \cdot \alpha(\theta|\lambda) \geq \int_0^1 d\theta \cdot \theta^{N-2} \cdot \frac{(\lambda(\theta)-\beta(\theta))^2}{2} \cdot C^{AP}$, where
    $$C^{AP} = \inf\left\{ -\frac{d}{d\mu} \left( \frac{1}{\beta^{-1}(\mu)^{N-2} \cdot v\beta^{-1}(\mu)} \right) \ \Bigg| \ \mu \in (0,\beta(1)]  \right\} \cup \left\{ \frac{1}{v(1) \max B} \right\} < 0.$$
    This holds if $d/d\theta [1/v(\theta)]$ has a non-zero limit at $0$.
\end{proposition*}

\begin{proof}
   For the all-pay auction, whenever $\lambda(1) > \beta(1)$, 
  \begin{align*}
    & \int_0^1 d\theta \cdot \frac{\delta\alpha(\theta|\lambda)}{\delta\lambda} \cdot \delta\lambda(\theta) \\ 
   =  & \int_{\lambda^{-1}\beta(1)}^{1} d\theta \cdot \left( \frac{\delta\lambda(\theta)}{v(\theta)} \right) + \int_{0}^{\lambda^{-1}\beta(1)} d\theta \cdot \delta\lambda(\theta) \cdot \theta^{N-2}\cdot \left( \frac{1}{\theta^{N-2}v(\theta)} - \frac{1}{\beta^{-1}\lambda(\theta)^{N-2}v\beta^{-1}\lambda(\theta)} \right).
  \end{align*}
  When we set $\delta\lambda = \beta - \lambda$, the resulting integrand should be pointwise lesser than the integrand of
  \begin{align*}
    & \int_0^1 d\theta \cdot \delta\lambda(\theta) \cdot \theta^{N-2} \cdot \frac{\delta}{\delta\lambda}\left(\frac{(\lambda(\theta)-\beta(\theta))^2}{2}\right) \cdot C^{AP}.
  \end{align*} 
  Now, if $\theta > \lambda^{-1}\beta(\theta)$, then this necessitates 
  \begin{align*}\frac{\beta(\theta) - \lambda(\theta)}{v(\theta)} \leq -(\beta(\theta) - \lambda(\theta))^2 \cdot \theta^{N-2} \cdot C^{AP} \Rightarrow C^{AP} 
    & \leq -\frac{1}{v(\theta) (\beta(\theta) - \lambda(\theta)) \theta^{N-2}} \\
    & \leq \frac{1}{v(\theta)\lambda(\theta)}.
  \end{align*}
  This is guaranteed whenever $C^{AP} \leq 1/(v(1) \max B)$. If instead $\theta < \lambda^{-1}\beta(1)$, then we require 
  $$ (\beta(\theta) - \lambda(\theta)) \cdot \left( \frac{1}{\theta^{N-2}v(\theta)} - \frac{1}{\beta^{-1}\lambda(\theta)^{N-2}v\beta^{-1}\lambda(\theta)} \right) \leq -(\beta(\theta) - \lambda(\theta))^2 \cdot C^{AP}.$$
  This is satisfied whenever 
  $$ C^{AP} \leq -\frac{1}{\beta(\theta) - \lambda(\theta)} \cdot \left(\frac{1}{\theta^{N-2}v(\theta)} - \frac{1}{\beta^{-1}\lambda(\theta)^{N-2}v\beta^{-1}\lambda(\theta)} \right).$$
  However, by the mean value theorem, the right hand side equals $-d/d\mu \ (1/\beta^{-1}(\mu)^{N-2} v\beta^{-1}(\mu))$, evaluated at some point $\mu$ between $\lambda(\theta)$ and $\beta(\theta)$. Under our sufficient condition, as $\mu \rightarrow 0$, this either converges to some negative constant or diverges to $-\infty$, which implies that choosing 
  $$ 0 < C^{AP} \leq \inf\left\{ -\frac{d}{d\mu} \left( \frac{1}{\beta^{-1}(\mu)^{N-2} \cdot v\beta^{-1}(\mu)} \right) \ \Bigg| \ \mu \in (0,\beta(1)]  \right\}$$ is sufficient for our purposes.
\end{proof}

\subsection{Section \ref{sec:bounds-theory}}

The following lemma is used in the proof of the Transport Lemma (Lemma \ref{lem:transport}) for a suitable permutation of bids for the first-price auction.

\begin{lemma}\label{lem:rearr}
  For some integer $z \in \mathbb{N}$, let $\sigma : [z] \rightarrow [z]$ be a non-identity permutation. Then there exists $i < j$ in $[z]$ such that $\sigma(i) \geq j$ and $\sigma(j) \leq j$.
\end{lemma}

\begin{proof}
  Let $j' = \max \{j'' \in [z] | \sigma(j'') \neq j'' \}$. As for any $j'' > j', \sigma(j'') = j''$, we have $\sigma(j') < j'$. We let $i = \sigma^{-1}(j')$. Now, consider the image $\sigma([i+1,z])$; since $\sigma$ is a bijection, the image contains $z-i$ elements. Moreover, it cannot be that $\sigma([i+1,z]) = [i+1,z]$, as that would imply that $\sigma^{-1}(\{j'\})$ is not a singleton, in contradiction to $\sigma$ being a permutation. Therefore, there exists $j \in [i+1,z]$ such that $\sigma(j) \notin [i+1,z]$, i.e. $i < j$ and $\sigma(j) \leq i$. By maximality of $j'$, $\sigma(i) \geq j$. 
\end{proof}

\begin{lemma*}[Transport Lemma / Lemma \ref{lem:transport}, all-pay auction]
  Suppose that $\lambda \in \Lambda_n$ satisfies Assumption \ref{asmptn:final}, then for the all-pay auction,
    \begin{align*}
      & \int_0^1 d\theta_i \cdot \frac{1}{v(\theta_i)} \left( U_i(\beta,\lambda_{-i}|\theta_i) - U_i(\beta,\mon \lambda|\theta_i) - U_i(\lambda_{-i}|\theta_i) + U_i(\mon \lambda |\theta_i) \right) \\ 
      \geq & \int_0^1 d \theta_i \cdot \theta^{N-2} \cdot \left( \frac{(\lambda_i(\theta_i)-\beta(\theta_i))^2}{2} -  \frac{(\mon \lambda_i(\theta_i)-\beta(\theta_i))^2}{2} \right)  \cdot D^{AP}, \\
      & \textnormal{ where } D^{AP} = \inf_{\theta \in (0,1]} \frac{ v'(\theta)}{(2N-3)v(\theta)^3}.
    \end{align*} 
    As a consequence, a sufficient condition for $D^{AP} > 0$ to hold is for $1/v(\theta)$ to have an upper bound on its derivative on $(0,1]$.
\end{lemma*}

\begin{proof}
    For the all-pay auction, the left hand side of (\ref{eq:transport-intermediate}) is the integral, after some cancellations,
  \begin{equation}\label{eq:all-pay-compare-1}
    \int_0^{1/\ell} d\theta \cdot \left(\frac{1}{v(\phi_k + \theta)} - \frac{1}{v(\phi_{k'} + \theta)}\right)\left( \lambda_m(\phi_k + \theta) - \lambda_m(\phi_{k'}+\theta) \right).
  \end{equation}
  We want to compare the integrands of (\ref{eq:Wasserstein-compare}) and (\ref{eq:all-pay-compare-1}) pointwise. Write $\theta^\pm$ for the higher and the lower types, and divide both integrands by $\lambda_m(\phi_k+\theta) - \lambda_m(\phi_{k'}+\theta)$, which is positive. It is thus sufficient to show that 
  $$ D^{AP} \cdot (\beta(\theta^+)\theta^{+ \ N-2} - \beta(\theta^-)\theta^{- \ N-2}) \leq \frac{1}{v(\theta^-)} - \frac{1}{v(\theta^+)} \textnormal{ for any } \theta^+ \geq \theta^-.$$
  Equality holds whenever $\theta^+ = \theta^-$, so if the derivative of the right-hand side with respect to $\theta^+$ is greater than that of the left-hand side, we are done. This gives us the bound,
  $$ D^{AP} \leq \inf_{\theta \in (0,1]} \frac{-\frac{d}{d\theta} \frac{1}{v(\theta)}}{\frac{d}{d\theta} \left( \beta(\theta) \theta^{N-2} \right)}.$$
  To eliminate the $\beta$ dependence, we note that $\beta'(\theta) = (N-1)\theta^{N-2}v(\theta) \leq (N-1)v(\theta)$ and $\beta(\theta) \leq v(\theta)$ for the all-pay auction, which implies that 
  $$\frac{d}{d\theta}(\beta(\theta)\theta^{N-2}) \leq \theta^{N-3}(\theta \beta'(\theta) + (N-2) \beta(\theta) ) \leq (2N-3) v(\theta).$$
  As a consequence a sufficient choice of $D^{AP}$ satisfies
  $$D^{AP} \leq \inf_{\theta \in (0,1]} \frac{ -\frac{d}{d\theta} \frac{1}{v(\theta)}}{(2N-3)v(\theta)} = \inf_{\theta \in (0,1]} \frac{ v'(\theta)}{(2N-3)v(\theta)^3}.$$
\end{proof}

\begin{proposition*}[Proposition \ref{prop:bounds}]
    Let $A_n$ be a $\delta$-fine family of discretised auctions, with a power law prior distribution with parameter $\alpha$. Then any BCCE $\sigma^n$ of $A_n$ with symmetric support satisfies 
    $$ \sum_{\lambda \in \Lambda_n} \sigma(\lambda) \cdot \sum_{i \in N} \int_0^1 d\theta_i \cdot \theta^{N-2} \cdot \frac{(\lambda_i(\theta_i)-\beta(\theta_i))^2}{2} \leq K(\alpha,N) N \begin{cases}
        4\delta(n) (1-\ln(4\delta(n)) & \alpha = 1 \\
        4\delta(n) \left( \frac{\alpha}{\alpha-1} \right) - \frac{(4\delta(n))^\alpha}{\alpha-1} & \alpha \neq 1
    \end{cases},$$
    \begin{enumerate}
        \item where if $A_n$ is an all-pay auction and $\alpha > 0$, 
       $K(\alpha,N) = \max\{ \max B, \alpha(2N-3) \}$, and
        \item if $A_n$ is a first-price auction and $\alpha \in (0,1)$, $K(\alpha, N) = \max\left\{ \max B, \alpha(2N-3), \frac{\alpha(N-1)}{1-\alpha} \right\}$.
    \end{enumerate}
\end{proposition*}

\begin{proof}
The form of the $\delta(n)$ dependence follows from Proposition \ref{prop:disc-to-approx-cont} and by solving, for each $\alpha > 0$, for the expression $$\int_{F(4\delta(n))}^1 d\theta_i \cdot \frac{4\delta(n)}{v(\theta_i)} = \int_{(4\delta(n))^\alpha}^1 d\theta_i \cdot 4\delta(n)\theta^{-1/\alpha} = \begin{cases}
        -4\delta(n) \ln(4\delta(n)) & \alpha = 1 \\
        \frac{\alpha}{\alpha-1} \cdot \left( 4\delta(n) - (4\delta(n))^\alpha \right) & \alpha \neq 1
    \end{cases}.$$
Then, following Theorem \ref{thm:main-disc-result}, Proposition \ref{prop:wasserstein-prior} and Lemma \ref{lem:transport}, we only need to compute sufficiently small $C^{AP}, D^{AP}, C^{FP}, D^{FP}$. We need to fix both $C^{AP}, C^{FP} \leq 1/ \max B$. For $C^{AP}$, we also require 
\begin{align*}
C^{AP} & \leq \inf_{\mu \in (0,\beta(1)]} \frac{d}{d\mu} \left( \frac{-1}{\beta^{-1}(\mu)^{N-2} \cdot v\beta^{-1}(\mu)} \right) \\
& = \inf_{\mu \in (0,\beta(1)]} \frac{1}{\beta'\beta^{-1}(\mu)} \cdot \left( \frac{(N-2) v\beta^{-1}(\mu) + v'\beta^{-1}(\mu) \beta^{-1}(\mu)}{\beta^{-1}(\mu)^{N-1} \cdot v\beta^{-1}(\mu)} \right)\end{align*}
Now, recall that the all-pay auction satisfies $\beta'(\theta) = (N-1) \theta^{N-2} v(\theta)$. This implies that for any pair $\theta = \beta^{-1}(\mu) \in (0,1]$, we have
\begin{align*}
    \frac{v'(\theta)}{(N-1)v(\theta)^2} & \leq \frac{v'(\theta)}{(N-1)\theta^{N-2}v(\theta)^2} = \frac{v'(\theta)}{\beta'(\theta) \cdot v(\theta)} \\
    & \leq \frac{v'(\theta)}{\beta'(\theta) \theta^{N-2} v(\theta)} \leq \frac{(N-2) v(\theta) + v'(\theta) \theta}{ \beta'(\theta) \theta^{N-1} v(\theta)}.
\end{align*}
It is thus sufficient to also pick 
$$C^{AP} \leq \inf_{\theta \in (0,1]} \frac{v'(\theta)}{(N-1) v(\theta)^2} = \inf_{\theta \in (0,1]} \frac{\theta^{-1-1/\alpha}}{\alpha(N-1)} = \frac{1}{\alpha(N-1)}.$$
Meanwhile, the sufficient bound on $D^{AP}$ is simply
\begin{align*}
    D^{AP} & = \inf_{\theta \in (0,1]} \frac{v'(\theta)}{(2N-3)v(\theta)^3} = \inf_{\theta \in (0,1]} \frac{\theta^{-1-2/\alpha}}{\alpha (2N-3)} = \frac{1}{\alpha(2N-3)}.
\end{align*}
It remains to calculate the relevant quantities for the first price auction. The proof of Proposition \ref{prop:wasserstein-prior} suggests that we may pick 
$$ C^{FP} \leq \inf_{\theta \in (0,1]} \frac{1}{N-1} \cdot \frac{d}{d\theta} \left( \frac{-\theta}{v(\theta)} \right) = \frac{1-\alpha}{\alpha(N-1)}.$$
Meanwhile, by Lemma \ref{lem:transport},
$$D^{FP} = \inf_{\theta \in (0,1]} \frac{v'(\theta)\theta^2}{(2N-3)v(\theta)^3} = \inf_{\theta \in (0,1]} \frac{\theta^{1-2/\alpha}}{\alpha(2N-3)} = \frac{1}{\alpha(2N-3)}.$$
\end{proof}

\end{document}